\documentclass[11pt]{amsart}
\usepackage{amsmath,amssymb,amsthm,amsfonts}
\usepackage[numbers]{natbib}
\usepackage{graphicx}
\usepackage{mhequ}
\usepackage{enumerate}
\usepackage{color}

\usepackage{algorithm}
\usepackage{algpseudocode}
\usepackage{comment}

\renewcommand{\bf}[1]{\mathbf{#1}}
\newcommand{\be}{\begin{equs}}
\newcommand{\ee}{\end{equs}}

\def \P{\mathbb{P}}
\def \E{\mathbb{E}}

\DeclareMathOperator{\Uniform}{Uniform}

\DeclareMathOperator{\var}{var}
\DeclareMathOperator{\Cost}{Cost}

\newcommand{\bb}[1]{\mathbb{#1}}
\newcommand{\mc}[1]{\mathcal{#1}}
\DeclareMathOperator{\TV}{TV}

\DeclareMathOperator{\N}{N}

\DeclareMathOperator{\Span}{Span}

\def \aI{a_{1}}
\def \aII{a_{3}}
\def \aIII{a_{4}}
\def \aIV{a_{2}}
\def \AI{A_{1}}
\def \AII{A_{3}}
\def \AIII{A_{4}}
\def \AIV{A_{2}}

\newtheorem{theorem}{Theorem}[section]
\newtheorem{lemma}[theorem]{Lemma}
\newtheorem{corollary}[theorem]{Corollary}
\newtheorem{prop}[theorem]{Proposition}
\newtheorem{assumptions}[theorem]{Assumptions}
\newtheorem{assumption}[theorem]{Assumption}

\newtheorem{definition}{Definition}

\theoremstyle{plain}
\newtheorem{thm}{Theorem}
\newtheorem*{thm-non}{Theorem}

\theoremstyle{definition}
\newtheorem{example}[theorem]{Example}
\newtheorem{defn}[theorem]{Definition}
\newtheorem{remark}[theorem]{Remark}

\begin{document}
\allowdisplaybreaks

\title[No Free Lunch for Approximate MCMC]
{No Free Lunch for Approximate MCMC}

\author{James Johndrow$^{\ddag}$}
\thanks{$^{\ddag}$johndrow@wharton.upenn.edu,
   Department of Statistics
    University of Pennsylvania, 3730 Walnut St, Philadelphia
    PA 19104, USA}

\author{Natesh S. Pillai$^{\flat}$}
\thanks{$^{\flat}$pillai@fas.harvard.edu,
   Department of Statistics
    Harvard University, 1 Oxford Street, Cambridge
    MA 02138, USA}

\author{Aaron Smith$^{\sharp}$}
\thanks{$^{\sharp}$smith.aaron.matthew@gmail.com,
   Department of Mathematics and Statistics
University of Ottawa, 585 King Edward Drive, Ottawa
ON K1N 7N5, Canada}

\maketitle






\begin{abstract}
It is widely known that the performance of Markov chain Monte Carlo (MCMC) can degrade quickly when targeting computationally expensive posterior distributions, such as when the sample size is large. This has motivated the search for MCMC variants that scale well to large datasets. One popular general approach has been to look at only a subsample of the data at every step. In this note, we point out that well-known MCMC convergence results often imply that these ``subsampling'' MCMC algorithms cannot greatly improve performance. We apply these abstract results to realistic statistical problems and proposed algorithms, and also discuss some design principles suggested by the results. Finally, we develop estimates for the singular values of random matrices bounds that may be of independent interest.
\end{abstract}

\section{Introduction}

Although they are ubiquitous in small-scale Bayesian inference, MCMC algorithms often perform poorly in ``big data'' -- i.e. large sample size -- problems. The simplest reason for this poor scaling is that most popular MCMC algorithms require a computation involving every data point at every time step. This suggests the heuristic that the per-step computational cost of MCMC scales linearly in the size $n$ of the data set.\footnote{In practice, the scaling behavior is more complicated and depends on details of hardware and memory management. We will largely ignore these hardware-specific issues in the current paper.}

There has been a great deal of recent work on methods to avoid this linear scaling by looking only at a few carefully-chosen control variates and a (possibly random) subsample of the data at every step, again using the heuristic that the per-step computational cost of such chains scales linearly in the number $r$ of data points used at each step. Much of this work was inspired by  \cite{korattikara2014austerity}. This note is an attempt to describe what improvements are possible using methods similar to those currently being proposed. Our basic conclusion supports that of the recent paper \cite{nagapetyan2017true}, which studied natural upper bounds on the convergence rate in the important special case of the stochastic gradient Langevin algorithm: in many realistic situations, the \textit{total} cost of an MCMC algorithm can't be improved by decreasing the size $r$ of the subsample. Our results, like those of \cite{nagapetyan2017true}, apply primarily to the uses of MCMC for posterior sampling - it is well-known that subsampling methods can lead to speedups for stochastic optimization methods.

Our main abstract result, Theorem \ref{thm:SpecGapBound}, gives upper bounds on accuracy of Monte Carlo estimators obtained from subsampling MCMC algorithms. Informally, we think of these results as giving conditions under which that \textit{at least} one of the following things must happen:

\begin{enumerate}
 \item The subsampling Markov chain mixes slowly enough that there is a very small overall computational cost advantage (though there may be a large \textit{per-step} advantage); or
 \item The stationary measure of the subsampling Markov chain is ``far'' from the posterior distribution conditioned on all data; or
 \item The algorithm takes advantage of control variates that are very nearly sufficient statistics; i.e. one can obtain a very accurate approximation to the posterior by conditioning only on the control variates, without seeing the rest of the data.
\end{enumerate}

The first two cases amount to failure of a subsampling MCMC algorithm to meet its main goal of speeding up the computation of accurate posterior integrals. The third case requires very good low-dimensional control variates, which presents two  problems that are more subtle. The first of these problems is practical: good control variates are often computationally expensive to construct, and this must be factored into the total computational cost of an algorithm. The second problem is related to motivation for using MCMC at all: once one has found control variates good enough to be in case 3, there is very little benefit to looking at the rest of the data, so one might as well run an algorithm using only the control variates and avoid the subsampling entirely.\footnote{There do exist models for which good control variates are easy to compute, but posterior computations based on these control variates are much harder than posterior computations based on the original data. This class of models gives an important ``way out'' of our trilemma - see Section \ref{SecLitRev}. In practice, this situation seems quite uncommon.} Thus, in all three cases, we find that there is a strong reason to avoid subsampling MCMC. Our results are quite general and explicitly apply to subsampling algorithms that use data augmentation or for which the invariant measure is not exactly the posterior distribution conditioned on the full dataset.

Although we can't always tell which of these three cases occur, Theorems \ref{ThmExampleTorpid} and \ref{ThmActualLogisticRegression} say that algorithms for sampling from the posterior of logistic regression and some other generalized linear models must end up in the first case. Furthermore, we can check that they end up in this first case even when using certain natural control variates. This confirms the ability of our general bounds to show that speedup does not occur in specific statistical applications.

Note that, on first reading our trilemma, one might have the following concern: perhaps most real-life control variates are ``nearly sufficient statistics" in the sense of our third case, but that the information in these control variates cannot easily be used outside a subsampling algorithm. If that were the case, then we would want to use subsampling algorithms anyway. Theorems \ref{ThmExampleTorpid} and \ref{ThmActualLogisticRegression} rule out this possibility, showing that in fact many commonly-used control variates do not carry enough information to speed up subsampling.

\subsection{Very Closely Related Work, Very Different Papers} \label{SecVeryClose}

As noted earlier, papers such as \cite{nagapetyan2017true} have already given heuristic arguments that it may not be possible to speed up many algorithms using subsampling. Although our basic conclusions are similar to \cite{nagapetyan2017true}, we emphasize that the technical details are very different. In contrast to \textit{e.g.} \cite{nagapetyan2017true}, our results apply to many chains (not just one simple family), allow the use of control variates, and give genuine upper bounds on algorithm efficiency (rather than merely upper bounds on natural lower bounds). This last difference is especially important: showing that the convergence time is \textit{genuinely large} is much harder than merely showing that a natural \textit{upper bound} on the convergence time is large.

As an unfortunate consequence of this added strength and generality, some of the assumptions of Theorem \ref{thm:SpecGapBound} are likely to be unfamiliar and strange-looking to most readers. To give a concrete example of a strange-looking definition, Section \ref{SecUsingData} gives a careful and generic notion of what it means to ``use" a data point in the presence of pre-computed control variates. Even the familiar assumptions can be difficult to verify in any great generality. For example, Assumption \ref{ass:LargeFluctuations} Part (c) is essentially the assumption that the subsampling chain is geometrically ergodic (with constants that are only ``polynomially bad" in the size $n$ of the dataset). While this is a common assumption, and we expect this to hold for essentially any useful subsampling algorithm, proving it for any specific algorithm is often the content of a full lengthy paper. Obviously, we have no hope of proving it in any great generality here.

To alleviate these unavoidable burdens on the reader, we give many pedagogical results that are intended to show how these definitions are used and how the assumptions of Theorem \ref{thm:SpecGapBound} may be verified. See Section \ref{SecPaperGuide} for a more detailed guide to the paper. 


\subsection{Illustrative Examples: What Can We Hope For?}

Before giving our results, it is worth recalling some goals of the subsampling literature in greater detail. Several papers on subsampling methods (see \textit{e.g.} Section 5 of \cite{fearnhead2016piecewise}) have discussed the possibility of developing algorithms whose computational cost scales ``sublinearly'' in the number of data points. We found this language mildly confusing when we first encountered it, and one goal of the present paper is to clarify it. Thus, before giving technical details, we consider a simpler question: how much improvement can we hope for with subsampling methods?  This section focuses on simple examples and can easily be skipped by readers interested only in the new results.

We begin by observing that a typical size-$m_{n}= o(n)$ subsample of a size-$n$ dataset cannot (generally) capture all the information in the original dataset. Consider the simple 2-level hierarchical model:
\be  \label{EqSimpleHierarchySetup}
\mu &\sim \mathcal{N}(0,1) \\
Y_1,\ldots,Y_n &\sim \mathcal{N}(\mu,1),
\ee
where $\mu$ is an unobserved hyperparameter and $Y_1,\ldots,Y_n$ is the observed data. If $m_n = o(n)$ is the size of a small subsample chosen uniformly at random, then it is straightforward to show that the posteriors associated with the full sample and the subsample are very far apart:
\be \label{EqSimpleExampleLimit}
\lim_{n \rightarrow \infty} \| p(\cdot \mid Y_1,\ldots,Y_n) - p(\cdot \mid Y_1,\ldots,Y_{m_n}) \|_{\TV} = 1.
\ee
Rescaling to \textit{e.g.} force the variances to match does not help for \textit{typical} subsamples. For all fixed $0 < c < 1$, we still have
\be \label{EqSimpleExampleLimitPerm}
\limsup_{n \rightarrow \infty} \P\left[ \left\| p(\cdot \mid Y_1,\ldots,Y_n) - p(\cdot \mid Y_{\sigma(1)},\ldots,Y_{\sigma(m_{n})})^{\frac{n}{m_{n}}} \right\|_{\TV} < c \right] = 0,
\ee
where $\sigma \sim \mathrm{Unif}(S_{n})$ is a random permutation and $S_n$ is the symmetric group. Thus, even in this toy problem, posterior approximations based on subsamples are not close to the full posterior in the total variation metric when the number of data points used in the subsample grows more slowly than the total sample size. Of course, this is the only setting in which computational costs per iteration are lower for subsampling.

These calculations are all well-known, and we believe that this is not what is usually intended by the phrase ``sublinear cost.''  Instead of considering algorithms that literally use a subsample that grows sublinearly with the size $n$ of the dataset, we should consider algorithms that are allowed to take a small subsample of size $m_{n} = o(n)$ \textit{and also} use some computationally-cheap summary statistics of the \textit{entire} dataset, such as the MLE. In the context of computational work, these summary statistics are also often referred to as ``control variates,'' and we will use the two terms interchangeably. The question then becomes: is it possible to develop an exact algorithm which scales sublinearly in the size $m$ of the subsample if you take advantage of summary statistics? The possibility of such an algorithm is not ruled out by the simple calculation above, and indeed it is often possible to get very good estimates using a few \textit{carefully-chosen} control variates and a very small subsample.  To cite a concrete bound, the approach of \cite{campbell2018giga} suggests that for appropriately-chosen metrics $d$,
\be
\inf_{\sigma \in S_{n}} d(p(\cdot \mid Y_1,\ldots,Y_n), \, p(\cdot \mid Y_{\sigma(1)},\ldots,Y_{\sigma(m_{n})})^{\frac{n}{m_{n}}}) \lesssim (1 - \alpha)^{m_{n}}
\ee
for some $\alpha > 0$ and all $n, m_{n}$ sufficiently large, for typical datasets. 

The perhaps-surprising conclusion of our paper is that it is often not possible to use subsampling MCMC to quickly refine initial estimates, whether they are good or not. This point is rather subtle: it suggests that very good estimates obtained by subsampling MCMC and careful use of control variates are largely due to the careful use of control variates, with the further MCMC steps improving this initial estimate very slowly as the size of the subsample increases.

To complete this discussion, we consider the reverse question: how quickly does the computational complexity of naive MCMC grow with the size $n$ of the data set? In the simple example considered in Equation \eqref{EqSimpleHierarchySetup}, it is straightforward to check that the spectral gap of a Metropolis-Hastings chain with appropriately-scaled proposals (\textit{e.g.} proposal kernel $Q_{n}(x,\cdot) = \mathrm{Unif}([x- n^{-1/2},x+ n^{-1/2}])$) does not go to 0 as $n$ goes to infinity. Thus, the cost of naive Metropolis-Hastings is $O(n)$, and the standard MCMC CLT (see Equation \eqref{eq:MCMCCLT}) tells us that it is possible to obtain an MCMC estimate with a desired fixed standard deviation $\epsilon > 0$ by accessing each data point a total of $O(\epsilon^{-2})$ times. This non-decay of spectral gap is fairly common in low-dimensional applied problems, and there is some evidence that it also occurs in realistic high-dimensional problems (see \textit{e.g.} the recent papers \cite{yang2017bigdim}, \cite{qin2017asymptotically}).

To summarize: in even the simplest example, it is necessary to view $\Omega(n)$ data points to obtain ``good'' estimates. Conversely, for a wide variety of problems it is possible to obtain estimates with error less than any desired constant $\epsilon > 0$ using standard -- i.e. non-subsampling based -- methods while viewing only $O( n \epsilon^{-2})$ data points. Thus, the main theoretical question is to determine if the dependence on the accuracy parameter $\epsilon$ can be improved.

As mentioned, this discussion avoids several important practical issues. Depending on the dataset size and computer setup, the per-step cost of running an algorithm can be highly discontinuous in the size $n$ of the dataset and algorithm parameters such as the size $r$ of each subsample. There are many causes for these discontinuities, from a fairly simple discontinuity when data becomes too large to fit in RAM to more complicated discontinuities caused by data-transfer problems on GPUs. These issues are largely out of the scope of the current paper, but we give some discussion in Section \ref{SecDisc}. 

\subsection{Literature Review} \label{SecLitRev}

 The present paper was motivated by recent work on ``online'' subsampling MCMC. Surveys of the topic include \cite{bardenet2015markov,fearnhead2016piecewise}, which focus on the interesting specific methods introduced in \cite{quiroz2015speeding,bierkens2016zig,pollock2016scalable} and more generally on the important role of control variates in exact subsampling MCMC. Auxiliary-variable methods are also an important part of this literature, and ``exact'' subsampling MCMC was introduced in this context in \cite{maclaurin2014firefly}. Several papers have recently given positive theoretical results on these sorts of subsampling MCMC algorithms; see \textit{e.g.} \cite{johndrow2015approximations,rudolf2018perturbation}. These results typically require that the approximating MCMC kernel is a uniformly good approximation of the ``exact'' transition kernel in a Wasserstein or total variation metric. An upshot of the results that we present here is that one typically requires $\Omega(n)$ subsample sizes to accomplish this using a subsampling approach.

There is also a large literature on trying to understand how subsampling methods work in other algorithms used in Bayesian statistics. Our current paper most closely resembles \cite{nagapetyan2017true}, which gives qualitatively similar conclusions for the stochastic gradient Langevin algorithm (SGLA) in some settings. Although our main messages are very similar, the details of our papers are quite different. Most obviously, we deal with generic subsampling chains while they focus on SGLA. More subtly, we provide absolute lower bounds on the convergence rates of our algorithms, while they provide only lower bounds on some natural upper bounds on convergence rates; this meant \cite{nagapetyan2017true} left open the small possibility that substantially faster convergence is still possible when their upper bounds are not sharp.

Most previous subsampling papers involve pre-computation of control variates or other surrogates for the target distribution. Some of these have suggested adaptive choice of control variates, as in \cite{bardenet2015markov} and also the author's previous \cite{conrad2016accelerating} amongst other places. From the point of view of this paper, we view the adaptive stages of algorithms as essentially a sophisticated way of choosing good control variates. In a somewhat related direction, there is also a literature on ``preprocessing'' subsampling algorithms. These algorithms usually involve extensive pre-computation; see for example \cite{huggins2016coresets,SHF22} and the references therein. This approach avoids many of the issues raised in the present paper, and we believe that it has the potential to make a large impact on Bayesian computation for popular models - especially as it becomes more integrated into MCMC techniques.

In the present paper, we have focused on GLMs and other models for which the computational complexity of easily-available ``default" methods are known to scale linearly in $n$. In other situations, default methods may scale much more unfavorably in $n$, which opens the door for subsampling methods to give a much larger improvement. For example, many MCMC algorithms require repeated solution of linear systems, which often has at least a quadratic cost in $n$. See \textit{e.g.} \cite{salomone2019spectral} for a recent example of subsampling in a setting in which computational costs scale super-linearly in $n$ for naive algorithms. We leave open the possibility of analogous bounds in this setting.

Finally, recall that in discussing the third part of our trilemma, we made the following somewhat facile suggestion: if the posterior conditioned on the data is extremely close to the posterior conditioned on just the control variates, you should simply run MCMC based on the control variates and ignore the data. While we think this is good advice and applies to many control variates used for generalized linear models, we point out that it is not always possible, even in the rather special case that the control variates are in fact sufficient statistics. See \cite{montanari2015computational}  for details in an important class of models.

\subsection{Paper Guide} \label{SecPaperGuide}

We start Section \ref{SecNotTypAppl} by introducing our notation and main assumptions. This section also includes pedagogical examples, a sketch of the proof of our main result (Section \ref{SecMBIS}), and an application of our abstract results to the familiar setting of GLMs (Theorem \ref{ThmExampleTorpid}). In Section \ref{SecMainResults}, we give a precise statement our main result and relate it back to our main pedagogical examples. In Section \ref{SecSubsecCounter}, we give some ``non-examples" of special situations in which it is possible to speed up MCMC without reducing accuracy. In Section \ref{SecApplMCMC}, we show that some of our abstract assumptions are satisfied for a variety of realistic subsampling chains. 
Finally, in Section \ref{SecDisc} we discuss the broader context of this work and highlight some important problems that are left open.

As mentioned in Section \ref{SecVeryClose}, verifying that a specific model or algorithm satisfies our assumptions with ``good" constants takes some work. Doing so in great generality is far beyond the scope of one paper. This leads to some conflict in the presentation of the results: we wish to have clear and simple applications for those interested in the ``big picture," while also helping readers to verify the assumptions in their particular examples. We have chosen to resolve this conflict as follows:

\begin{itemize}
\item Our pedagogical examples and Section \ref{SecApplMCMC} focus on applications to specific simple models and subsampling algorithms that are currently popular in the literature, in an effort to succinctly present a variety of concrete examples and results that are most likely to be of interest to practitioners. Most proofs from these sections are deferred to technical appendices.
\item Throughout the paper, we include many short remarks on technical issues that appear when using our results. We also reference the lemmas and constructions in the appendices that can be used to resolve them for most models and algorithms that we are aware of. We believe that these technical issues are to some extent unavoidable due to the generality of the current approach and the fact that algorithms can be presented in very different forms. 

To contrast again with the closely-related work \cite{nagapetyan2017true}: our current paper allows for \textit{generic} subsampling chains and gives \textit{lower bounds on their convergence rates}, while \cite{nagapetyan2017true} studies a \textit{specific} family of subsampling chains and gives only \textit{lower bounds on upper bounds on their convergence rates}. Most of the technical difficulties come from these two differences: the first means we must have a very careful definition of when a datapoint is ``used," and the second means that we must check that posterior distributions are not ``almost completely" determined by small subsamples and a collection of control variates. Appendix \ref{SecAltCons} gives alternative representations that can help with the first issue, Appendix \ref{SecGenSuffCondLarge} gives new bounds in random matrix theory that can help with the second issue, and Appendix \ref{AppSubsecGrid} illustrates how adding carefully-constructed but ``unused" control variates can help both issues.
\item Finally, Appendix \ref{AppendixLargeFluxChecking} begins with a short guide to our constructions. We emphasize here one important part of this guide. Most of the difficulty in checking our assumptions for new models comes from a single technical condition related to the singular values of certain random matrices. While we found this condition difficult to verify by hand, in many cases it is easy to verify by computer. We give a short algorithm for this verification in Remark \ref{RemSimpleAlgBadCond}.
\end{itemize}

Although the paper focuses on Markov chain Monte Carlo methods, the appendices include some more general material that may be of independent interest: Appendix \ref{SecAppAntiConc} includes some new bounds on anticoncentration of random variables, which are applied in Appendix \ref{SecAppSingVal} to obtain bounds on the condition number of certain highly-structured random matrices.

\section{Notation and Typical Applications} \label{SecNotTypAppl}

\subsection{Probability Notation}

We say that a sequence of events $\{A_{n}\}_{n \in \mathbb{N}}$ holds ``with extreme probability'' (abbreviated w.e.p.) if $\P[\{A_n^c \,\, i.o.\}]=0$, where $i.o.$ denotes ``infinitely often.''

By a slight abuse of notation, we use $p$ to denote both a distribution \textit{and} the density of the distribution when there is no possibility of confusion. 

\subsection{A Typical Bayesian Setting for Subsampling MCMC}

We give a quick description of the most common setup for approximate MCMC in Bayesian inference, so that we can easily explain what our results say in this setting. We will use similar notation throughout the paper, though our results will apply in slightly greater generality.

Denote by $\bf Z$ a random $n \times k$ matrix, representing the observed data. We assume that the rows $Z_i$ of $\bf Z$ are iid samples drawn from a distribution that has some density $f$ with respect to some dominating measure $\lambda$ on $\bb R^k$. We assume that $f$ is a member of a parametric family $\{f(\cdot,x) : x \in \mc X\}$ of probability distributions. In the setting of Bayesian inference, we also fix a prior $\Pi_0$ on $\mc X$. When $\mc X = \bb R^d$, $\lambda$ is the Lebesgue measure, and the prior has density $\pi_0$ with respect to Lebesgue measure, the posterior has density $\pi_{\bf Z}$ with respect to Lebesgue measure which is given by
\be
\pi_{\bf Z}(x) \propto \left( \prod_{i=1}^n f(Z_i,x) \right) \pi_0(x).
\ee


Denote by $K_{\bf Z}$ the transition kernel of some Markov chain associated with the observation $\bf Z$, and we denote by $X_0,X_1,\ldots$ some Markov chain evolving according to $K_{\bf Z}$. In the setting of approximate Bayesian inference, we typically want $K_{\bf Z}$ to have a stationary measure that is ``close" to the posterior $\Pi_{\bf Z}$.

Most MCMC algorithms with invariant measure \textit{exactly equal} to $\Pi_{\bf Z}$ require performing a computation with each data point to sample the next state of the chain. As a result, these algorithms have per-step computation cost that scales at least linearly in $n$. With large datasets growing more common, an extensive literature over the past decade has proposed alternative algorithms that use only a fraction of the data at each step. These algorithms may not -- and typically do not -- have invariant measure $\Pi_{\bf Z}$. We refer to algorithms that do not use all the data at each step as ``minibatching'' or ``subsampling'' algorithms. 

We now introduce two simple pedagogical examples that we will return to: stochastic gradient Langevin dynamics (SGLD) and a subsampling pseudo-marginal Metropolis-Hastings (SPMMH) algorithm similar to algorithms proposed in \cite{quiroz2015speeding}. In our pedagogical examples, we will consider the basic form of both algorithms without control variates, but our main result applies also to versions of both algorithms that use control variates. SGLD without control variates has been analyzed previously \cite{nagapetyan2017true}, and therefore, in addition to helping make our abstract conditions and results more concrete, it allows us to compare our bounds to existing results.

\begin{example}[SGLD] \label{ExSGLD}
Fix parameter $r  \in [n] = \{1,\ldots,n\}$ and a sequence $\{\eta_t\}$ of positive reals. Let $\{\epsilon_t\}$ be a sequence of independent $N(0,I_d)$ random variables and $\{S_t \}$ a sequence of independent uniform samples of size $r$ taken without replacement from the integers $[n]$. The SGLD algorithm is a Markov chain with update rule
\be
X_t = X_{t-1} - \frac12 \eta_t \left(\nabla \log \pi_0(X_{t-1}) + \frac{r}{n} \sum_{i \in S_t} \nabla_x \log f(Z_i,X_{t-1})\right) + \eta_t^{1/2} \epsilon_t. 
\ee
\end{example}

We note that these dynamics are \textit{not} reversible, and thus are \textit{not} covered by the strongest results in our paper (though they are covered by weaker results). We keep them as-is because they are the simplest commonly-used minibatch algorithm. Very similar reversible algorithms have been developed (see \textit{e.g.} \cite{RevSGLD22}), and our results do apply directly to these algorithms. See our forthcoming paper \cite{UpcomingAzeem} for strong bounds that apply directly to SGLD.

\begin{example}[SPMMH] \label{ExSPMMH}
Let $\ell_i(x) = \log f(Z_i,x)$, and $Q(x,\cdot)$ a reversible Markov transition kernel on $\mc X$ with transition density $q(x,x^*)$ with respect to Lebesgue measure on $\mc X$. For a subset $S \subset [n]$, define the log-likelihood estimator and an estimator of its variance
\be
\hat{\ell}_S := \frac{n}{r} \sum_{i \in S} \ell_i, \quad \hat \sigma^2_S := \frac{1}r \sum_{i \in S} (\ell_i - \hat \ell_S)^2.
\ee
Now define a bias-adjusted likelihood estimator
\be \label{eq:BiasAdjust}
\hat{f}_{S} = e^{\hat \ell_S - \frac{n^2}{2r} \hat \sigma^2_S}.
\ee
The SPMMH algorithm is a Markov chain on the augmented state space $\mathcal{X} \times 2^{[n]}$. We give the update rule assuming that the current state is $(X_{t-1},S_{t-1})$.
\begin{enumerate}
    \item Sample $X^* \sim Q(X_{t-1},\cdot)$. Sample a random subset $S^*$ of size $r$ uniformly from $[n]$, and compute $\hat{f}_{S^*}(X^*)$. 
    \item Calculate
    \be
    \alpha((X_{t-1},S_{t-1}),(X^*,S^*)) := 1 \wedge \frac{\hat f_{S^*}(X^*) \pi_0(X^*) q(X^*,X_{t-1})}{\hat f_{S_{t-1}}(X_{t-1}) \pi_0(X_{t-1}) q(X_{t-1},X^*)}.    
    \ee
    \item Sample $U \sim \Uniform(0,1)$. If $U \le \alpha$, move to $(X_t,S_t) = (X^*,S^*)$. Otherwise remain at $(X_t,S_t) = (X_{t-1},S_{t-1})$. 
\end{enumerate}
\end{example}

The authors of \cite{quiroz2015speeding} show that SPMMH is a pseudo-marginal MCMC algorithm with invariant distribution $\bar \pi_{\bf Z}(x)$ given by
\be
\bar \pi_{\bf Z}(x) &= \frac{f_r(x,\bf Z) \pi_0(x)}{\bar f_r(x,\bf Z)} \\
\intertext{where}
f_r(x,\bf Z) &= \sum_{S \in \mc S_r(n)} \binom{n}{r}^{-1} \hat f_{S}(x,\bf Z), \text{ and } \bar f_r(x,\bf Z) &= \int f_r(x,\bf Z) \pi_0(x) dx. \\
\ee
In general, $\bar \pi_{\bf Z}(x) \ne \pi_{\bf Z}(x)$, since, despite the first order correction in \eqref{eq:BiasAdjust}, $\hat f$ is not proportional to $f$.

Both the SGLD and SPMMH algorithms use exactly $r$ data points, sampled uniformly at random from the $n$ rows of $\bf Z$, to approximate either the gradient of $\log \pi_{\bf Z}$ or the acceptance probability $\alpha$ at each step. If the gradient evaluation and likelihood evaluation on a single $Z_i$ both cost order $d$, which is typical, this will mean that the computational cost per step of either algorithm is order $r d$. It is immediately clear that if we wish to reduce the computational cost per step compared to exact full data algorithms such as the Metropolis-Adjusted Langevin algorithm or Metropolis-Hastings, which costs order $n d$ per step, then we must have $r = o(n)$.

\subsection{Convergence Rates and MCMC}

The cost per step does not entirely capture the computational cost of MCMC - we must also take into account the number of steps that are run. We briefly recall standard convergence results from the MCMC literature, and use them to motivate a notion of ``total" cost for a subsampling MCMC chain.

In this section, we consider a Markov chain $K$ with unique invariant measure $\mu$. We are mainly interested in Markov chains for which convergence occurs at an exponential rate in number of steps:

\begin{definition} \label{def:Gergo}
A Markov transition kernel $K$ is \textit{geometrically ergodic} if there exists a constant $\lambda \in (0,1)$ and a function $L : \mc X \to [0,\infty)$ such that
\be \label{eq:Gergo}
\|\delta_x K^t - \mu\|_{\TV} < L(x) (1-\lambda)^t.
\ee
\end{definition}

If $K$ is geometrically ergodic, there is a Markov chain analogue of the usual central limit theorem. For functions satisfying $\varphi \in L^{2+\delta}(\mu)$ for some $\delta > 0$, we have from \cite{jones2004markov}[Theorem 9]
\be \label{eq:MCMCCLT}
\sqrt{T} \left( \mu (\varphi) - T^{-1} \sum_{t=1}^T \varphi(X_t) \right) \rightsquigarrow \N(0,\sigma^2_{\varphi}),
\ee
where $\sigma^2_\varphi$ is called the asymptotic variance of the function $\varphi$.  The quantity $\sigma^2_\varphi$ can be bounded by $\lambda$ as
\be \label{eq:AsymptoticVarianceNonRev}
\sup_{\varphi \in L^{2+\delta}(\mu)} \frac{\sigma^2_\varphi}{\var_{\mu}(\varphi)} \le \frac{2}{\lambda}.
\ee

Under the additional assumption that $K$ is reversible, substantially sharper results hold. First, the central limit theorem \eqref{eq:AsymptoticVarianceNonRev} holds for all $\varphi \in L^2(\mu)$. Next, recall the definition of the \textit{spectrum}:

\be 
\mathrm{Spec}(K) = \{\lambda \in \bb C \setminus \{0,1\}, \,   (K - \lambda \, \mathrm{Id}) \, \text{ is not invertible.} \},
\ee 

where $\mathrm{Id}$ is the identity operator. The (absolute) \textit{spectral gap} $\lambda(K)$ of a transition kernel $K$ is given by
\be \label{EqDefAbsSpec}
\lambda(K) = 1 - \sup \left\{ |\lambda| \, : \, \lambda \in \mathrm{Spec}(K)  \right\}.
\ee
When $K$ is reversible and the spectrum $\mathrm{Spec}(K)  \subset [0,1]$ is nonnegative, this gives a precise bound on the convergence of the Monte Carlo averages for the \textit{worst-case} function $\varphi$:

\be \label{eq:AsymptoticVarianceRev}
\sup_{\varphi \in L^2(\mu)} \frac{\sigma^2_\varphi}{\var_{\mu}(\varphi)} = \frac2{\lambda} - 1 \le \frac{2}{\lambda}.
\ee

In the context of the main results in this paper, we can think of Equation \eqref{eq:AsymptoticVarianceRev} as ``almost" holding for reversible kernels with a possibly-negative spectrum, as made formal in the following remark. 

\begin{remark}
Consider a reversible kernel $K$, and let $K' = \frac{1}{2}(K + \mathrm{Id})$ be the associated ``half-lazy" kernel. We note two important facts about $K'$: it always has nonnegative spectrum
\be 
\mathrm{Spec}(K') = \left\{ \frac{1}{2}( 1 + \lambda) \, : \, \lambda \in \mathrm{Spec}(K) \right\} \subset [0,1], 
\ee 
and its asymptotic variance differs by at most a factor of 2 relative to $K$. Putting these together, the asymptotic variance of the ergodic averages associated with $K$ are very tightly linked to the spectral gap of its $\frac{1}{2}$-lazy kernel $K'$. Furthermore, whenever $K$ satisfies  the assumptions of our main theorems, the half-lazy version $K'$ does as well. Thus, when considering a given kernel $K$, we can always apply our results to $K'$ and obtain as an immediate corollary essentially the same conclusion for $K$ itself, losing at most a factor of 2.

For this reason, issues of positivity of the spectrum for reversible kernels can be safely ignored throughout the remainder of the paper - we can always apply our results to an associated kernel with nonnegative spectrum and then draw conclusions for the original kernel of interest.
\end{remark}

Notice that by \eqref{eq:AsymptoticVarianceRev}, even if we can initialize from $\mu$, we will require order $\lambda^{-1}$ samples to achieve some fixed approximation accuracy for the quantity $\mu(\varphi)$. We denote by $\tau_{\bf z}(m)$ a quantity that is, roughly, the typical number of steps it takes for a subsampling algorithm to ``see" $m$ out of $n$ datapoints\footnote{The precise definition of the time $\tau_{\bf z}(m)$ to ``see" $m \in \{1,2,\ldots,n\}$ is given in Equations \eqref{EqDefTauPre} and \eqref{EqDefTau}. For the algorithms in Examples \ref{ExSGLD} and \ref{ExSPMMH} that use $r$ datapoints per sample, it is straightforward to check that this definition satisfies $\tau_{\bf Z}(m) \geq \frac{m}{r}$ for all $m,n,r$. In our applied examples, we will set $m = n - d - 1$, where $d$ is the ambient dimension of the parameter space. }; for our pedagogical examples, this will be $\tau_{\bf z}(m) \lessapprox \frac{m}{r} \log(n)$. Following previous work such as \cite{nagapetyan2017true, bornn2017use}, we define the overall cost of a  subsampling MCMC algorithm $K_{\bf Z}$ that satisfies \ref{def:Gergo} with optimal constant $\lambda_{\bf Z}$ by 

\be \label{eq:CompCost}
\Cost(K_{\bf Z}) = \frac{m}{\tau_{\bf Z}(m) \lambda_{\bf Z}}.
\ee

This definition accounts for the computation time per step and the number of samples needed to obtain a good approximation to $\mu_{\bf Z} (\varphi)$ for worst-case $L_2(\mu_{\bf Z})$ functions $\varphi$, though it does not account for the fact that $\mu_{\bf Z}$ need not be the posterior of interest; see Section \ref{SecMBIS} for a discussion that accounts for bias. Note that, though we can use this definition for all kernels $K$, Equality \eqref{eq:AsymptoticVarianceNonRev} may be far from true for non-reversible $K$. See Section \ref{SubsecHighlyNonrev} for a discussion of how the looseness of \eqref{eq:AsymptoticVarianceNonRev} may allow some non-reversible subsampling algorithms to escape our main conclusions. We give some limits on this in our forthcoming note \cite{UpcomingAzeem}.

Finally, we note that the prefactor $L(x)$ in Definition \ref{def:Gergo} can be enormous, in which case you can't ``observe" the convergence rate $\lambda$ until some very large number of steps $t \gtrsim \lambda^{-1} \log(L(x))$. In practice, it is often possible (and desirable) to choose a starting point or measure for which the prefactor to this exponential convergence rate is not too large. This motivates the definition:

\begin{definition} \label{def:Warm}
Let $\mu^*$ be a measure on $\mc X$, and let $K$ be a Markov chain that is geometrically ergodic with rate $\lambda$. We say that $\mu^*$ is a warm start for Markov transition kernel $K$ (and this rate $\lambda$) if there exists a constant $C_{\mu^*}<\infty$ such that
\be \label{eq:WarmStart}
\|\mu^* K^t - \mu\|_{\TV} < C_{\mu^*} (1-\lambda)^t.
\ee
\end{definition}

For reversible chains, we can always take $C_{\mu^{*}} = \| \frac{d \mu^{*}}{d \mu} \|_{2}^{2}$ when that quantity is finite.

\subsection{Informal Sketch of Theorem \ref{thm:SpecGapBound}} \label{SecMBIS}

Consider two datasets $Z^{(1)},Z^{(2)}$ of size $n$ having the same distribution (but perhaps not sampled independently), and denote by $T$ a set of control variates/summary statistics allowed to depend on the entire dataset. For fixed $M \subset \{1,2,\ldots,n\}$ and $z, t$, let $\mc A = \mc A(M,z,t)$ be the event
\be \label{EqDefAInf}
\mc A = \{Z^{(2)} = z, Z_{M}^{(1)} = z_{M}, T(Z^{(2)}) = T(Z^{(1)}) = t\}.
\ee
 Let $K_{Z^{(j)}}$ be a subsampling MCMC transition kernel using subsets of the dataset $Z^{(j)}$ at each iteration with invariant measure $\pi_j$, and denote by $p_j(\cdot) = p(\cdot \mid Z^{(j)})$ the posterior measure under dataset $j$. Note that it need not be the case that $\pi_j = p_j$. 

One of our main results is the following. So long as there exists a coupling of $Z^{(1)},Z^{(2)}$ such that, \emph{conditional on $\mc A$}, the inequality
\be \label{eq:IneqLargeFluctuationInformal}
\| \pi_1 - p_1 \|_{\TV} \ll \|p_1 - p_2\|_{\TV}
\ee
holds with high probability with respect to the distribution of $(Z^{(1)},Z^{(2)})$, then subsampling MCMC can reduce overall computational cost by at most a factor of $\log(n)$. For an explanation of why using total variation in Inequality \eqref{eq:IneqLargeFluctuationInformal} gives a strong result, see Remark \ref{RemChoiceMet} immediately following Assumption \ref{ass:LargeFluctuations}.

We refer to Inequality \eqref{eq:IneqLargeFluctuationInformal} as a ``large fluctuation'' condition. When this holds, \textit{even if} we fix two data sequences that agree in their first $m < n$ elements \textit{and} some set of statistics $T$ calculated from the entire dataset, changing \textit{only} the remaining $(n-m) \ll n$ data points still changes the posterior ``substantially.'' Heuristically, it says that the invariant measure $\pi_1$ of $K_{Z^{(1)}}$ is a better approximation to the posterior $p_1$ given data $Z^{(1)}$ than the posterior $p_2$ given some \textit{other} typical dataset $Z^{(2)}$.

We can now give a slightly more precise statement of our earlier trilemma in the context of Inequality \eqref{eq:IneqLargeFluctuationInformal}. There are three possibilities, at least one of which must be true:
\begin{enumerate}
\item Inequality \eqref{eq:IneqLargeFluctuationInformal} holds (in which case the computational cost can improve by at most a factor of $\log(n)$ with subsampling); or
\item Inequality \eqref{eq:IneqLargeFluctuationInformal} fails with the LHS large (which suggests the algorithm is not much use for Bayesian inference); or
\item Inequality \eqref{eq:IneqLargeFluctuationInformal} fails with both sides small (which says that the control variates essentially determine the posterior).
\end{enumerate}

It is natural to try to separately estimate the spectral gap $\lambda_z$ and per-step cost $n\tau_{z}^{-1}$ in Equation \eqref{eq:CompCost}. This is basically the approach of many previous articles on computational efficiency of approximate MCMC algorithms, including \cite{nagapetyan2017true}. However, if we want to analyze many complex algorithms, we have found that it is often \textit{vastly} easier to bound the product directly. To give some indication of the disparity in difficulty, using existing techniques to bound the spectral gap of a \textit{single} Markov chain targeting logistic regression can be the subject of a long paper \cite{johndrow2019mcmc,qin2019convergence}; here we set up new techniques and give good bounds on the cost of \textit{all} Markov chains targeting logistic regression.

\subsection{Using Data: Formal Definitions and Common Algorithms} \label{SecUsingData}

For SGLD, SPMMH, or any algorithm that uses a fixed minibatch size at each step and \textit{does not} use any control variates, there is a natural definition of the ``data used each step." However, this natural definition is not appropriate  for all subsampling algorithms. In this section, we will formalize the notion of ``using" a data point and give a more general definition of a subsampling algorithm that agrees with the intuitive one in the case of simple examples.

To formally define ``using data,'' we recall the iterated random function representation of a Markov chain. Essentially any implementation of a transition kernel $K$ on a state space $\mc X$  has a ``random mapping" representation in terms of a collection of functions $\{\zeta_{U} : U \in \mathcal{U}\}$ and distribution $\nu$ on a space $(\mathcal{U}, \mc F_{\mathcal{U}})$. Such a mapping satisfies the following equality of distributions: for all $x \in \mc X$, if $U \sim \nu$, then $\zeta_{U}(x) \sim K(x,\cdot)$.\footnote{If we think of an implementation of $K$ on a computer, $U$ corresponds to the collection of all random variables generated during a call to the algorithm and $\zeta_{U}(x)$ corresponds to the result of a call to the algorithm from starting point $x$.} In all of our examples, we can take $\mathcal{U} = [0,1]^{\tilde d}$ for $\tilde d \in \bb N$.

In the current setting, we will be computing probabilities both with respect to the law of the Markov chain and the data $\bf Z$, so it is useful to explicitly indicate dependence on the data in the notation. With this in mind, we define the state evolution inductively by 
\be \label{eq:IRFRecursion}
X_{t+1} = \zeta_{U_{t}}(X_t,\bf Z), \quad U_{t} \sim \nu.
\ee
The advantage of this representation is that, conditional on the sequence $U_0,\ldots,U_t$, the data $\bf Z$, and the initial state $X_0 = x_0$, the path $X_0,X_1,\ldots,X_{t+1}$ is deterministic. For illustration, we now give an explicit iterated random function representation for SGLD:

\begin{example}[SGLD continued] \label{eq:SGLDRandomFunctions}
Let $\psi$ be a bijective map between the partition of $[0,1]$ consisting of consecutive intervals of width $\binom{n}{r}^{-1}$ and the subsets $\mc A_r(n)$ of size $r$ of $\{1,\ldots,n\}$. Define the function $\zeta :  [0,1]^{d+1} \to \bb R^{d} \times \mc A_r(n)$ by 
\be
\zeta(u) = (\Phi^{-1}(u_1),\ldots,\Phi^{-1}(u_d),\psi(u_{d+1}))', 
\ee
where $\Phi^{-1}$ is the standard Gaussian quantile function. Then the SGLD algorithm can be written as $X_{t+1} = \zeta_{U}(X_t,\bf Z)$ where
\be 
\zeta_{U}(x,\bf Z) = x - \frac12 \eta_t \left(\nabla \log \pi_0(x) + \frac{r}{n} \sum_{i \in \zeta_{d+1}(U)} \nabla \log f(Z_i,x)\right) + \eta_t^{1/2} \zeta_{1:d}(\theta) 
\ee
and $U$ is a $(d+1)$-vector of independent $\Uniform(0,1)$ random variables.
\end{example}

Many minibatching algorithms use control variates in addition to minibatches of data. For our purposes, a \textit{control variate} is a pre-defined function $H : \mc Z \to \mc H$ that is computed once on the data $\bf Z$ prior to running MCMC. For any $\bf Z \in \mc Z^{n}$ and $M \subset [n]$, let $\bf Z_M$ denote the rows of $\bf Z$ corresponding to the indices in $M$. Following Equation \eqref{EqDefAInf}, define the collection of datasets that agree with $\bf Z$ on $M$, and also agree with the value $h$ for the control variates, by
\be \label{eq:ConditioningEvent}
\mc A(M,\bf Z, h) = \{\tilde{\bf Z} : \bf Z_M = \tilde{\bf Z}_{M}, H(\bf Z) = H(\tilde{\bf Z}) = h\}.
\ee

We now give a formal definition of accessing data using the iterated random functions construction. 

\begin{definition}[Using data] \label{def:UsingData}

Fix a dataset $\bf Z$, initial state $x_0$, and sequence $U_0,\ldots,U_t \stackrel{iid}{\sim} \nu$. Let $h = H(\bf Z)$. Let $x_{t} = (\zeta_{U_{t-1}} \circ \cdots \circ \zeta_{U_0})(x_0,\bf Z)$. We will define the collection $M_{t}$ of indices $i \in [n]$ that are \textit{used by time $t$} inductively, as follows. Set $M_{-1} = \emptyset$. Say that a set $I \subset [n] \backslash M_{t-1}$ is \textit{used for the first time at $t$} if there exists $\bf Z^{(1)}, \bf Z^{(2)} \in \mc A(M_{t-1},\bf Z,h)$ satisfying \textit{both} the condition
\be 
\zeta_{U_t}(x_{t},\bf Z^{(1)}) \neq \zeta_{U_t}(x_{t},\bf Z^{(2)})
\ee 
\textit{and} the condition 
\be \label{eqMinimalityPart}
\bf Z^{(1)}_{j} = \bf Z^{(2)}_{j}, \, j \notin M_{t-1} \cup I,
\ee 
and finally no strict subset of $I$ satisfies \textit{both} of these conditions.

Finally, let $I_{t}$ be the sets used at time $t$ and let $M_t = M_{t-1} \cup I_t$. Set $m_{t} = |M_{t}|$. Notice that once we fix $\bf Z, x_{0}, U_0,\ldots,U_t$, this set is deterministic, so it makes sense to write $m_{t} = m_{t}(\bf Z, x_{0}, U_0,\ldots,U_t)$.
\end{definition}

\begin{remark}
For both Examples \ref{ExSGLD} and \ref{ExSPMMH}, the data ``first used" at time step $t$ is (with probability 1) exactly the minibatch $S_{t}\backslash (\cup_{j=1}^{t-1}S_{j})$. We use a more technical definition to account for algorithms that use variable amounts of data or more complicated and discontinuous decision rules.

\end{remark}

We now formally define a subsampling algorithm. Let
\be \label{EqDefTauPre}
\hat{\tau}_{\bf Z}(m) = \hat{\tau}_{\bf Z}(m;X_0,U_0,U_1,\ldots) = \min_t \{ t : m_t(\bf Z,x_0,U_0,\ldots) \geq m\},
\ee
the number of steps required to use $m$ data points. The behavior of $\hat{\tau}_{\bf Z}$ and $m_t$ will be of particular interest, and we define subsampling MCMC in terms of these quantities.

\begin{definition} \label{def:Subsampling}
Fix a control variate $H$ and a density $f$ on $\mc Z$. Fix a sequence of integers $m = m(n)$ and a family of probability measures $\{\mu^*_{z}\}_{z \in \mc Z}$ on $\mc X$. We say that $\{K_{z} \}_{z \in \mc Z}$ is a subsampling MCMC algorithm if there exist constants $0 < c_0 < 1$ and $0 < C_0 < \infty$ such that
\be \label{eq:SubsamplingProb}
\bb P[\hat{\tau}_{\bf Z}(m;X_0,U_1,U_2,\ldots) < C_0 n^{c_0}] \to 0
\ee
where 
\be
\bf Z \sim f, \quad X_0 \sim \mu^*_{\bf Z}, \quad U_0,U_1,\ldots \stackrel{ind}{\sim} \nu.
\ee
\end{definition}

Finally, define
\be \label{EqDefTau}
\tau_z(m) = \inf_{x} \sup \left\{r \in \mathbb{R} \, : \, \P[\hat{\tau}_z(m; x, U_{0},U_{1},\ldots) >r] \geq 1 - n^{-2} \right\},
\ee
the time it takes to ``use'' $m$ data points with high probability, from the best possible starting point.

\begin{remark}
Note that the constants $c_{0}, C_{0}$ can depend on all of the fixed quantities, including \textit{e.g.} the control variate $H$. This may be surprising at first, but it is not an accident. The observation behind this choice is that $H$ may contain a large amount of information; in the extreme, $H$ may just be the identity map. When $H$ is very informative in this way, an algorithm that \textit{appears} to be subsampling data at each step may \textit{effectively} have access to all of the data through $H$. 

Of course, a specific implementation of an algorithm may access a datapoint that is not ``used" in the sense of Definition \ref{def:Subsampling}. Because of this discrepancy, our definition may \textit{undercount} the number of datapoints accessed by an algorithm, but it will never \textit{overcount.} In this sense our definition is conservative, portraying subsampling algorithms in the best possible light even if the usual implementations are less efficient.
\end{remark}

\subsubsection{Behaviour of $\hat{\tau}$ for n\"{a}ive subsampling}

It's clear that for a n\"{a}ive subsampling algorithm, $\hat{\tau}_{\bf Z} > \lfloor n/r \rfloor$ always. Thus, if $r = n^{1-c_0}$, \eqref{eq:SubsamplingProb} holds. The next remark shows that for n\"{a}ive subsampling, the estimate $\hat{\tau}_{\bf Z} = n/r$ cannot be off by more than a factor of $\log n$. 

\begin{remark}
Suppose $K_{\bf Z}$ selects exactly $r$ data points uniformly at random, without replacement, at each step. Then for any $\epsilon > 0$ and any $\bf Z \in \mc Z$
\be
\bb P\left[ \hat{\tau}_{\bf Z}(n) > \frac{n \log n + \epsilon n}{\lfloor r \rfloor} +1 \right] < e^{-\epsilon}
\ee
\end{remark}
\begin{proof}
Let $T$ be the time required to collect $n$ data points when selecting one at a time, and let $\tilde T$ be the time required to collect $n$ data points when drawing $r$ at a time \emph{with replacement}. Clearly, $T/r-1 \le \tilde T \le T/r$. Moreover, $\tilde T$ stochastically dominates the time required to collect all $n$ data points when collecting $r$ at a time \emph{without replacement}, then replacing them and drawing again.

From \cite{levin2017markov}[Proposition 2.4],
\be
\bb P[T > \lceil n \log n + c n \rceil] < e^{-c},
\ee 
which implies 
\be
\bb P[\tilde T > \lceil \frac{n \log n + c n}{r} \rceil + 1] < e^{-c}
\ee
and hence the conclusion.
\end{proof}

\subsection{Application to MCMC for Generalized Linear Models} \label{SecAppGLM1}

To illustrate these ideas, we now give a corollary of our main results that applies to an important special case: using a reversible Markov chain to sample from the posterior under a generalized linear model, perhaps taking advantage of some simple control variates. In this setting, our result gives sufficient conditions under which computation time for subsampling MCMC must scale almost-linearly in the total size of the dataset, regardless of the details of the MCMC algorithm used. 

We first fix a class of statistical models. Denote by $\theta = (\beta,\sigma)$ the model parameters and
\be \label{EqLogRefDef}
p( y \mid \theta, x) = \prod_{i=1}^{n} a(y_i,\sigma) e^{ \frac{(x_i\beta) y_i - c(x_i\beta)}{d(\sigma)}}
\ee
the usual generalized linear model (GLM) in canonical form, where $y = (y_1,\ldots,y_n)$ and $x=(x_1,\ldots,x_n)$ denote the response data and covariate data, respectively. Let $z = (x,y)$ denote the full dataset. The functions $a,c,d$ are well-known for commonly used models. For example, for logistic regression we have $c(x) = \log(1+x)$, $d(\sigma) = 1$, and $a(y,\sigma) = 1$. 

Fix a ``true'' model parameter $\theta_0 = (\beta_{0}, \sigma_{0})$. For the rest of this section, we view both the covariate and response data as random, and we view the parameter $\sigma_{0}$ as fixed and known. Assume that the covariate data is sampled $X_1, X_2,\ldots \stackrel{ind}{\sim} \gamma$.\footnote{When defining data from GLMs, we use $X_{1},X_{2},\ldots$ and $Y_{1},Y_{2},\ldots$ to refer to data rather than coupled Markov chains on the associated parameter space. This slightly conflicts with the notation used in the majority of the paper, but agrees with the usual notation for GLMs.} We make the following mild assumption on $\gamma$:

\begin{assumption} \label{GoodGamma}
We assume that
\begin{enumerate}
\item $\gamma$ is supported on a compact set $I \subset \mathbb{R}^{d}$ containing the unit ball, and furthermore it has a smooth density $g$ with respect to Lebesgue measure. For $n \in \mathbb{N}$, define $g_{n}(x_{1},\ldots,x_{n}) = \prod_{i=1}^{n} g(x_{i})$ to be the product density; we assume there exist constants $0 < d_{1} < \infty$, $1 < d_{2} < \infty$ so that, for $X \sim g_{n}$, the following holds w.e.p.: 
\be \label{IneqGoodGammaUncond}
\sup_{x \in I^{n} } \sup_{\tilde x \, : \, \|X-\tilde x\| \leq n^{-d_{1}}} \left| \frac{g_{n}(X)}{g_{n}(\tilde x)} - 1 \right| \leq n^{-d_{2}}.
\ee
\item Denote by $\gamma_{n,Y}$ the conditional distribution of $X$ given response variable $Y$, and $g_{n,Y}$ its density. Assume that w.e.p. Inequality \eqref{IneqGoodGammaUncond} also holds with $g_{n}$ replaced by $g_{n,Y}$. Note that, in this probability statement, $\gamma_{n,Y}$ is a random measure that depends on $Y$.
\end{enumerate}
We say that a \textit{deterministic} $y$ satisfies the second part of this assumption if \eqref{IneqGoodGammaUncond} holds with $g_{n}$ replaced by $g_{n,y}$. 
\end{assumption}

Let $Y_i \mid X_i \stackrel{ind}{\sim} p(\cdot \mid \theta_0,X_i)$ be drawn independently from the same regression model with parameter $\theta_0$ and covariates $X_i$. Finally, fix some prior distribution on $\bb R^d$ that has a smooth density with respect to Lebesgue measure and subexponential support. Under this assumption on the prior and $\gamma$, and with data generated from a true GLM with parameter $\theta_{0}$, the MLE $\hat{\beta}$ converges to $\beta_{0}$ almost surely (see \textit{e.g.} Theorem 2 of \cite{fahrmeir1985consistency}).

Next, we define the large class of data-augmentation Markov chains that our results apply to. We consider \textit{any} fixed collection $\{K_{(x,y)}\}_{x \in \bb R^{dn}, y \in \bb R^n}$ of reversible Markov kernels on an ``augmented'' state space $\mathbb{R}^{d} \times \mathbb{A}$, indexed by data sets $z=(x,y)$. We allow augmented state spaces here because many subsampling algorithms, such as \cite{maclaurin2014firefly}, use them. Note that the more abstract notation used for subsampling algorithms up to this point allows for such augmented state spaces. Denote by $\pi_z$ the stationary distribution of $K_z$. 

Define the maps $T_{1}^{(n)}, \ldots, T_{d}^{(n)}$ to be the maps that return the components of the sample MLE for $\beta$. For this section only, we make the assumption that the kernel $K_Z$ with stationary measure $\mu_Z$ satisfies the \textit{geometric ergodicity} condition 

\be \label{SpecGeomErg}
\| \delta_{x} K_Z^{t} - \mu_{Z} \|_{\TV} \leq L(x) \, (1-\lambda_{Z})^{t}, 
\ee 
and furthermore that for $x^{\ast}$ the MLE we have $\log(M(x^{\ast})) = O(\log(n)).$\footnote{As discussed in Appendix \ref{AppSecWarmStart}, this bound is known to hold for many popular algorithms. While it is easy to build chains for which this assumption fails, in the context of this paper we view the assumption as being very weak. When the assumption fails, the ``burn-in" time of the Markov chain grows very quickly in $n$, so the resulting algorithm is very computationally inefficient and likely not of practical interest.}

We then have the following result:

\begin{thm} [Slow mixing of Subsampling MCMC for GLMs] \label{ThmExampleTorpid}
Set $m = n-d-1$ and consider the usual logistic regression model. Under the notation given in Section \ref{SecAppGLM1} and Assumptions \ref{GoodGamma} and \eqref{SpecGeomErg}, there exists some constant $0 < C < \infty$ so that the cost of $K_Z$ satisfies
\be \label{EqThmExampleTorpid}
\Cost(K_Z) \equiv \frac{m}{\lambda_{Z} \, \tau_{Z}(m) } \geq C \frac{n}{\log(n)}
\ee
w.e.p., where the probability is with respect to the distribution of $Z$ and where the constant $C$ depends only on the data-generating process $\gamma$ and the prior.
\end{thm}

\begin{proof}
This result is proved as an immediate corollary of Theorems \ref{thm:SpecGapBound} and \ref{ThmActualLogisticRegression}, the latter of which appears in Appendix \ref{AppendixLargeFluxChecking}.

\end{proof}

We note that Theorem \ref{ThmExampleTorpid} looks quite different from most theorems in the literature on bounding the spectral gaps of Markov chains. On one hand, the result holds for \textit{any} sequence of Markov chains, based only on their stationary measures and some smoothness properties. It is much more common to give bounds for a \textit{specific} Markov chain with a given stationary distribution, and so in this sense our result is extremely general. On the other hand, our result is only valid with high probability for large $n$, whereas many theorems about Markov chains would hold for every finite $n$. Since our goal is to do a computational complexity analysis, this result is sufficient for our purposes.

These two oddities are closely related. It is not difficult to construct a family of Markov chains that is very efficient for \textit{one} dataset $(x,y)$. Thus, we can't hope to get a bound that holds for \textit{all} datasets. However, our theorem is nearly as good: while you can ``cheat'' and get a good Markov chain for a specific dataset, you can't write down a family of Markov chains that is good for \textit{typical} datasets. In particular, any family of Markov chains you write down is unlikely to be good for \textit{your} dataset. Thus, at a minimum one needs to ``solve for'' a good transition kernel $K_z$ that works for each particular dataset $z$, and of course we expect this to be computationally difficult.

\subsubsection{Sharpness of Theorem \ref{ThmExampleTorpid}}

For many specific families of algorithms that have been proposed, it is straightforward to sharpen the right-hand side of Inequality \eqref{EqThmExampleTorpid} to $Cn$ (see Section \ref{SecDataAug}). On the other hand, for many specific examples, the computational cost of well-tuned naive Metropolis-Hastings is $O(n)$. When these both apply, subsampling can reduce computational cost by at most a constant factor. In Section \ref{SecDataAug}, we check that some simple subsampling algorithms really do incur this extra ``logarithmic penalty'' under some moderate conditions. In Section \ref{SecSystScan}, we discuss an approach to avoiding this logarithmic penalty.

While Theorem \ref{ThmExampleTorpid} is stated for specific control variates and a specific GLM, our results hold under vastly more general conditions. See Section \ref{RemVerifContVar} in the appendix for a guide to various techniques that can be used to check this condition, either by pen-and-paper or computer-aided calculations. In particular, the biggest difficulty in extending this result to more general models is checking that a certain determinant is not too close to 0; Section \ref{RemVerifContVar} points out that checking this with a computer for a fixed model is often easy, even if checking it in full generality with pen-and-paper is difficult. 

We are not aware of any literature on sample statistics with fixed size $k$ that do not satisfy these conditions for most GLMs, nor are we aware of any general strategies for finding ``good'' fixed-size collections that are not adaptive. This includes \textit{e.g.} \cite{huggins2017pass}, which gives exponential convergence \textit{in $k$} for \textit{fixed $n$}. In contrast, our conditions concern exponential convergence \textit{in $n$} for \textit{fixed $k$}, which is much more difficult to achieve. In Section \ref{SecSubsecCounter}, we discuss some simple models and control variates that fail our conditions; we believe these non-examples shed some light on why our conditions hold so often for realistic examples.

\section{General Theorems} \label{SecMainResults}

We introduce our last substantial assumptions and then state our main abstract results.

\subsection{Large fluctuations}

Informally, we might expect subsampling MCMC to perform poorly when changing one datapoint leads to a ``substantial" change to the posterior. This section formalizes this idea in a ``large fluctuation" condition. 

Recall the definition of $\mc A = \mc A (\tilde{S},\tilde{\bf Z},h)$ from Equation \eqref{EqDefAInf}. Denote by $f_{\mc A(\tilde{S},\tilde{\bf Z},h)}$ the conditional distribution of $\bf Z$  given $\mc A(\tilde{S},\tilde{\bf Z},h)$. 
Our main ``fluctuation" assumptions are:


\begin{assumption} \label{ass:LargeFluctuations}
Fix a collection of constants $\{c_{i},C_{i}\}_{i=0}^{2}$ and sequences $\{\omega_{i}(n)\}_{i =1}^{3}$ satisfying $\omega_i(n) \searrow 0$. 

We consider sequences of positive integers $\{\tilde m_n\}, \{s_n\}$ with $\tilde m_n < n$ and $s_n < C_0 n^{c_0}$ and sequences of sets $\tilde{\mc H}^{(n)} \subset \mc H$, $\tilde{\mc Z}^{(n)} \subset \bb R^{n k}$, and $\mathcal{M}^{(n)} \subset 2^{[n]}$. Define  $\mathcal{G}^{(n)} = \{z \in \mathbb{R}^{nk} \, : \, z \in \tilde{\mc Z}^{(n)}, \, H(z) \in \tilde{\mc H}^{(n)}\}$. 

We assume that for any $\tilde{\bf Z}^{(n)} \in  \mathcal{G}^{(n)}$ and $\tilde{M}^{(n)} \in \mathcal{M}^{(n)}$ we have:

\begin{enumerate}[(a)]
\item Let $\bf Z^{(n)}_1, \bf Z^{(n)}_2$ be two independent samples from $f_{\mc A(\tilde{M}^{(n)},\tilde{\bf Z}^{(n)},h_n)}$. 
Then 
\be
\bb P[\|\Pi_{\bf Z^{(n)}_1} - \Pi_{\bf Z^{(n)}_2} \|_{\TV} > C_1 n^{-c_1}] > 1-\omega_1(n).
\ee
\item Let $\bf Z^{(n)}_1, \bf Z^{(n)}_2$ be two independent samples from $f_{\mc A(\tilde{M}^{(n)},\tilde{\bf Z}^{(n)},h_n)}$. Then 
\be
\bb P\left[ \|\mu_{\bf Z^{(n)}_1} - \Pi_{\bf Z^{(n)}_2}\|_{\TV} < \frac14 \|\Pi_{\bf Z^{(n)}_1} - \Pi_{\bf Z^{(n)}_2} \|_{\TV}  \right] > 1- \omega_2(n).
\ee
\item Let $\bf Z^{(n)}_{1} \sim f_{\mc A(\tilde{M}^{(n)},\tilde{\bf Z}^{(n)},h_n)}$. There exists a family of measures $\{\mu^*_{h_n}\}_{h_n \in \tilde{\mc H}^{(n)}}$ so that 
\be \label{IneqCondWarmStart}
\bb P\left[ \|\mu^*_{H(\bf Z^{(n)}_{1})} K^t - \mu_{\bf Z^{(n)}_1}\|_{\TV} \leq C_2 n^{c_2} (1-\lambda_{\bf Z^{(n)}_{1}})^t  \right] > 1-\omega_3(n).
\ee
\item  For $X_0 \sim \mu^*_{H({\bf Z^{(n)}})}$ and $U_0,U_1,\ldots,U_{s_n} \stackrel{iid}{\sim} \nu$, we have 
\be
\bb P[ M_{s_{n}}(\bf Z^{(n)},X_0,U_0,\ldots,U_{s_n}) \in \mathcal{M}^{(n)} ] > 1-\frac{1}{32} C_1 n^{-c_1}.
\ee
\end{enumerate}
\end{assumption}

We briefly comment on how one might check these assumptions, with more detailed discussion of the consequences in Section \ref{SecTrilLast}:

\begin{enumerate}[(a)]
\item This assumption depends only on the density $f$ of the data-generating process and the control variates $H$ accessible to the algorithm. In particular, it does not depend on the algorithm.
\item When this assumption fails, the stationary distribution of the Markov chain is about as close to the posterior associated with a ``random" dataset in $f_{\mc A(\tilde{S},\tilde{\bf Z^{(1)}},h)}$ as it is to the posterior associated with the observed dataset $\bf Z^{(1)}.$ As discussed next to Equation \eqref{eq:IneqLargeFluctuationInformal}, when this occurs we informally say that the extra information beyond the first $m$ datapoints does not significantly improve MCMC accuracy.
\item This says that (with high probability) the chain is geometrically ergodic, \textit{and furthermore} that the associated constants don't grow too quickly in $n$. Of course, when this condition fails, the ``burn-in" time for the chain grows quickly in $n$ and so the overall running time of the MCMC chain grows super-linearly in $n$. Thus, we typically view chains that fail this condition to be too slow to be useful in the present context.

Proving that this condition actually does hold for a specific algorithm often takes an entire paper, and so is beyond the scope of this paper. See Appendices \ref{AppSecWarmStart} and \ref{AppWarmStart} for more detailed discussion and references. 
\item When $\mathcal{M}^{(n)}$ consists of all sets below a certain size, this can be easily verified for many subsampling algorithms where the selection of data points at each iteration is ``close to'' weighted random sampling with bounded weights. This includes both of our pedagogical examples, which \textit{exactly} use uniform random sampling of minibatches.\end{enumerate}

Note that Parts (c) and (d) must be checked for each algorithm. Part (d) holds for both of our pedagogical examples when minibatches are small; see Proposition \ref{PropIneqChoiceOfSAust}. Part (c) depends on the algorithm \textit{and} the target, and so it seems difficult to obtain fully general results. See Appendix \ref{AppSecWarmStart} for a survey of several results that can be used to prove this condition, as well as a worked proof that this condition holds for SGLD and references to a proof that this condition holds for simple perturbations of the unadjusted Langevin algorithm. 

\begin{remark} [Density and Warm Starts] 
We note that, by the comment immediately following  Inequality \eqref{eq:WarmStart}, Part (c) is implied by the existence of a measure $\mu^*_{H(\bf Z^{(n)}_{1})}$ that satisfies 
\be 
\left\| \frac{d \mu^*_{H(\bf Z^{(n)}_{1})}}{d \mu_{\bf Z^{(n)}_1}} \right\|_{2}^{2} \leq C_{2} n^{c_{2}}.
\ee 

This condition is difficult to verify in great generality, as $\mu_{\bf Z^{(n)}_1}$ may be very far from the posterior distribution $\Pi_{\bf Z^{(n)}_1}$ of interest. For this reason, we introduce the following closely-related condition: 

\be \label{IneqWarmStartFake}
\left\| \frac{d \mu^*_{H(\bf Z^{(n)}_{1})}}{d \Pi_{\bf Z^{(n)}_1}}  \right\|_{2}^{2} \leq C_{2} n^{c_{2}}.
\ee 

This does \textit{not} imply Part (c) in general. See Appendix \ref{AppSecWarmStart} for conditions that relate this to Part (c). We also note that, if Inequality \eqref{IneqWarmStartFake} holds but Part (c) does not, this immediately implies that the posterior distribution of interest $\Pi_{\bf Z^{(n)}_1}$ is far from the stationary distribution of the Markov chain  $\mu_{\bf Z^{(n)}_1}$; this is of course undesirable in and of itself.
\end{remark}

\begin{remark} [Choice of Metric] \label{RemChoiceMet}
The choice of the total variation metric in Assumption \ref{ass:LargeFluctuations} is very important. Notice that the stronger the metric used in Assumption \ref{ass:LargeFluctuations}, the weaker the assumption, since it is ``easier'' for two distributions to be far away in a strong metric than in a weak one. If total variation were replaced by a weaker metric, such as the Prokhorov metric, the assumption would hold for far fewer cases of interest.
\end{remark}

\begin{remark} [Extra Control Variates and Verifying Assumptions] \label{RemExtraControl}

At first glance, it seems that Inequality \eqref{IneqCondWarmStart} from Assumption \ref{ass:LargeFluctuations} must fail for algorithms that don't use any control variates - it is typically impossible to satisfy Inequality \eqref{IneqCondWarmStart} uniformly over all datasets. Fortunately, this is often easy to fix by adding in additional ``unused" control variates. More precisely: in our current framework, the collection of control variates $T$ must include all control variates that the Markov chain actually accesses when it is implemented. However, our framework also allows us to add any number of additional control variates to this collection, even if they are not actually used to sample from $K$. These ``purely formal" control variates can be used to prove stronger theorems, especially in the case that the actually-used control variates are very uninformative.

By construction, these ``extra" control variates will not influence the algorithm. They also do not need be computationally tractable. However, well-chosen extra control variates can allow us to achieve Inequality \eqref{IneqCondWarmStart} and the rest of Assumption \ref{ass:LargeFluctuations} simultaneously. \textit{In other words,} when verifying our assumptions, we are always allowed to simply throw in more control variates than the algorithm actually takes advantage of - the results will still stand. As a concrete application of this idea, Theorem \ref{AppWarmStart} implies that adding the MLE is sufficient to obtain Inequality \eqref{IneqCondWarmStart} with reasonable constants for a wide variety of models. See Appendix \ref{AppSubsecGrid} for a simple construction of grid-based control variates that roughly mimic quadrature rules and which can be useful in this context.
\end{remark}

We also make the assumption that our data $Z^{(n)}$ lands in the ``good" sets with high probability:

\begin{assumption} \label{ass:GoodSetHP}
There exists a sequence $\omega_4(n) \searrow 0$ such that for $\mathbf Z^{(n)} \sim f^n$
\be
\bb P[H(\mathbf Z^{(n)}) \in \tilde{\mc H}^{(n)}, \mathbf Z^{(n)} \in \tilde{\mc Z}^{(n)}] > 1-\omega_4(n).
\ee
\end{assumption}

\begin{remark} [Why Do We Need Sets $\tilde{\mc H}^{(n)}$, $\tilde{\mc Z}^{(n)}$?]

Note that Assumption \ref{ass:LargeFluctuations}(a) will often fail if we allow \textit{any} value for the data $Z^{(n)}$. To see this, consider a very light-tailed distribution $f$ and the test statistic $T(Z) = \frac{1}{n} \sum_{i=1}^{n} z_i$. As $t$ and $n$ become very large, conditioning on the event $\{ T(Z) = t \}$ becomes more and more similar to conditioning on the event that $z_i \approx t$ for all $1 \leq i \leq n$. If $p$ is sufficiently light-tailed, the error in the approximate equality $z_1 \approx t$ can go to 0 very quickly. In effect, this means that the sample average $T(Z)$ almost determines the full vector $Z$ and thus acts (quantitatively) as a sufficient statistic. See Section \ref{SecSubsecCounter} for a simple explicit example that illustrates this general phenomenon.

In practice, we expect it to be easy to find a ``good" set that contains all $n$ data points with high probability, while remaining fairly small. For example, if the data is taken from a standard Gaussian and we have no control variates, we might choose \textit{e.g.} $\tilde{\mc Z}^{(n)} = (-10 \sqrt{\log(n)}, 10 \sqrt{\log(n)})$. We do not expect Assumption \ref{ass:GoodSetHP} to be a large obstacle to the analysis of realistic algorithms and datasets - it simply means excluding datasets that are extremely implausible under the model of interest.
\end{remark}

\subsection{Main Abstract Theorem}

Our main abstract theorem is:

\begin{theorem} \label{thm:SpecGapBound}
Let $\bf Z^{(n)} \sim f^{n}$ be a data sequence and $K_{\bf Z^{(n)}}$ be the corresponding sequence of transition kernels of a subsampling algorithm. Let $\{\tilde m_n\},\{s_n\}$ be sequences and $\{c_j,C_j\}_{j=0,1,2}$ be constants for which the algorithm, its invariant measure, and the data distribution satisfy Assumptions \ref{ass:LargeFluctuations} and \ref{ass:GoodSetHP}. Then there exists a sequence $\omega(n) \searrow 0$ such that for all large $n$,
\be
-\log(1 - \lambda_{\bf Z^{(n)}})  \leq \frac{c_1 + c_4}{2 s_n} \log(n)
\ee
with probability at least $1-\omega(n)$ with respect to the distribution of $\bf Z^{(n)}$.
\end{theorem}

\begin{proof}
See Appendix \ref{SecAppPrfSpecGapBd}.
\end{proof}

\subsection{Application for Uniform Minibatch Choice}

We now return to Equation \eqref{eq:CompCost} in light of Theorem \ref{thm:SpecGapBound}:
\be
\Cost(K_{\bf Z}) = \frac{m}{\tau_{\bf Z}(m) \lambda_{\bf Z}}.
\ee
When applying Theorem \ref{thm:SpecGapBound}, note that since $s_n > C_0 n^{c_0}$ for all large $n$, we have that
\be
\lambda_{\bf Z} < \frac{c_1 + c_4}{4 s_n} \log(n),
\ee
with high probability for all large $n$. Thus,
\be \label{IneqCorrThmSimple}
\Cost(K_{\bf Z}) \ge C \frac{n s_n}{\tau_{\bf Z} \log(n) }
\ee
with high probability for all large $n$.

With the expression in \eqref{IneqCorrThmSimple}, we can now interpret Theorem \ref{thm:SpecGapBound} in the simple case that every step of the algorithm uses a minibatch of exactly $r = r(n)$ datapoints, chosen uniformly at random. Consider a sequence $\tilde{m}_{n}$ to be determined and define the sets $\mathcal{M}^{(n)} = \{ m \subset \{1,2,\ldots,n\} \, : \, |m| \leq \tilde{m}_{n}\}$.  Consider constants $\{C_j, c_j\}$ for $j=0,1,2$ and rates $\omega_j$ for $j=1,2,3,4$ so that parts (a) and (c) of Assumption  \ref{ass:LargeFluctuations}, and all of Assumption  \ref{ass:GoodSetHP}, hold for \textit{some} choice of $\{ \tilde m_n \}$. For fixed $n$, let $M_{\max} = M_{\max}(n)$ be the \textit{largest} value of $\tilde{m}_{n}$ for which Assumption \ref{ass:GoodSetHP} holds for these choices of $C_j,c_j,\omega_j$. Note that these parts of Assumption \ref{ass:LargeFluctuations} depend only on the choice of $\tilde{m}$ and the constants fixed so far, so this choice is well-defined. 

Now, for each $n \in \mathbb{N}$, define 
\be 
M(n) := 1 \vee   \max\{ 1 \le M \le M_{\max}(n) \, : \, \text{ Assumption \ref{ass:LargeFluctuations}(b) holds for } \tilde{m}_{n} = M \}.
\ee 
Finally, set $\tilde m_n = M(n)$ and $\tilde s_n = \lfloor \frac{M(n)}{r} \rfloor - 1.$ In our simple minibatch setting, Assumption \ref{ass:LargeFluctuations}(d) holds for these choices. Thus, with this choice of $\tilde m_n$, all parts of Assumption \ref{ass:LargeFluctuations} hold. We now have the following corollary of Theorem \ref{thm:SpecGapBound}:

\begin{corollary} \label{CorNFLCostCor}
Let Assumption \ref{ass:GoodSetHP} hold, and choose the sequences $\tilde m_n, s_n, \omega(n)$ as immediately above. Then there exists a constant $C^* = C^*(c_{0},\ldots,C_{0},\ldots) > 0$ so that
\be 
\Cost(K_{\bf Z}) \ge C \frac{M(n)}{\log(n)}
\ee 
with probability at least $1-\omega(n)$.
\end{corollary}

\begin{proof}
As shown immediately after the statement of Theorem \ref{thm:SpecGapBound} in Equation \eqref{IneqCorrThmSimple}, with probability at least $1-\omega(n)$
\be 
\Cost(K_{\bf Z}) \ge C \frac{n s_n}{\tau_{\bf Z} \log(n) }
\ee
for some constant $C > 0$.

We have chosen $s_{n} = \lfloor \frac{M(n)}{r} \rfloor - 1.$ Since we are describing a minibatch algorithm that uses exactly $r$ datapoints per step, we know that $\tau_{\bf Z} \geq \frac{n}{r}$. Thus, we have (for a slightly different constant $C^* > 0$):
\be 
\Cost(K_{\bf Z}) \ge C^* \frac{M(n)}{\log(n)}.
\ee
\end{proof}

We now compare this to an alternative strategy of throwing out all but $M(n)$ datapoints after computing control variates, and then running standard random walk Metropolis-Hastings using only the control variates and these $M(n)$ data, with no subsampling at each iteration. When the target is a posterior under a generalized linear model, \cite{belloni2009computational} show that random walk Metropolis-Hastings typically has spectral gap that is bounded away from zero when $d$ is fixed with $n \to \infty$. Since it has a per step cost of $M(n)$, we obtain
\be 
\Cost(\tilde K_{\bf Z}) \le C' M(n)
\ee 
for some constant $C'< \infty$, with $\tilde K_{\bf Z}$ the transition kernel of the Metropolis-Hastings algorithm that only uses control variates and $M(n)$ data points. Thus, in this regime, no subsampling algorithm can pick up an improvement of more than $\mc O(\log(n))$ over the n\"{a}ive algorithm that uses a single subset of the data, \textit{even when accounting for control variates.}

\subsection{A Subsampling Trilemma from Assumption \ref{ass:LargeFluctuations}}\label{SecTrilLast}

We interpret our result as saying that subsampling MCMC can't be reliably better than all of a small number of simple competitor algorithms. We go through Assumption \ref{ass:LargeFluctuations}  and explain this interpretation in more detail.

In the rest of this section, we assume that $\mathcal{M}^{(n)}$ is of the form 
\be \label{EqSimpleCalMSet}
\mathcal{M}^{(n)} = \{ m \subset [n] \, : \, |m| \leq \tilde{m}_{n}\}
\ee 
for some sequence $\tilde{m}_{n}$ that is to be determined. This means that Assumption \ref{ass:LargeFluctuations}(d) is typically easy to verify. We now consider the three ``bad" things can happen (along with a technical assumption):

\begin{enumerate}[(a)]
    \item \textbf{Remaining data is irrelevant:} If Assumption \ref{ass:LargeFluctuations}(a) fails for some value of $c_{1}$, then the data contained in the control variate $H(\bf{z}^{(n)})$ and the subsample $\{ \bf{z}^{(n)}_{i}\}_{i \in \tilde{M}^{(n)}}$ contains essentially all the information in the full dataset $\bf{z}^{(n)}$ up to some error $O(n^{-c_{1}})$ (as measured in the very strong total variation metric).
    
    If this assumption fails for some value of $c_{1}$ for which an error of size $O(n^{-c_{1}})$ is acceptable, then we lose only a negligible amount of information by throwing out the remaining data and running an algorithm only on the remaining subsample.
    
    In the remaining notes, we assume that Assumption \ref{ass:LargeFluctuations}(a) holds for all $c_{1} \geq a_{1}$ sufficiently large, and that an error of size $O(n^{-a_{1}})$ is unacceptably large. In practice, we expect that $a_{1} = 0.5$ is a reasonable threshold here.
    
    \item \textbf{Subsampling algorithm is not accurate.} We consider the case that Assumption \ref{ass:LargeFluctuations}(a) holds for all $c_{1} \geq a_{1}$ but Assumptions \ref{ass:LargeFluctuations}(a,b) jointly fail. In this case, the distance between the posterior distribution and the stationary distribution of the Markov chain is already of size $\Omega(n^{-a_{1}})$, which we have already decided to be unacceptably large.
    
    In the remaining notes, we assume that Assumption \ref{ass:LargeFluctuations}(a,b) hold for some $c_{1} > 0$, and that an error of size $O(n^{-c_{1}})$ is unacceptably large.
        \item \textbf{Technical assumptions.} Assumption \ref{ass:LargeFluctuations}(c) says that the chain is geometrically ergodic, with associated constants not growing too quickly. This assumption must be checked on an algorithm-by-algorithm basis, and it can require substantial work. We don't focus on this assumption, but note that when it fails, the burn-in time of the chain grows super-linearly in $n$, which already makes it slower than the naive Metropolis-Hastings algorithm we use as a baseline.
    \item \textbf{Subsampling algorithm is slow.}  When all the assumptions hold, our main conclusion is that the spectral gap is not large and thus the subsampling algorithm is ``slow." More precisely, its computational cost is not much better than an algorithm that simply throws away all but $\tilde{m}_{n}$ datapoints before running any MCMC algorithm at all.
\end{enumerate}

\section{Non-Examples: Models and Control Variates} \label{SecSubsecCounter}

Our main fluctuation assumption, Assumption \ref{ass:LargeFluctuations}, does \textit{not} always hold. One typical ``problem" is that some low-dimensional summary statistic is quantitatively close to being a sufficient statistic. For genuine sufficient statistics, the ``problem" is clear:

\begin{example} [Sufficient Statistics]
Consider a Gaussian model
\be
p(\theta \mid y) \propto e^{-\theta^{2}} \prod_{i=1}^{n} e^{-(y_i - \theta)^{2}}
\ee
and the usual sufficient statistic $T_1(Y) = n^{-1} \sum_{i=1}^n y_i$. Defining $\mathcal{A}$ as in Part (1) of Assumption \ref{ass:LargeFluctuations},
\be
\P\left[ \| p(\cdot \mid Y^{(1)} ) - p(\cdot | Y^{(2)}) \|_{\mathrm{TV}} >0 \mid \mc A \right] = 0.
\ee
\end{example}

The same ``problem" can appear in a more subtle form: even models that exhibit large posterior fluctuations for typical data sets may exhibit very small fluctuations for extremely atypical datasets. This is the situation that Assumption \ref{ass:GoodSetHP} is meant to avoid. We give a somewhat cartoonish example to avoid difficult computations:

\begin{example} [Exponentially Unlikely Data]
Consider the light-tailed model and prior
\be
p(y \mid \theta) \propto \prod_{i=1}^n e^{-\theta |y_i|}, \quad p_0(\theta) = e^{-\theta}
\ee
with $\theta \in (0,\infty)$ and $y_i \in \bb R$. Let $Y^{(1,n)}= (Y^{(1)}_1,\ldots,Y^{(1)}_n)$, $Y^{(2,n)}= (Y^{(2)}_1,\ldots,Y^{(2)}_n)$ be two datasets drawn from the model with ``true'' parameter $\theta_{0} = \frac{1}{2}$. Define the summary statistic $T_n(x_1,\ldots,x_n) = \max(x_{1},\ldots,x_{n})$; note that this is clearly \textit{not} a sufficient statistic.

Define $S_{n,k} = \sum_{i=1}^n |Y^{(k)}_i|$, $k \in \{1,2\}$ and $\delta_{n} = e^{n^{2}} t(n)^{-1}$ for some sequence $t(n) \geq e^{n^4}$. Conditional on the sequence of events $\{T_{n}(Y^{(k,n)}) = t(n)\}$, standard large-deviation results tell us that the event
\be \label{IneqCondSum}
\left\{ \left| \frac{S_{n,k}}{t(n)} - 1 \right| \leq \delta_{n} \right\}
\ee
holds w.e.p. On this event, we also have the event
\be
\{ p(\theta | Y^{(k,n)}) = t(n) (1 + O(\delta_{n})) e^{-\theta t(n) (1 + O(\delta_{n}))} \}.
\ee
Directly calculating the total variation distance, we find that on this event we also have
\be \label{IneqCondTV}
\{\| p(\cdot \mid Y^{(1,n)}) - p(\cdot \mid Y^{(2,n)}) \|_{\TV} = e^{-\Omega(n)}\}.
\ee

Thus, we have seen that, conditional on the sequence of events $\{T_{n}(Y^{(k,n)}) = t(n)\}$, the sequence of events \eqref{IneqCondTV} holds w.e.p. Of course, the event $\{T_{n}(Y^{(k,n)}) \geq e^{n^4}\}$ has (tail) probability zero, so the introduction of the ``good'' set allows us to essentially discount the existence of such ``bad'' sequences of data in all of our proofs.
\end{example}

\begin{example} [Densely-Observed Processes]

So far, all of our data $\{ (X_{i},Y_{i})\}_{i=1}^{n}$ has been i.i.d. In this setting, we generally expect each datapoint to have a substantial contribution to the posterior distribution. This can easily fail to occur for far-from-i.i.d. data. To give an informal example, consider dense observations $\left\{ \left(\frac{i}{n}, f\left(\frac{i}{n}\right) \right)\right\}_{i=1}^{n}$ from a Gaussian process on $[0,1]$. Under reasonable conditions on the prior and data-generating process, the contribution of the middle data-point $(0.5, f(0.5))$ to the posterior will often decay super-polynomially quickly.

As a consequence, we expect that subsampling MCMC algorithms, such as the recent clever work in \cite{salomone2019spectral}, may be able to give accurate inferences based on very small subsamples in this regime.
\end{example}

Of course, in practice this is not really a problem and we expect that subsampling MCMC would essentially be avoided in the situations described in these examples: a practicing statistician could compute very accurate posterior distributions by ignoring most of the the data set and using only a fixed subsample and the summary statistic $T(Z)$.

We  recall here one caveat that might be an interesting avenue of research. As pointed out in \textit{e.g.} \cite{montanari2015computational}, there are realistic situations in which computations using low-dimensional sufficient statistics are much harder than computations with the original high-dimensional data set. On the other hand, the application in \cite{montanari2015computational} does appear to lend itself to an efficient subsampling algorithm. However, in this application the data are not independently generated as they are in our examples, as well as most of those in the approximate MCMC literature. Thus, it is plausible to us that certain subsampling algorithms would converge extremely quickly (as they avoid our main results) while remaining practical (as direct computations with sufficient statistics are not feasible), though perhaps only for data generation processes that differ from those considered here. We are not aware of research in this direction.

\section{Examples: Minibatch Algorithms} \label{SecApplMCMC}

In this section, we discuss some specific classes of subsampling algorithms and use them to illustrate what our assumptions mean for specific popular algorithms. Although we include a few computations in this section, they are included only as illustrations. We have made no attempt to find the weakest possible conditions under which our conclusions hold, and have favored including a variety of simple calculations over optimal constants.

We note that \textit{most} of our assumptions pertain to the \textit{posterior} of interest rather than the \textit{specific algorithm.} The main difficulty in applying our results to specific algorithms is verifying Assumption \ref{ass:LargeFluctuations}(d) for a reasonable sequence $\{ s^{(n)}\}$. In the current section, we use some simple algorithms to illustrate how this can be done. Detailing more complicated algorithms is beyond the scope of the current paper, but Section \ref{AppSubsecGrid} in the appendix includes a broadly-applicable construction and bound that we would suggest for readers interested in obtaining sharp estimates for such algorithms.

\subsection{Inexact Minibatching MCMC} \label{SecAustMCMC}

This paper was inspired by \cite{korattikara2014austerity} and the outpouring of related work. We begin by writing down a very simple and fairly generic subsampling algorithm that includes as special cases many subsampling algorithms in the literature, including Algorithm 1 of \cite{korattikara2014austerity}, the stochastic gradient Langevin algorithm of \cite{welling2011bayesian}, some but not all of the pseudomarginal-based algorithms of \cite{quiroz2015speeding}, and many others.

For a finite set $S$, define $C_{k}(S)$ to be the uniform measure on size-$k$ subsets of $S$ and $2^{S}$ to be the collection of all subsets of $S$. Fix a model with parameter space $\Theta \subset \mathbb{R}^{m}$ and a dataset $X = (X_{1},\ldots,X_{n})$. The algorithm itself is written in terms of three functions $F_{1} \, : \, \Theta \times [0,1]^{2} \times 2^{X} \mapsto [0,1]$, $F_{2} \, : \, \Theta \times [0,1] \times 2^{X} \mapsto \Theta$, and $F_{3} \, : \, \Theta^{2} \times [0,1] \times 2^{X} \mapsto \{0,1\}$ and two parameters $k \in \{1,2,\ldots,n\}$ and $\alpha > 0$, which will be explained below.

Our ``generic" subsampling algorithm is defined by Algorithm \ref{GenAusterity}, which describes a forward-mapping representation for sampling from the kernel $K(\theta,\cdot)$.

\begin{algorithm}
\caption{Generic subsampling} \label{GenAusterity}
 \begin{algorithmic}[1]
    \State Sample driving randomness $U = (U_{1}, U_{2}) \sim \mathrm{Unif}([0,1]^{2})$ and initial set $Y \sim C_{k}(X)$.
    \While{$F_{1}(\theta,U,Y) > \delta$}
\State Sample $Y^{\ast} \sim C_{k}(X \backslash Y)$ and set $Y = Y \cup Y^{\ast}$.
\EndWhile
\State Propose $\theta^{\ast} = F_{2}(\theta,U_{1},Y)$.
\If{$F_{3}(\theta,\theta^{\ast},U_{2},Y) =0$}
\State Return $\theta^{\ast}$.
\Else
\State Return $\theta$.
\EndIf
 \end{algorithmic}
\end{algorithm}

We informally relate this to the ``usual" and less-abstract ways to write down MCMC algorithms:

\begin{enumerate}
\item Metropolis-Hastings algorithms are usually written with two sources of driving randomness; $U_{1}$ is the driving randomness for \textit{proposing} the next state, and $U_{2}$ is the driving randomness for \textit{accepting/rejecting} this proposal.
\item The function $F_{1}$ is a measurement of how well $Y$ approximates the entire dataset $X$, and $\delta$ is a threshold for how good this approximation must be.
\item The function $F_{2}$ is the forward-mapping representation for the proposal kernel.
\item The function $F_{3}$ approximates the usual accept/reject rule with Metropolis-Hastings acceptance probability $\alpha(\theta,\theta^*)$.

\end{enumerate}

This algorithm includes Examples \ref{ExSGLD} and  \ref{ExSPMMH} as special cases. Many variations of these algorithms are also included. For one example, one could incorporate control variates as long as they did not affect the chance of choosing a given minibatch.\footnote{Weighted subsampling, as in \cite{maire2019informed}, is \textit{not} allowed by Algorithm \ref{GenAusterity}. The proof of Proposition \ref{PropIneqChoiceOfSAust} below holds with only minor changes when the weights are not too far from 1. On the other hand, results such as \cite{huggins2016coresets} imply that our qualitative conclusions do \textit{not} hold if most weights are replaced by 0 in a clever way. } For another, note that it is common to take larger subsamples in the ``tails" of a distribution in order to encourage stabilization (see \textit{e.g.} \cite{RePEc:eee:econom:v:171:y:2012:i:2:p:134-151}). This is covered by allowing $F_{1}(\theta,U,Y)$ to be large when $\theta$ is in the tails of a distribution.

We would like to study when Algorithm \ref{GenAusterity} satisfies Part (d) of Assumption \ref{ass:LargeFluctuations}. This is quite a broad question, so we consider the important but simple case where all datapoints influence the algorithm \textit{and} it is rare to use very large subsamples:

\begin{assumptions} \label{EqSimpleAusterityAssumption}
Assume that the subsampling algorithm satisfies
\begin{enumerate}
\item  The output of the algorithm depends on every component of the sampled $Y$ for every value of $U$.
\item For all $n$, the distribution of the number of steps $N$ in the``while" loops has CDF $F^{(N)}$ that satisfies
\be
1 - F^{(N)}(x) \leq e^{-\lambda_{2}^{-1}(x-\lambda_{1})}
\ee
for some constants $0 < \lambda_{1},\lambda_{2} < \infty$. In other words, $N$ is stochastically dominated by a constant $\lambda_{1}$ plus an exponential random variable with mean $\lambda_{2}$. Define $\lambda = \lambda_{1} + \lambda_{2}$. \footnote{Note that this assumption allows for a deterministic sample size.}
\end{enumerate}
\end{assumptions}

As  Algorithm \ref{GenAusterity} is written above, it seems possible that we ``use" every datapoint at every iteration. Fortunately, this problem occurs only because Algorithm \ref{GenAusterity} is a rather unsuitable representation of the kernel. This can be remedied by choosing another kernel that represents the same Markov chain but gives better constants:

\begin{prop} \label{PropIneqChoiceOfSAust}
Let Assumptions \ref{EqSimpleAusterityAssumption} hold and fix constants $0 < \epsilon <1$ and $m \in \mathbb{N}$. Let $\mathcal{M}^{(n)} = \{S \subset [n-m] \}$. Then there exists a representation of the Markov chain defined in Algorithm \ref{GenAusterity} for which Part (d) of Assumption \ref{ass:LargeFluctuations} holds with
\be \label{IneqChoiceOfSAust}
s^{(n)} = C n \log(n),
\ee
where $C = C(\epsilon,\lambda,k,m) \equiv \frac{1 - \epsilon}{\lambda \, k \, m}$.

\end{prop}
\begin{proof}
See Appendix \ref{SecAltCons} for a proof, including an explicit construction of the forward-mapping representation.
\end{proof}

Proposition \ref{PropIneqChoiceOfSAust} is important for two reasons. First, it implies a lower bound on the computational cost that is linear  in $n$, with no logarithmic factors. Second, it illustrates the fact that the natural representation of an algorithm may not give the best estimates. 

\subsection{Exact Minibatching MCMC}  \label{SecDataAug}

To our knowledge, the first \textit{exact-approximate subsampling algorithm} that applied in some generality was the FlyMC algorithm of \cite{maclaurin2014firefly}. Theorem \ref{thm:SpecGapBound} applies directly to this algorithm, showing that the spectral gap of this algorithm must be rather small for problems such as logistic regression. We give a few additional comments on how our results relate to this algorithm. Throughout, we mimic the notation of \cite{maclaurin2014firefly}, but to save space we don't restate the full algorithm.

\subsubsection{Lower and Upper bound functions}
The algorithm of \cite{maclaurin2014firefly} involves the computation of lower and upper bound functions $L_{n},B_{n}$ that satisfy $L_n(\theta) \leq p(\theta \mid Z) \leq B_n(\theta)$. When applying Theorem \ref{thm:SpecGapBound}, we view $L_{n}$ and $B_{n}$ as being further pre-computed control variates. This presents no difficulty for our framework.

\subsubsection{Convergence of Full and Reduced Markov Chains}
Theorem \ref{thm:SpecGapBound} applies to the mixing of the full FlyMC Markov chain. However, FlyMC is a data-augmentation algorithm; the associated Markov chain may be written in the form $\{ (\theta_t, \omega_t) \}_{t \in \mathbb{N}}$, where $\mathcal{L}(\theta_{t})$ converges to the posterior distribution $p(\theta \mid z)$ and $\omega_{t} \in \{0,1\}^{n}$ is a collection of auxiliary variables that is generally ignored after the computation is complete.

For statistical applications, we are only interested in the mixing of the sequence $\{ \theta_{t} \}_{t \in \mathbb{N}}$. It is natural to ask: could FlyMC as a whole mix slowly while $\{ \theta_{t} \}_{t \in \mathbb{N}}$ mixes quickly? If so, this would make our lower bound on the mixing rate much less relevant and allow for the possibility of an efficient FlyMC algorithm. Unfortunately, the answer is often \textbf{no} in the case of FlyMC.

This is easiest to see in the special case that the parameter \textit{ResampleFraction} $=1$ in their Algorithm 1, so that the entire subsample is refreshed at each iteration. In this case, the first coordinate $\{ \theta_{t} \}_{t \in \mathbb{N}}$ is in fact a Markov chain. By Theorem 3.3 of \cite{liu1994covariance} the spectral gap of this chain is equal to the spectral gap of the joint chain $\{ (\theta_{t}, \omega_{t}) \}_{t \in \mathbb{N}}$, and so it cannot mix more quickly. The same argument gives the same conclusion for the improved Algorithm 2 of \cite{maclaurin2014firefly}.

\subsubsection{Typical covering times}
In Theorem \ref{ThmExampleTorpid}, the bound on the relaxation time depends on the (random) time $\tau_z$ required to see almost all the dataset. This number is often substantially larger than the number of datapoints $n$ divided by the mean number of datapoints seen per step of the Markov chain.

We give a simple illustration in the context of \cite{maclaurin2014firefly}. Assume that the lower and upper bound functions $L_{n}, B_{n}$ used in Algorithm 1 of \cite{maclaurin2014firefly} satisfy
\be \label{IneqUniformBoundednessLowerUpper}
\epsilon < \frac{L_{n}(\theta)}{B_{n}(\theta)} < 1 - \epsilon
\ee
for some $0 < \epsilon < \frac{1}{4}$ and all $\theta \in \Theta$. Then, by the well-known coupon-collector calculation (see \cite{ErRo61Coupon}),
\be \label{IneqUniformBoundednessLowerUpperConc}
\frac{1}{2} \log(n) \leq \tau_{z} \leq \frac{2}{\epsilon} \log(n)
\ee
holds w.e.p., while the average number of datapoints used per transition is on the order of $cn$ for some $\epsilon \leq c \leq 1- \epsilon$. Thus, returning to Theorem \ref{ThmExampleTorpid}, we see that in fact the ratio of relaxation time to average number of datapoints per step scales linearly in $n$. This is slightly better than the naive lower bound of $\Omega \left(\frac{n}{\log(n)} \right)$ on the scaling rate.

The uniform bound \eqref{IneqUniformBoundednessLowerUpper} is quite strong, and will not hold in all applications. Improving on this is far beyond the scope of this paper, but we note that well-known variants of the coupon collector problem tell us that the conclusion \eqref{IneqUniformBoundednessLowerUpperConc} holds under much weaker conditions than \eqref{IneqUniformBoundednessLowerUpper}. In particular, having \eqref{IneqUniformBoundednessLowerUpper} for a positive fraction of the values of $\theta$ is enough, and this covers many more cases of practical interest.

\section{Discussion} \label{SecDisc}

We have given various ``no-free-lunch'' theorems for subsampling MCMC, showing that subsampling MCMC can't give ``large'' improvements under some conditions. We summarize our main message: for reversible MCMC algorithms and statistical models such as GLM, subsampling and approximate MCMC don't ``work" - you either don't get much speedup (compared to the associated exact chain) or have estimates that aren't much better (compared to doing inference based on just your control variates and a single fixed subsample).

To us this feels like an argument for focusing primarily on control variates over attempts to do clever subsampling schemes, \textit{when our results apply.} In this section, we discuss several situations in which our results \textit{don't} apply. We also lay out the most important questions that are left open by this paper.

\subsection{Free Lunches for Highly Nonreversible Chains} \label{SubsecHighlyNonrev}

Our main result for reversible chains, Theorem \ref{thm:SpecGapBound}, is easy to interpret: if you choose a family of subsampling Markov chains in advance, most of the associated Monte Carlo averages will have large asymptotic variance. It is possible to obtain analogous results for nonreversible chains, but the most similar results are much harder to interpret, and we consider this the \textbf{most important} open problem left by this paper.

The basic problem is that, for reversible chains, the following are all characterized by the spectral gap:

\begin{enumerate}
\item The total variation mixing time from a ``warm start'' (see e.g. Proposition 3.4 of \cite{paulin2015concentration}).
\item The medium-time error of an estimator for the ``worst'' test function (see \textit{e.g.} Proposition 3.8 of \cite{paulin2015concentration}, though other results along these lines exist).
\item The asymptotic variance of an estimator for the ``worst'' test function (see \textit{e.g.} \cite{jones2004markov} and the references therein).
\end{enumerate}

Since the three different notions of mixing are all characterized by a single quantity, we can get a bound on the spectral gap by taking advantage of the first property and know that it is relevant to the accuracy of our Monte Carlo estimators by using the third property.

In the nonreversible case, these three notions of mixing may be very different. Following the proof of Theorem \ref{thm:SpecGapBound}, it is straightforward to find an analogous result for non-reversible chains with the spectral gap replaced by the pseudospectral gap (replacing the geometric ergodicity condition by the analogous bound in Proposition 3.4 of \cite{paulin2015concentration}). Qualitatively similar results are easy to obtain for various other summary measures of a nonreversible Markov chain's performance, using essentially the same methods. See \textit{e.g.}  \cite{kontoyiannis2012geometric}, \cite{kontoyiannis2005large}, \cite{kontoyiannis2003spectral}, \cite{chatterjee2023spectralgapnonreversiblemarkov} for relationships between the spectral theory and long-time asymptotics of nonreversible Markov chains, and \cite{montenegro2006mathematical}, \cite{choi2017metropolis} for some discussion of the short- and medium-time mixing properties.

Our understanding is that the pseudospectral gap is a good proxy for the convergence rate of Monte Carlo ergodic averages of many simple ``quickly-mixing'' nonreversible chains, and we conjecture that this is also true for some proposed nonreversible subsampling MCMC chains. However, we emphasize that we do \textit{not} know if this is true, and there is some strong evidence \textit{against} this heuristic. For example,  there are many examples in the literature of nonreversible chains for which these two quantities are very different. Even worse for our assertion, some of these can be cast as subsampling chains! As an extreme example, the herded Gibbs sampler of \cite{chen2016herded} is used to sample from a Markov random field on a graph. If we imagine a sequence of graphs for which the number of vertices goes to infinity while the maximum degree remains bounded, these Gibbs samplers use only $O(1)$ data points per step, and so we can obtain a nontrivial bound on the pseudospectral gap. However, the bound  is misleading: the herded Gibbs sampler is deterministic, and has no (pseudo)-spectral gap at all, but the associated ergodic averages converge extremely quickly even for large data sets.

The main open problem left by our work can be summarized as: \emph{can the phenomenon that occurs for \cite{chen2016herded}, where the (pseudo)-spectral gap is not even close to determining the convergence rate for ergodic average, occur for more general target measures and subsampling algorithms?}

\subsubsection{Relationship to Existing Heuristics} \label{SecSummaryHeuristic}

We note that various heuristic arguments have been put forward to explain the good performance of subsampling-based MCMC algorithms. In this section, we recall these heuristics and give a quick illustration of why they don't appear to answer our open question. The heuristic we present below is based on our understanding of the informal argument in Section 5 of \cite{bierkens2016zig}, where it is applied to their zig-zag process. We suggest interested readers go through their more detailed discussion, as our present version is necessarily shorter and simpler; we also note that very similar arguments also apply to \textit{e.g.} \cite{welling2011bayesian,korattikara2014austerity} with appropriate control variates and target distributions.

 We continue by considering the special case studied in Theorem \ref{ThmExampleTorpid}: our data is drawn from some ``nice" logistic regression model, and we wish to run a subsampling MCMC algorithm targeting a ``very good" approximation of the true posterior distribution. We also precompute the MLE for this model, as well as the Hessian at the MLE, and allow the subsampling MCMC algorithm to use these precomputed values as control variables. The heuristic convergence argument has the following basic form:

\begin{enumerate}
\item After appropriate rescaling, the target distribution \textit{and} the transition kernel of the MCMC algorithm converge to a standard Gaussian and some limiting kernel as the amount of data goes to infinity.
\item The limiting kernel is geometrically ergodic, with stationary measure equal to the standard Gaussian.
\item Since the limiting kernel doesn't include any terms related to the size of the dataset $n$, its convergence rate is independent of $n$.
\item Since the individual kernels converge to the limiting kernel, they should inherit its mixing properties; therefore, the convergence rates of the individual kernels should be bounded uniformly in $n$.
\end{enumerate}

This argument is informal, and the authors of the above papers have not claimed otherwise. The main difficulty is clearly in going from step \textbf{(3)} to step \textbf{(4)} - its correctness depends on exactly how the various types of convergence are measured. Although this heuristic doesn't immediately yield a theorem, it did seem fairly convincing to us at first glance (and even at third glance).

However, we have already seen that the computational cost of subsampling MCMC scales like at least $\Omega(\frac{n}{\log(n)})$ for a broad class of reversible MCMC chains, and that the pseudospectral gap also shrinks at least this quickly for non-reversible chains. In other words: for a variety of algorithms including both \cite{bierkens2016zig} and \cite{welling2011bayesian,korattikara2014austerity}, the informal argument does not give the ``right" answer for some measures of the rate of convergence (in particular, the pseudospectral and spectral gaps respectively).

This leads us to the following very concrete open problem. Our favorite version of this informal argument appears in Section 5 of \cite{bierkens2016zig}, where it is written very clearly (and where there are several clear caveats). The results of our paper tell us that the \textit{pseudospectral gap} doesn't converge as this informal argument would suggest. However, we are not able to draw any conclusions at all about the convergence rate of ergodic averages. Our main question is:

\emph{What is the convergence rate of ergodic averages for the process described in \cite{bierkens2016zig} applied to GLMs, for example? }

\subsubsection{Further Subsamplers}

Although we have focused on the zig-zag process, we note that there has been quite a lot of work recently on the development of other subsamplers that are not reversible Markov chains. We give brief comments on how our results apply or fail to apply.

Several recent papers, including \cite{bierkens2016zig} and \cite{bouchard2017bouncy}, propose continuous-time Markov processes that occasionally ``change directions'' at random times that are more likely to occur in low-density regions of the target distribution. While our results do not apply to the original continuous-time process, they do give bounds on the pseudospectral gap for the discrete-time ``skeleton" chain that records the location and time of each direction change. As a result, the basic conclusions apply.

The ScaLE algorithm of \cite{pollock2016scalable} is a recent and very promising approach to subsampling Monte Carlo. This algorithm is based on the convergence of a Markov process \textit{with killing} to its \textit{quasistationary} distribution, rather than the convergence of a Markov process \textit{without killing} to its \textit{stationary} distribution. As such, results such as Theorem \ref{thm:SpecGapBound} do not apply directly.

We outline the technical problem. The main ingredient of our simple proofs is the ``warm start'' bound of Inequality \eqref{eq:WarmStart}. We are not aware of any analogous bound for the convergence of Markov processes with killing to their quasistationary distributions (though see \textit{e.g.} \cite{van1991quasi, diaconis2014quantitative} for analogous results for discrete Markov chains). Such a bound would be implied by any sensible estimate of the constant $C$ appearing in Theorem 2 of \cite{Wang2019TheoQSMC}. 

\subsection{Free Lunches for Very Expensive Posteriors}

Like the existing literature, we have focused on posterior distributions with the product form
\be
p(z \mid \theta) = p(\theta) \, \prod_{i=1}^{n} p(z_i \mid \theta).
\ee

For such distributions, the per-step cost of a subsampled chain scales like the number $m$ of datapoints used, while our results say (very roughly) that the relaxation time scales like $\frac{n}{m}$ at best. Taking these facts together, this implies that the computational cost per effectively-independent sample from any exact subsampling step must be roughly $\Omega(n)$, regardless of the value of $m$.

However, if we consider more computationally expensive posterior distributions with costs scaling like $m^{c}$ for some $c > 1$, this conclusion disappears: the same heuristic allows for a cost reduction from $\Omega(n^{c})$ to $O(n)$ in the case that $m = O(1)$. We suspect this is a fruitful avenue for future research, and note that the types of problems studied in \cite{montanari2015computational} often fall into this category.

\subsection{Cheaper Lunches for Systematic Data Scans} \label{SecSystScan}

In principle, Theorem \ref{ThmExampleTorpid} allows speedups by a factor of $\log(n)$. In Section \ref{SecDataAug}, we saw that even this logarithmic speedup is impossible when only a small random subsample of data is used at each step.

This problem is reminiscent of the problem of choosing the scan order for Gibbs samplers. By essentially the same ``coupon-collector'' calculation referenced in Section \ref{SecDataAug}, the mixing time of any random-scan Gibbs sampler targeting a density supported on $\mathbb{R}^{d}$ must be at least $\frac{1}{2} d \log(d)$. It is well-known that systematically scanning coordinates in order can sometimes give a Markov chain with mixing time that is $\Theta(d)$. Unfortunately, this speedup is somewhat subtle, and in some cases systematic scans can be \textit{slower} than random scans (see further discussion in \cite{roberts2015surprising}). We suspect a similar phenomenon for subsampling MCMC: systematic scans might be able to give a speedup of roughly $\log(n)$ for some examples, but we expect this speedup to be delicate and quite problem-specific.

\subsection{Free Lunches for Data Augmentation}

Many approximate MCMC algorithms are defined on augmented state spaces. Our results imply that the full MCMC algorithms cannot mix too quickly, but in principle it is possible that all statistically-relevant functions of the chain \textit{do} mix quickly.

We are not aware of any realistic situations in which this occurs, and in Section \ref{SecDataAug}, we used \cite{liu1994covariance} to show that this \textit{can't} occur for a specific algorithm. In general, however, this question appears to be open.

\bibliographystyle{apalike}
\bibliography{ESCBib}

\begin{thebibliography}{}

\bibitem[Alquier et~al., 2016]{alquier2016noisy}
Alquier, P., Friel, N., Everitt, R., and Boland, A. (2016).
\newblock Noisy {M}onte {C}arlo: Convergence of {M}arkov chains with
  approximate transition kernels.
\newblock {\em Statistics and Computing}, 26(1-2):29--47.

\bibitem[Andrieu and Vihola, 2015]{ConvPseudo15}
Andrieu, C. and Vihola, M. (2015).
\newblock {Convergence properties of pseudo-marginal Markov chain Monte Carlo
  algorithms}.
\newblock {\em The Annals of Applied Probability}, 25(2):1030 -- 1077.

\bibitem[Bardenet et~al., 2017a]{bardenet2015markov}
Bardenet, R., Doucet, A., and Holmes, C. (2017a).
\newblock On {M}arkov chain {M}onte {C}arlo methods for tall data.
\newblock {\em The Journal of Machine Learning Research}, 18(1):1515--1557.

\bibitem[Bardenet et~al., 2017b]{bardenet2017markov}
Bardenet, R., Doucet, A., and Holmes, C. (2017b).
\newblock On {M}arkov chain {M}onte {C}arlo methods for tall data.
\newblock {\em The Journal of Machine Learning Research}, 18(1):1515--1557.

\bibitem[Belloni and Chernozhukov, 2009]{belloni2009computational}
Belloni, A. and Chernozhukov, V. (2009).
\newblock On the computational complexity of {MCMC}-based estimators in large
  samples.
\newblock {\em The Annals of Statistics}, 37(4):2011--2055.

\bibitem[Bierkens et~al., 2019a]{bierkens2016zig}
Bierkens, J., Fearnhead, P., and Roberts, G. (2019a).
\newblock The zig-zag process and super-efficient sampling for {B}ayesian
  analysis of big data.
\newblock {\em The Annals of Statistics}, 47(3):pp. 1288--1320.

\bibitem[Bierkens et~al., 2019b]{bierkens2017ergodicity}
Bierkens, J., Roberts, G., and Zitt, P.-A. (2019b).
\newblock Ergodicity of the zigzag process.
\newblock {\em The Annals of Applied Probability}, 29(4):2266 -- 2301.

\bibitem[Bornn et~al., 2017]{bornn2017use}
Bornn, L., Pillai, N.~S., Smith, A., and Woodard, D. (2017).
\newblock The use of a single pseudo-sample in approximate {B}ayesian
  computation.
\newblock {\em Statistics and Computing}, 27(3):583--590.

\bibitem[Bouchard-C{\^o}t{\'e} et~al., 2018]{bouchard2017bouncy}
Bouchard-C{\^o}t{\'e}, A., Vollmer, S.~J., and Doucet, A. (2018).
\newblock The bouncy particle sampler: A nonreversible rejection-free {M}arkov
  chain {M}onte {C}arlo method.
\newblock {\em Journal of the American Statistical Association},
  113(522):855--867.

\bibitem[Campbell and Broderick, 2018]{campbell2018giga}
Campbell, T. and Broderick, T. (2018).
\newblock Bayesian coreset construction via greedy iterative geodesic ascent.
\newblock {\em arXiv 1802.01737}.

\bibitem[Chafai, 2009]{chafai2009singular}
Chafai, D. (2009).
\newblock Singular values of random matrices.
\newblock {\em Lecture Notes}.

\bibitem[Chatterjee, 2019]{chatterjee2019general}
Chatterjee, S. (2019).
\newblock A general method for lower bounds on fluctuations of random
  variables.
\newblock {\em The Annals of Probability}, 47(4):2140--2171.

\bibitem[Chatterjee, 2023]{chatterjee2023spectralgapnonreversiblemarkov}
Chatterjee, S. (2023).
\newblock Spectral gap of nonreversible {M}arkov chains.

\bibitem[Chen et~al., 2022]{SHF22}
Chen, N., Xu, Z., and Campbell, T. (2022).
\newblock Bayesian inference via sparse {H}amiltonian flows.
\newblock 35:20876--20888.

\bibitem[Chen et~al., 2016]{chen2016herded}
Chen, Y., Bornn, L., De~Freitas, N., Eskelin, M., Fang, J., and Welling, M.
  (2016).
\newblock Herded {G}ibbs sampling.
\newblock {\em The Journal of Machine Learning Research}, 17(1):263--291.

\bibitem[Choi, 2017]{choi2017metropolis}
Choi, M.~C. (2017).
\newblock Metropolis-{H}astings reversiblizations of non-reversible {M}arkov
  chains.
\newblock {\em arXiv preprint arXiv:1706.00068}.

\bibitem[Conrad et~al., 2016]{conrad2016accelerating}
Conrad, P.~R., Marzouk, Y.~M., Pillai, N.~S., and Smith, A. (2016).
\newblock Accelerating asymptotically exact {MCMC} for computationally
  intensive models via local approximations.
\newblock {\em Journal of the American Statistical Association},
  111(516):1591--1607.

\bibitem[Cule and Samworth, 2010]{cule2010theoretical}
Cule, M. and Samworth, R. (2010).
\newblock Theoretical properties of the log-concave maximum likelihood
  estimator of a multidimensional density.
\newblock {\em Electronic Journal of Statistics}, 4:254--270.

\bibitem[Diaconis and Miclo, 2015]{diaconis2014quantitative}
Diaconis, P. and Miclo, L. (2015).
\newblock On quantitative convergence to quasi-stationarity.
\newblock {\em Annales de la Facult\'e des sciences de Toulouse :
  Math\'ematiques}, Ser. 6, 24(4):973--1016.

\bibitem[{Erd\"{o}s, Paul and R\'{e}nyi, Alfr\'{e}d}, 1961]{ErRo61Coupon}
{Erd\"{o}s, Paul and R\'{e}nyi, Alfr\'{e}d} (1961).
\newblock On a classical problem of probability theory.
\newblock {\em {Magyar Tudom\'{a}nyos Akad\'{e}mia Matematikai Kutat\'{o}
  Int\'{e}zet\'{e}nek K\"{o}zlem\'{e}nyei}}, 6:215--220.

\bibitem[Fahrmeir and Kaufmann, 1985]{fahrmeir1985consistency}
Fahrmeir, L. and Kaufmann, H. (1985).
\newblock Consistency and asymptotic normality of the maximum likelihood
  estimator in generalized linear models.
\newblock {\em The Annals of Statistics}, 13(1):342--368.

\bibitem[Fearnhead et~al., 2018]{fearnhead2016piecewise}
Fearnhead, P., Bierkens, J., Pollock, M., and Roberts, G.~O. (2018).
\newblock Piecewise deterministic {M}arkov processes for continuous-time
  {M}onte {C}arlo.
\newblock {\em Statistical Science}, 33(3):386--412.

\bibitem[Friedland and Giladi, 2013]{friedland2013simple}
Friedland, O. and Giladi, O. (2013).
\newblock A simple observation on random matrices with continuous diagonal
  entries.
\newblock {\em Electronic Communications in Probability}, 18.

\bibitem[Gallegos-Herrada et~al., 2023]{equivConv22}
Gallegos-Herrada, M.~A., Ledvinka, D., and Rosenthal, J.~S. (2023).
\newblock Equivalences of geometric ergodicity of {M}arkov chains.
\newblock {\em Journal of Theoretical Probability}, 37(2):1230--1256.

\bibitem[Guo and Liao, 2021]{expconv21}
Guo, X. and Liao, Z.-W. (2021).
\newblock Estimate the exponential convergence rate of f -ergodicity via
  spectral gap.
\newblock {\em Statistics \& Probability Letters}, 168:108924.

\bibitem[Huggins et~al., 2017]{huggins2017pass}
Huggins, J., Adams, R.~P., and Broderick, T. (2017).
\newblock {PASS-GLM}: polynomial approximate sufficient statistics for scalable
  bayesian {GLM} inference.
\newblock In {\em Advances in Neural Information Processing Systems}, pages
  3611--3621.

\bibitem[Huggins et~al., 2016]{huggins2016coresets}
Huggins, J., Campbell, T., and Broderick, T. (2016).
\newblock Coresets for scalable {B}ayesian logistic regression.
\newblock In {\em Advances in Neural Information Processing Systems}, pages
  4080--4088.

\bibitem[{Jiang} et~al., 2020]{PseudoMinor2020}
{Jiang}, Y.~H., {Liu}, T., {Lou}, Z., {Rosenthal}, J.~S., {Shangguan}, S.,
  {Wang}, F., and {Wu}, Z. (2020).
\newblock {The Coupling/Minorization/Drift Approach to {M}arkov Chain
  Convergence Rates}.
\newblock {\em arXiv e-prints}, page arXiv:2008.10675.

\bibitem[Johndrow et~al., 2015]{johndrow2015approximations}
Johndrow, J.~E., Mattingly, J.~C., Mukherjee, S., and Dunson, D. (2015).
\newblock Approximations of {M}arkov chains and {B}ayesian inference.
\newblock {\em arXiv preprint arXiv:1508.03387}.

\bibitem[Johndrow et~al., 2019]{johndrow2019mcmc}
Johndrow, J.~E., Smith, A., Pillai, N., and Dunson, D.~B. (2019).
\newblock {MCMC} for imbalanced categorical data.
\newblock {\em Journal of the American Statistical Association},
  114(527):1394--1403.

\bibitem[Jones, 2004]{jones2004markov}
Jones, G.~L. (2004).
\newblock On the {M}arkov chain central limit theorem.
\newblock {\em Probability surveys}, 1(299-320):5--1.

\bibitem[Kontoyiannis and Meyn, 2005]{kontoyiannis2005large}
Kontoyiannis, I. and Meyn, S. (2005).
\newblock Large deviations asymptotics and the spectral theory of
  multiplicatively regular {M}arkov processes.
\newblock {\em Electronic Journal of Probability}, 10:61--123.

\bibitem[Kontoyiannis and Meyn, 2003]{kontoyiannis2003spectral}
Kontoyiannis, I. and Meyn, S.~P. (2003).
\newblock Spectral theory and limit theorems for geometrically ergodic {M}arkov
  processes.
\newblock {\em Annals of Applied Probability}, pages 304--362.

\bibitem[Kontoyiannis and Meyn, 2012]{kontoyiannis2012geometric}
Kontoyiannis, I. and Meyn, S.~P. (2012).
\newblock Geometric ergodicity and the spectral gap of non-reversible {M}arkov
  chains.
\newblock {\em Probability Theory and Related Fields}, 154(1-2):327--339.

\bibitem[Korattikara et~al., 2014]{korattikara2014austerity}
Korattikara, A., Chen, Y., and Welling, M. (2014).
\newblock Austerity in {MCMC} land: Cutting the {M}etropolis-{H}astings budget.
\newblock In {\em Proceedings of the 31st International Conference on Machine
  Learning (ICML-14)}, pages 181--189.

\bibitem[Le~Cam and Yang, 2012]{le2012asymptotics}
Le~Cam, L. and Yang, G.~L. (2012).
\newblock {\em Asymptotics in statistics: some basic concepts}.
\newblock Springer Science \& Business Media.

\bibitem[LEE and ŁATUSZYŃSKI, 2014]{LeeLatVarBound}
LEE, A. and ŁATUSZYŃSKI, K. (2014).
\newblock Variance bounding and geometric ergodicity of {M}arkov chain {M}onte
  {C}arlo kernels for approximate bayesian computation.
\newblock {\em Biometrika}, 101(3):655--671.

\bibitem[Levin and Peres, 2017]{levin2017markov}
Levin, D.~A. and Peres, Y. (2017).
\newblock {\em Markov chains and mixing times}, volume 107.
\newblock American Mathematical Soc.

\bibitem[Liu et~al., 1994]{liu1994covariance}
Liu, J.~S., Wong, W.~H., and Kong, A. (1994).
\newblock Covariance structure of the {G}ibbs sampler with applications to the
  comparisons of estimators and augmentation schemes.
\newblock {\em Biometrika}, 81(1):27--40.

\bibitem[Maclaurin and Adams, 2014]{maclaurin2014firefly}
Maclaurin, D. and Adams, R.~P. (2014).
\newblock Firefly {M}onte {C}arlo: Exact {MCMC} with subsets of data.
\newblock In {\em In 30th Conference on Uncertainty in Artificial Intelligence
  (UAI}. Citeseer.

\bibitem[Maire et~al., 2019]{maire2019informed}
Maire, F., Friel, N., and Alquier, P. (2019).
\newblock Informed sub-sampling {MCMC}: approximate {B}ayesian inference for
  large datasets.
\newblock {\em Statistics and Computing}, 29(3):449--482.

\bibitem[Montanari, 2015]{montanari2015computational}
Montanari, A. (2015).
\newblock Computational implications of reducing data to sufficient statistics.
\newblock {\em Electronic Journal of Statistics}, 9(2):2370--2390.

\bibitem[Montenegro and Tetali, 2006]{montenegro2006mathematical}
Montenegro, R. and Tetali, P. (2006).
\newblock Mathematical aspects of mixing times in {M}arkov chains.
\newblock {\em Foundations and Trends{\textregistered} in Theoretical Computer
  Science}, 1(3):237--354.

\bibitem[Nagapetyan et~al., 2017]{nagapetyan2017true}
Nagapetyan, T., Duncan, A.~B., Hasenclever, L., Vollmer, S.~J., Szpruch, L.,
  and Zygalakis, K. (2017).
\newblock The true cost of stochastic gradient {L}angevin dynamics.
\newblock {\em arXiv preprint arXiv:1706.02692}.

\bibitem[Negrea and Rosenthal, 2021]{negrea_rosenthal_2021}
Negrea, J. and Rosenthal, J.~S. (2021).
\newblock Approximations of geometrically ergodic reversible {M}arkov chains.
\newblock {\em Advances in Applied Probability}, 53(4):981–1022.

\bibitem[Paulin, 2015]{paulin2015concentration}
Paulin, D. (2015).
\newblock Concentration inequalities for {M}arkov chains by {M}arton couplings
  and spectral methods.
\newblock {\em Electronic Journal of Probability}, 20.

\bibitem[Phien, 2012]{phien2012some}
Phien, P. (2012).
\newblock Some quantitative results on {L}ipschitz inverse and implicit
  functions theorems.
\newblock {\em arXiv preprint arXiv:1204.4916}.

\bibitem[Pillai et~al., 2023]{UpcomingAzeem}
Pillai, N., Smith, A., and Zaman, A. (2023+).
\newblock Free lunches for nonreversible approximate {MCMC}.
\newblock Forthcoming.

\bibitem[Pitt et~al., 2012]{RePEc:eee:econom:v:171:y:2012:i:2:p:134-151}
Pitt, M.~K., Silva, R. d.~S., Giordani, P., and Kohn, R. (2012).
\newblock On some properties of {M}arkov chain {M}onte {C}arlo simulation
  methods based on the particle filter.
\newblock {\em Journal of Econometrics}, 171(2):134--151.

\bibitem[Pollock et~al., 2016]{pollock2016scalable}
Pollock, M., Fearnhead, P., Johansen, A.~M., and Roberts, G.~O. (2016).
\newblock The scalable {L}angevin exact algorithm: {B}ayesian inference for big
  data.
\newblock {\em arXiv preprint arXiv:1609.03436}.

\bibitem[Qin and Hobert, 2017]{qin2017asymptotically}
Qin, Q. and Hobert, J.~P. (2017).
\newblock Asymptotically stable drift and minorization for markov chains with
  application to {A}lbert and {C}hib's algorithm.
\newblock {\em arXiv preprint arXiv:1712.08867}.

\bibitem[Qin and Hobert, 2019]{qin2019convergence}
Qin, Q. and Hobert, J.~P. (2019).
\newblock Convergence complexity analysis of {A}lbert and {C}hib's algorithm
  for {B}ayesian probit regression.
\newblock {\em Annals of Statistics}, 47(4):2320--2347.

\bibitem[Quiroz et~al., 2019]{quiroz2015speeding}
Quiroz, M., Villani, M., and Kohn, R. (2019).
\newblock Speeding up {MCMC} by efficient data subsampling.
\newblock {\em Journal of the American Statistical Association},
  114(526):831--843.

\bibitem[Raginsky et~al., 2017]{pmlr-v65-raginsky17a}
Raginsky, M., Rakhlin, A., and Telgarsky, M. (2017).
\newblock Non-convex learning via stochastic gradient {L}angevin dynamics: a
  nonasymptotic analysis.
\newblock In Kale, S. and Shamir, O., editors, {\em Proceedings of the 2017
  Conference on Learning Theory}, volume~65 of {\em Proceedings of Machine
  Learning Research}, pages 1674--1703. PMLR.

\bibitem[Roberts and Rosenthal, 2015]{roberts2015surprising}
Roberts, G.~O. and Rosenthal, J.~S. (2015).
\newblock Surprising convergence properties of some simple {G}ibbs samplers
  under various scans.
\newblock {\em International Journal of Statistics and Probability}, 5(1):51.

\bibitem[Rudolf and Schweizer, 2018]{rudolf2018perturbation}
Rudolf, D. and Schweizer, N. (2018).
\newblock Perturbation theory for {M}arkov chains via {W}asserstein distance.
\newblock {\em Bernoulli}, 24(4A):2610--2639.

\bibitem[Salomone et~al., 2020]{salomone2019spectral}
Salomone, R., Quiroz, M., Kohn, R., Villani, M., and Tran, M.-N. (2020).
\newblock Spectral subsampling {MCMC} for stationary time series.
\newblock 119:8449--8458.

\bibitem[Van~Doorn, 1991]{van1991quasi}
Van~Doorn, E.~A. (1991).
\newblock Quasi-stationary distributions and convergence to quasi-stationarity
  of birth-death processes.
\newblock {\em Advances in Applied Probability}, 23(4):683--700.

\bibitem[Vempala, 2005]{vempala2005geometric}
Vempala, S. (2005).
\newblock Geometric random walks: a survey.
\newblock {\em Combinatorial and Computational Geometry}, 52(2):573--612.

\bibitem[Vempala and Wibisono, 2019]{NEURIPS2019_65a99bb7}
Vempala, S. and Wibisono, A. (2019).
\newblock Rapid convergence of the unadjusted {L}angevin algorithm:
  Isoperimetry suffices.
\newblock In Wallach, H., Larochelle, H., Beygelzimer, A., d\textquotesingle
  Alch\'{e}-Buc, F., Fox, E., and Garnett, R., editors, {\em Advances in Neural
  Information Processing Systems}, volume~32. Curran Associates, Inc.

\bibitem[Wang et~al., 2019]{Wang2019TheoQSMC}
Wang, A.~Q., Kolb, M., Roberts, G.~O., and Steinsaltz, D. (2019).
\newblock {Theoretical properties of quasi-stationary Monte Carlo methods}.
\newblock {\em The Annals of Applied Probability}, 29(1):434--457.

\bibitem[Welling and Teh, 2011]{welling2011bayesian}
Welling, M. and Teh, Y.~W. (2011).
\newblock Bayesian learning via stochastic gradient {L}angevin dynamics.
\newblock In {\em Proceedings of the 28th international conference on machine
  learning (ICML-11)}, pages 681--688.

\bibitem[Wu et~al., 2022]{RevSGLD22}
Wu, T.-Y., Wang, Y. X.~R., and Wong, W.~H. (2022).
\newblock Mini-batch {M}etropolis-{H}astings with reversible {SGLD} proposal.
\newblock {\em Journal of the American Statistical Association},
  117(537):386--394.

\bibitem[Yang and Rosenthal, 2017]{yang2017bigdim}
Yang, J. and Rosenthal, J. (2017).
\newblock Complexity results for {MCMC} derived from quantitative bounds.
\newblock {\em arXiv preprint arXiv:1708.00829}.

\end{thebibliography}

\newpage

\begin{appendix}

\section{Proof of Theorem \ref{thm:SpecGapBound}} \label{SecAppPrfSpecGapBd}

We use the notation of Theorem \ref{thm:SpecGapBound} throughout this section. 

\begin{proof}

 We consider a ``grand coupling" of Markov chains $\{X_{t}^{(z)}\}_{z \in \mathcal{G}^{(n)}, \, t \geq 0}$, with $\{X_{t}^{(z)}\}_{t \geq 0}$ sampled from $K_{z}$ as follows: 

\begin{enumerate}
    \item For $h \in \tilde{\mc H}^{(n)}$, sample $X^{(h)} \sim \mu^*_{h_n}$. Then set $X_{0}^{(z)} = X^{(H(z))}$.
    \item Sample a sequence  $U_1,U_2,\ldots,U_{s_n} \stackrel{iid}{\sim} \nu$.
    \item Inductively set 
    \be 
    X_{t+1}^{(z)} = f_{U_{t}}(X_{t}^{(z)},z)
    \ee 
    for $t \in [s_{n}]$.
\end{enumerate}

Note that this procedure also determines a (coupled) collection $\{M_{t}^{(z)}\}$ of sets of datapoints used by these chains up to time $t$, via Definition \ref{def:UsingData}. 

For any set $M \subset [n]$, denote by $\mathcal{E}_{M}$ the (random) collection of pairs of points $z^{(1)},z^{(2)}$ satisfying 
\be \label{CondEqualFirstPoints}
M_{s_{n}}^{(z^{(1)})}, M_{s_{n}}^{(z^{(2)})} & \subset M \\
z^{(1)}_{i} &= z^{(2)}_{i}, \qquad \forall i \in M.
\ee 

Under this coupling, we have

\be \label{EqInEM}
X_{s_{n}}^{(z^{(1)})} = X_{s_{n}}^{(z^{(2)})}
\ee 
for all  $M \subset [n]$ and  pairs $z^{(1)},z^{(2)} \in \mathcal{E}_{M}$.

On the other hand, for any fixed pair of points $z^{(1)},z^{(2)}$ we have 

\be 
P[X_{s_{n}}^{(z^{(1)})} &\neq X_{s_{n}}^{(z^{(2)})}] \geq  \| \mu^{*}_{H(z^{(1)})} K_{z^{(1)}}^{s_{n}} -  \mu^{*}_{H(z^{(2)})} K_{z^{(2)}}^{s_{n}} \| \\
&= \| (\mu^{*}_{H(z^{(1)})} K_{z^{(1)}}^{s_{n}} - \mu_{z^{(1)}}) + (\mu_{z^{(1)}} - \Pi_{z^{(1)}})  + (\Pi_{z^{(1)}} -\Pi_{z^{(2)}}) \\
& \qquad +  ( \Pi_{z^{(2)}} - \mu_{z^{(2)}}) + (\mu_{z^{(2)}} -\mu^{*}_{H(z^{(2)})} K_{z^{(2)}}^{s_{n}} ) \| \\
&\geq\| \Pi_{z^{(1)}} -\Pi_{z^{(2)}}\|  - \| \Pi_{z^{(1)}} - \mu_{z^{(1)}} \| - \| \mu_{z^{(2)}} - \Pi_{z^{(2)}} \| \\
&\qquad - \|\mu_{z^{(1)}} - \mu^{*}_{H(z^{(1)})} K_{z^{(1)}}^{s_{n}} \|  - \| \mu^{*}_{H(z^{(2)})} K_{z^{(2)}}^{s_{n}} - \mu_{z^{(2)}}  \|. \\
\ee 

Rearranging, we have 
\be \label{BigMessyTrianglePre}
\| \Pi_{z^{(1)}} -\Pi_{z^{(2)}}\| &\leq (\| \Pi_{z^{(1)}} - \mu_{z^{(1)}} \| + \| \mu_{z^{(2)}} - \Pi_{z^{(2)}} \|) \\
&\qquad +  (\|\mu_{z^{(1)}} - \mu^{*}_{H(z^{(1)})} K_{z^{(1)}}^{s_{n}} \|  + \| \mu^{*}_{H(z^{(2)})} K_{z^{(2)}}^{s_{n}} - \mu_{z^{(2)}}  \|) \\
& \qquad + P[X_{s_{n}}^{(z^{(1)})} \neq X_{s_{n}}^{(z^{(2)})}].
\ee

We now consider fixed $z^{(1)}, z^{(2)}$ and define $z= (z^{(1)}, z^{(2)})$,  $\delta_z(U,x_0) = X_{s_n}^{(z^{(1)})} - X_{s_n}^{(z^{(2)})}$, $\mu^{+}_{z} = |\Pi_{z^{(1)}} - \Pi_{z^{(2)}}|$, and 
\be
\mu^{-}_z &= |\Pi_{z^{(1)}} - \mu_{z^{(1)}}| + |\mu_{z^{(2)}} - \Pi_{z^{(2)}}| \\
&\quad+  |\mu_{z^{(1)}} - \mu^{*}_{H(z^{(1)})} K_{z^{(1)}}^{s_{n}}|  + |\mu^{*}_{H(z^{(2)})} K_{z^{(2)}}^{s_{n}} - \mu_{z^{(2)}}|.
\ee

Define $U^{(s)} = \{U_1,\ldots,U_s\}$, $\nu^{(s)} = \prod_{t=1}^s \nu$, and $\nu^* = \nu^{(s)} \times \mu^*$. Then we can write \eqref{BigMessyTrianglePre} as
\be 
\int \textbf{1}_{\delta_{z}(U^{(s)},x_0) \neq 0} \nu^{*}(d U^{(s)} dx_0) \geq \int \mu^+_{z}(dx) - \int \mu_{z,-}(dx).
\ee 

Since $\nu^*$ is a probability measure (unlike $\mu^+_z$ and $\mu^-_z$), we can rewrite this as:
\be \label{EqBeforeSplitNuStar}
\int \textbf{1}_{\delta_{z}(U^{(s)},x_0) \neq 0} \nu^{*}(d U^{(s)} dx_0) &\geq \int \mu^+_{z}(dx) \nu^*(dU^{(s)} dx_0) \\
&\qquad - \int \mu_{z,-}(dx) \nu^*(dU^{(s)} dx_0).
\ee

At this point, it is helpful to split the randomness of the measure $\nu^{*}$ into two parts - the choice of $M$ and everything else. More precisely, for fixed $z^{(1)}$, the measure $\nu^{*}$ induces a measure on the collection of indices $M_{s_{n}}^{(z^{(1)})}$ ``seen" by time $s_{n}$. Suppressing some of the dependencies to reduce notational clutter, we denote this induced measure as $\nu_M$. Denote by $\nu^*_{\mid m}$ be the measure $\nu^*$ conditioned on the value $M=m$. Since $M$ is a random variable on a finite set, the conditional measure $\nu^*_{\mid m}$ is always well-defined.

We can then write Inequality \eqref{EqBeforeSplitNuStar} as:\footnote{At this point, we have not yet guaranteed that $z = (z^{(1)},z^{(2)}) \in \mathcal{E}_{M}$, so this random $M$ is not yet terribly meaningful. }

\be 
\int \textbf{1}_{\delta_{z} \neq 0}  \nu^*_{\mid m}(dU^{(s)} dx_0) \nu_{M}(dm)  &\geq \int \mu_{z,+}(dx) \nu^*_{\mid m}(dU^{(s)} dx_0) \nu_{M}(dm) \\
&\qquad - \int \mu_{z,-}(dx) \nu^*_{\mid m}(dU^{(s)} dx_0) \nu_{M}(dm).
\ee

We now bound this in two cases: $M \in \mathcal{M}$ or not. Noting that all of the terms being integrated on the right-hand side are bounded by 4,  we get:

\be \label{EqIntermediateLongCond}
0 &\leq 8 \bb P[M \notin \mathcal{M}] + \sup_{m \in \mathcal{M}} \big( \int \textbf{1}_{\delta_{z} \neq 0}  \nu^*_{\mid m}(dU^{(s)} dx_0)  \\
& \quad - \int \mu_{z,+}(dx) \nu^*_{\mid m}(dU^{(s)} dx_0)  + \int \mu_{z,-}(dx) \nu^*_{\mid m}(dU^{(s)} dx_0) \big).
\ee 

We are now ready to consider random choices of $z$. Sample $\bf Z^{(n)} \sim f^{(n)}$ and then set $\tilde{M}_{n} = M_{s_{n}}^{(\bf{Z}^{(n)})}$, $h_{n} = H(\bf{Z}^{(n)})$. Conditional on these choices, sample $\bf Z_{1}^{(n)},\bf Z_{2}^{(n)} \stackrel{i.i.d.}{\sim}  f_{\mc A(\tilde{M}_{n},\bf{Z}^{(n)},h_n)}$ and set $\bf Z = (\bf Z_{1}^{(n)}, \bf Z_{2}^{(n)})$. We note that  $M_{s_{n}}^{\bf Z_{1}^{(n)}} = M_{s_{n}}^{\bf Z_{2}^{(n)}} = \tilde{M}_{n}$, and $\tilde{M}_{n}$ has distribution $\nu_{M}$. Continuing from Inequality \eqref{EqIntermediateLongCond} with this choice,

\be 
8 \bb P[\tilde{M}_{n} \notin \mathcal{M}] &+ \sup_{m \in \mathcal{M}} (\int \textbf{1}_{\delta_{\bf Z} \neq 0}  \nu^*_{\mid m}(dU^{(s)} dx_0) f_{\mc A(m,\bf Z^{(n)},h_n)}(dz)  - \int \mu_{\bf Z,+}(dx) \nu^*_{\mid m}(dU^{(s)} dx_0) f_{\mc A(m,\bf Z^{(n)},h_n)}(dz) \\
&+ \int \mu_{\bf Z,-}(dx) \nu^*_{\mid m}(dU^{(s)} dx_0) f_{\mc A(m,\bf Z^{(n)},h_n)}(dz)) \geq 0.
\ee

Simplifying and applying Fubini's theorem,

\be 
8 \bb P[\tilde{M}_{n} \notin \mathcal{M}] &+ \sup_{m \in \mathcal{M}} (\int \textbf{1}_{\delta_{\bf Z} \neq 0}   f_{\mc A(m,\bf Z^{(n)},h_n)}(dz) \nu^*_{\mid m}(dU^{(s)} dx_0)  - \int \mu_{\bf Z,+}(dx)  f_{\mc A(m,\bf Z^{(n)},h_n)}(dz) \\
&+ \int \mu_{\bf Z,-}(dx)  f_{\mc A(m,\bf Z^{(n)},h_n)}(dz)) \geq 0.
\ee 

But as pointed out in Equation \eqref{EqInEM}, in fact $\delta_{z} =0$ whenever $m \in \mathcal{M}$ and $z \in \mc A(m,\bf Z^{(n)},h_n)$, so we get:

\be 
8 \bb P[\tilde{M}_{n} \notin \mathcal{M}] + \sup_{m \in \mathcal{M}} (- \int \mu_{\bf Z,+}(dx)  f_{\mc A(m,\bf Z^{(n)},h_n)}(dz) + \int \mu_{\bf Z,-}(dx)  f_{\mc A(m,\bf Z^{(n)},h_n)}(dz)) \geq 0.
\ee

Rewriting in the notation of \eqref{BigMessyTrianglePre}, we have:
\be \label{BigMessyTriangle}
\| \Pi_{\bf Z^{(1)}} -\Pi_{\bf Z^{(2)}}\| &\leq (\| \Pi_{\bf Z^{(1)}} - \mu_{\bf Z^{(1)}} \| + \| \mu_{\bf Z^{(2)}} - \Pi_{\bf Z^{(2)}} \|) \\
&+  (\|\mu_{\bf Z^{(1)}} - \mu^{*}_{H(\bf Z^{(1)})} K_{\bf Z^{(1)}}^{s_{n}} \|  + \| \mu^{*}_{H(\bf Z^{(2)})} K_{\bf Z^{(2)}}^{s_{n}} - \mu_{\bf Z^{(2)}}  \|) \\
& \qquad + 8 P[\tilde{M} \notin \mathcal{M}^{(n)}].
\ee

We now need to bound these terms. Almost all of these bounds follow from parts of Assumption \ref{ass:LargeFluctuations}.
 
\begin{enumerate}
    \item Let 
    \be 
    \mathcal{A}_{4} = \{ \tilde{M} \in \mathcal{M}^{(n)}\}.
    \ee 
     By Assumption \ref{ass:LargeFluctuations} (d), $P[\mathcal{A}_{4}] \geq 1 - \frac{1}{32} C_{1} n^{-c_{1}}$.
    \item Let  
    \be \mathcal{A}_{1} = \{\bf Z^{(n)}, \bf Z_{1}^{(n)},\bf Z_{2}^{(n)} \in \mathcal{G}^{(n)}\}.
    \ee By Assumption \ref{ass:GoodSetHP}, $P[\mathcal{A}_{1} \cup \mathcal{A}_{4}^{c}] \geq 1 - 3 \, \omega_{4}(n).$ 
    \item Let 
    \be \mathcal{A}_{2} = \{ \| \Pi_{z^{(1)}} -\Pi_{z^{(2)}}\| - (\| \Pi_{z^{(1)}} - \mu_{z^{(1)}} \| + \| \mu_{z^{(2)}} - \Pi_{z^{(2)}} \|)  > \frac{C_{1}}{2} n^{-c_{1}} \}.
    \ee 
    By Assumptions \ref{ass:LargeFluctuations} (a,b), $P[\mathcal{A}_{2} \cup \mathcal{A}_{1}^{c} \cup \mathcal{A}_{4}^{c}] \geq 1 - \omega_{1}(n) - 2 \omega_{2}(n)$.
    \item Let 
    \be 
    \mathcal{A}_{3} = \{ \|\mu_{\bf{Z}_{n}^{(1)}} - \mu^{*}_{H(\bf{Z}_{n}^{(1)})} K_{\bf{Z}_{n}^{(1)}}^{s_{n}} \|  + \| \mu^{*}_{H(\bf{Z}_{n}^{(2)})} K_{\bf{Z}_{n}^{(1)}}^{s_{n}} - \mu_{\bf{Z}_{n}^{(1)}}  \| \leq C_{2} n^{c_{2}} ((1 - \lambda_{\bf{Z}_{n}^{(1)}})^{s_{n}} + (1 - \lambda_{\bf{Z}_{n}^{(2)}})^{s_{n}})\}. 
    \ee 
    By Assumption \ref{ass:LargeFluctuations} (c), $P[\mathcal{A}_{3} \cup \mathcal{A}_{1}^{c}] \geq 1 - 2 \omega_{3}(n)$.
\end{enumerate}

Putting these together with Inequality \eqref{BigMessyTriangle},
\be 
\frac{1}{4} C_{1}n^{-c_{1}} \leq C_{2} n^{c_{2}} ((1 - \lambda_{\bf{Z}_{n}^{(1)}})^{s_{n}} + (1 - \lambda_{\bf{Z}_{n}^{(2)}})^{s_{n}})
\ee 
holds with probability at least 

\be 
p(n) \equiv 1 - 3 \omega_{4}(n) - \omega_{1}(n) - 2 \omega_{2}(n) - 2 \omega_{3}(n).
\ee 

We now recall the following simple fact: if two independent, nonnegative random variables $X,Y$ with the same distribution satisfy $P[X+Y>A] > q$ for some $A,q > 0$, then $P[X>A/2]>q$ as well.

Rearranging and recalling that $\lambda_1 = \lambda_{\bf Z^{(n)}_1},\lambda_{\bf Z^{(n)}_2} = \lambda_2$ are independent and have the same distribution, we use this fact to conclude that
\be \label{eq:SmallSpecGap}
(1 - \lambda_1)^s \geq \frac{1}{8} C_2^{-1} C_1 n^{-c_1 - c_2}
\ee
with probability $p(n)$ (based on taking a union bound that any of the necessary events does not occur).

\end{proof}

\section{Discussion of Warm Starts} \label{AppSecWarmStart}

As mentioned immediately following the statement of Assumption \ref{ass:LargeFluctuations}, the only part of Assumption \ref{ass:LargeFluctuations} involving both the model \textit{and} the algorithm is part (c), which we repeat here for convenience:

\be \label{EqWarmStartConv}
\bb P\left[ \|\mu^*_{H(\bf Z^{(n)}_{1})} K^t - \mu_{\bf Z^{(n)}_1}\|_{\TV} \leq C_2 n^{c_2} (1-\lambda_{\bf Z^{(n)}_{1}})^t  \right] > 1-\omega_3(n).
\ee

Because this involves both the algorithm and model, it seems difficult to give completely general sufficient conditions. In this appendix we describe a few techniques and earlier results that can be used to show that these conditions do hold for common models, including our pedagogical examples of SGLD. We do not give a similar argument for our other pedagogical example, as it is well-known that ``naive" pseudomarginal algorithms can easily fail to be geometrically ergodic. In practice, it is common to adjust the algorithm, but this is beyond the scope of the current paper. See \cite{ConvPseudo15} for discussion of geometric ergodicity of pseudomarginal algorithms and \cite{LeeLatVarBound} for adjustments that restore geometric ergodicity when it would otherwise fails.

\subsection{Warm Starts, Geometric Ergodicity and Usable Algorithms}

Recall that, if a transition kernel $Q$ with stationary measure $\nu$ is \textit{geometrically ergodic}, then for any initial measure $\nu_{0}$ we have:

\be 
\| \nu_{0}Q^{t} - \nu \|_{\TV} \leq C_{\nu_{0}} \rho^{t}
\ee 

for some $0 < \rho < 1$ and some $C_{\nu_{0}} < \infty$. The rate of convergence does not depend on the initial condition, and is equal to the spectral gap under mild conditions (see \cite{expconv21} for this fact, and \cite{equivConv22} for more general equivalences). Thus, a geometrically ergodic chain satisfies \eqref{EqWarmStartConv} as long as the constant $C_{\nu_{0}}$ grows at most polynomially in the size $n$ of the dataset of interest.

 Proving this for new algorithms is often a substantial amount of work, but the work has already been done for many subsampling algorithms of interest (see \textit{e.g.} \cite{bierkens2017ergodicity}, \cite{negrea_rosenthal_2021} and references therein). Note that these papers generally do not focus on the constant $C_{\nu_{0}},$ but it is possible to trace the arguments in these papers (and most others in the literature) to obtain explicit bounds. In practice, these bounds often go through the choice of an explicit \textit{Lyapunov function}, and choosing $\nu_{0}$ to be concentrated at the minimum of this function will tend to yield good constants if any choice does.
 
 We note that this approach can be used to show that our pedagogical example of SGLD for fixed-dimension logistic regression satisfies the warm start condition. Unfortunately, as mentioned above, proving geometric ergodicity for a specific algorithm often takes an entire paper. Adding that would be substantially beyond the scope of the current paper, so  we give a quick sketch of the relevant bounds in the literature.  Inequality (3.3) of \cite{pmlr-v65-raginsky17a} shows that SGLD converges quickly in the \textit{Wasserstein distance}. Since we consider fixed dimension as $n$ goes to infinity, it is straightforward to check that SGLD also satisfies a pseudo-minorization condition (see Section 4.2 of \cite{PseudoMinor2020} for a definition). Finally, combining the Wasserstein convergence condition and the pseudo-minorization condition leads directly to a geometric ergodicity bound (see e.g. Theorem 2 of \cite{PseudoMinor2020}, which uses a weaker Lyapunov condition in place of the stronger Wasserstein convergence condition mentioned here).

We believe that this approach to verifying Inequality \eqref{EqWarmStartConv} is important beyond the fact that it often works. Markov chains that are not \textit{geometrically ergodic} are usually seen as very poor Markov chains for MCMC in practice (see \textit{e.g.} the discussion in \cite{LeeLatVarBound}). Thus, we would view a failure to satisfy Inequality \eqref{EqWarmStartConv} as a fairly strong indication \textit{in and of itself} that the underlying is not sound.

\subsection{Warm Starts for Moderate Densities}

As mentioned immediately following Definition \ref{def:Warm}, one always has

\be 
 \|\mu^*_{H(\bf Z^{(n)}_{1})} K^t - \mu_{\bf Z^{(n)}_1}\|_{\TV} \leq C \, (1-\lambda_{\bf Z^{(n)}_{1}})^t  
\ee
with 
\be 
C = \| \frac{d \mu^*_{H(\bf Z^{(n)}_{1})}}{d \mu_{\bf Z^{(n)}_1}} \|_{2}^{2}.
\ee 

Thus, Inequality \eqref{EqWarmStartConv} is satisfied as long as this density grows at most polynomially in $n$. Verifying this condition would usually be straightforward if the stationary measure $\mu_{\bf Z^{(n)}_1}$ were replaced by the posterior distribution of interest (see Section \ref{AppWarmStart} for a simple lemma of this form). However, techniques have recently been developed to prove these estimates even for ``approximate" MCMC algorithms. See Theorem 8 of \cite{NEURIPS2019_65a99bb7} for a general bound of this form, and Example 3 of the same paper for a concrete estimate of the prefactor $C$ (in both cases, take parameter $q=2$).

\section{Sufficient Conditions for Large Fluctuations and Warm Starts} \label{AppendixLargeFluxChecking}

We begin by giving some simple conditions under which our ``warm start" and  ``large fluctuation'' assumptions hold (see Sections \ref{AppWarmStart} and \ref{SecGenSuffCondLarge} respectively). We then apply these to prove Theorem \ref{ThmActualLogisticRegression}, our main result on logistic regression, in Section \ref{SecProofMainThm}. Finally, Section \ref{AppSubsecGrid} gives some suggestions on constructing ``additional" control variates.

We note that the main obstacle in extending Theorem \ref{ThmActualLogisticRegression} to other models and control variates is checking that a certain matrix is far from singular, as discussed in Section \ref{SecGenSuffCondLarge}. See that section for an overview of the main difficulties, including an algorithm for using a computer to verify the main conditions for specific examples in Remark \ref{RemSimpleAlgBadCond}.

\subsection{Warm Starts for Target Measure} \label{AppWarmStart}

We introduce sufficient conditions for warm starts that apply to all GLMs considered in the current paper, and many other models. Define $B_{r}(x) = \{ y \, : \, \| x - y \| \leq r\}$ to be the usual ball of radius $r$ around $x$.

\begin{lemma} \label{LemmaWarmStartSuffCond}
Let $\{\nu_{n}\}_{n \in \mathbb{N}}$ be a sequence of measures on $\mathbb{R}^{d}$ with densities $\{p_{n}\}_{n \in \mathbb{N}}$. Assume there exist constants $\delta, c, N_{0}$ so that
\begin{enumerate}
\item For $n > N_0$, $p_n$ has a unique maximizer $\hat \theta_n$.
\item For $n > N_0$, $\nu_n(B_{1}(\hat \theta_n)) > 0.5$.
\item For all $n > N_{0}$,
\be \label{IneqMLEOK}
\sup_{\| \theta - \hat{\theta}_{n}\| \leq \delta} \sup_{\| v \| = 1 } \frac{ | \nabla_{v} p_{n}(\theta)|}{p_{n}(\theta)} \leq n^{c}.
\ee 
\end{enumerate}

Then the sequence of measures $\mathrm{Unif}(B_{n^{-c}}(\hat{\theta}_{n}))$, with densities $\lambda_{n}$, satisfies
\be \label{IneqWarmStartConcApp}
\frac{\lambda_{n}(\theta)}{p_{n}(\theta)} \leq C n^{c d}
\ee
for some constant $C$ that does not depend on $n, c$ or $d$.
\end{lemma}
\begin{proof}
Note that, for any $0 \leq r < \delta$ and $\theta \in B_{r}(\hat{\theta}_{n})$, we have by Inequality \eqref{IneqMLEOK} and Gr\"{o}nwall's inequality
\be
p_{n}(\theta) \geq e^{-r n^{c}} p_{n}(\hat{\theta}_{n}).
\ee
In particular, choosing $r(n) = n^{-c}$,
\be \label{IneqDensityBigInRegion}
p_{n}(\theta) \geq e^{-1}  p_{n}(\hat{\theta}_{n}), \qquad \forall \theta \in B_{r}(\hat{\theta}_{n})
\ee
for all $n > N_0 \vee \delta^{-1/c}$. Also for $n>N_0$, we have by assumption that $\nu_n(B_{1}(\hat \theta_n))>0.5$. Then it follows that 
\be
0.5 &< \nu_n(B_{1}(\hat \theta_n)) < p_n(\hat \theta_n) \text{Vol}(B_{1}(\hat \theta_n)) \\
p_n(\hat \theta_n) &> 0.5 \frac{\Gamma(d/2+1)}{\pi^{d/2} }
\ee
On the other hand
\be
\lambda_n(\theta) = \frac{1}{\text{Vol}(B_{r(n)}(\hat \theta_n))} = r(n)^{-d} \frac{\Gamma(d/2 + 1)}{\pi^{d/2}} 
\ee
Giving
\be
\frac{\lambda_n(\theta)}{p_n(\theta)} < \frac{r(n)^{-d} }{e^{-1} 0.5 r(n)^{-d}} = 2 e n^{cd}.
\ee

\end{proof}

Note that posterior distributions themselves will generally not satisfy the conditions of this lemma. However, if $q_{n}$ is a posterior distribution, we expect that $p_{n}(\theta) \equiv q_{n}(B_n(\theta - a_{n}))$ \textit{will} satisfy the conditions for appropriate choice of $a_n \in \bb R^d$ and $B_n$ a $d \times d$ matrix with operator norm $\|B_n\| = n^{O(1)}$, which is sufficient to recover the required warm start condition by applying the chain rule. In the context of this paper, we choose the same $a_{n}, B_{n}$ as appears in the Bernstein-von Mises theorem (see \textit{e.g.} \cite{le2012asymptotics}\footnote{A small warning to the reader uninitiated in asymptotics of the MLE: as discussed in \cite{le2012asymptotics}, actually verifying convergence to a Gaussian with nonrandom normalizing sequences is surprisingly time-consuming. However, our theorem above requires only tightness of the posterior distributions, and so merely having a bounded sequence of normalizing constants is enough. }).

Although it is beyond the scope of this paper to survey this literature, there is a large literature on similar ``warm start'' conditions, and much stronger results are possible. For example,  Section 9 of \cite{vempala2005geometric} gives a construction of a distribution that satisfies a far stricter ``warm start'' condition for a very broad class of distributions.

In the special case of logistic regression, we recall that the posterior distribution for logistic regression with suitable prior is log-concave. When the posterior also converges (upon rescaling) to a standard Gaussian, very strong results on log-concavity, such as Proposition 2 of \cite{cule2010theoretical}, imply that a bound of the form Inequality \eqref{IneqWarmStartConcApp} holds with the right-hand side replaced by some absolute constant that does not depend on $n$ or $d$.

\subsection{General Sufficient Condition for Large Fluctuations} \label{SecGenSuffCondLarge}

In this section, we give some notation and  sufficient conditions for Assumption \ref{ass:LargeFluctuations} to hold, with an emphasis on Part (a) of that assumption. Throughout the section, fix notation as in that assumption. Assume that $\tilde{M}^{(n)}$ is of the form $\tilde{M}^{(n)} = \{S \subset [n] \, : \, |S| \leq m^{(n)}\}$ for some sequence $m^{(n)}$. We will primarily be concerned with regression-like models, so we assume that $\mathcal{Z} = \mathcal{X} \times \mathcal{Y}$, where $\mathcal{X} = \mathbb{R}^{d}$ for some $d \in \mathbb{N}$.

Rather than giving a minimal-length proof of our application in the main text, we give a collection of conditions that will be more useful to readers who wish to establish large fluctuation conditions for their own models. We give these conditions in two parts. In Section \ref{SubsecGeomCondApp}, we state ``geometric" conditions stated in terms of objects such as manifolds, tangent spaces and so on associated with the model; see Assumption \ref{AssumptionsGoodPoints}. In Section \ref{SubsecExpRegApp}, we specialize to the case of GLMs and give simple conditions on the model that imply these geometric conditions; see Assumption \ref{AssumptionsGLMVersion}. Between these two sections, we hope to illustrate that the large fluctuation condition is both (i) fairly generic, even for models that don't look much like GLMs, and (ii) often straightforward to verify for popular models. We note that the conditions in these sections are still far from optimal - for example, the conditions are based on taking various derivatives, so they don't apply to discrete covariate data.

We note that the main technical difficulty throughout essentially all of this section comes from the fact that we are conditioning on the values of the test statistics $T = (T_{1},\ldots,T_{k})$. This will generally restrict our allowed data to a submanifold defined by the level sets of $T$, and we must confirm that this submanifold is not too badly-behaved (primarily that it is well-approximated by its tangent planes on a sufficiently large scale). In the end, this boils down to verifying that a certain random matrix with highly dependent entries is not very close to singular; this seems to often be true, but such random matrices are not easy to analyze using currently-available tools. While this presents a major mathematical difficulty towards obtaining a general theory, it is fortunately simple to check on a case-by-case basis. Remark \ref{RemSimpleAlgBadCond} immediately follows the condition that is difficult to verify, explains the difficulty, and gives  a straightforward (and often very fast) algorithm for verifying it for any specific model.

\subsubsection{Geometric Conditions for Large Fluctuations} \label{SubsecGeomCondApp}

Roughly speaking, we expect a large fluctuation condition to hold if the following hold at ``most" points:

\begin{enumerate}
\item It is possible to perturb the data by some small amount, \textit{while still} holding the test statistics constant (see the first part of \ref{AssumptionsGoodPoints}).
\item Most small perturbations to the data give similar-size perturbations to the posterior density at most points  (see the rest of \ref{AssumptionsGoodPoints}).
\end{enumerate}

This section's main result, Lemma \ref{LemmaSuffCondDerivatives}, is a precise statement of this heuristic. For any fixed collection of test statistics $T_{1}, \ldots, T_{k}$, covariate data $X = (x_{1},\ldots,x_{n}) \in \mathcal{X}^{n}$, response data $Y \in \mathcal{Y}^{n}$ and integer $1 \leq m \leq n$, define
\be \label{ManifoldDefinition}
\mathcal{G} &= \mathcal{G}(X) \\
&= \{G = (g_{1},\ldots,g_{n}) \in \mathcal{X}^{n} \, : \, T_{i}(G,Y) = T_{i}(X,Y) \, \forall i \leq k, \, g_{1:m} = x_{1:m}  \}
\ee
to be the collection of extensions of $x_{1:m}$ that keep the existing statistics $T_{1},\ldots,T_{k}$  fixed. In other words, this is the set of all sequences of length $n$ having both the same statistics and the same first $m$ elements. We note that we are suppressing the response data $Y$ in this notation. This is largely for convenience: our proof results are based on taking derivatives of various objects, and in our prototypical example $Y$ takes on discrete values. We resolve this tension by viewing $Y$ as fixed (rather than random) throughout the following calculations - doing so makes the proofs harder but the notation simpler. We denote by $\gamma_{n,Y}$ the conditional distribution of $X$ given these observed response variables. Finally, we note that $\mathcal{G}(X)$ is determined by any of its elements: for all $X$ and all $X' \in \mathcal{G}(X)$, we have $\mathcal{G}(X) = \mathcal{G}(X')$.

We now make an assumption about the ``inner curvature" of $\mc G$. Before giving a formal definition of inner curvature, we give a very special case that is easy to understand: when $\mc G$ can be written in the form $\mc G = \{ (x,f(x)) \, : \, x \in \mathbb{R} \}$ for some smooth $f \, : \, \mathbb{R} \mapsto \mathbb{R}$, the inner curvature of $\mc G$ is just $\sup_x |f''(x)|^{-1/2}$.

We say that an embedded manifold $\mc M$ in Euclidean space has \textit{strong inner curvature} at least $c > 0$ if, for all $x \in \mc M$ and all balls $B$ tangent to $\mc M$ at $x$ with radius $c$, we have $\mc M \cap B = \{x \}$. We need to slightly relax this condition to allow for the possibility that  $\mc M$ has self-intersections or otherwise ``loops back" on itself:

\begin{defn} \label{DefnInnercurvature}

We say that a closed subset $\mc M$ of Euclidean space has \textit{inner curvature} at least $c > 0$ if there exists an \textit{open cover}  of $\mc M$ by embedded manifolds $\{U_{\alpha} \}_{\alpha \in \mathbb{N}}$, and furthermore this open cover satisfies:
\begin{enumerate}
\item For every ball $B$ of radius $c$, there exists $\alpha \in \mathbb{N}$ so that $B \cap \mc M \subset U_{\alpha} $.
\item For every $\alpha$, $U_{\alpha}$ has strong inner curvature at least $c$.
\end{enumerate}
\end{defn}

This condition may look intimidating at first, but our manifold of interest $\mathcal{G}(X)$ is rather special: it is exactly the set on which a collection of equations defined by $T_{1},\ldots,T_{k}$ is constant. In the end, this means that it is possible to check that  $\mathcal{G}(X)$ has inner curvature at least $c > 0$ by doing some fairly explicit computations involving derivatives of $T_{1},\ldots,T_{k}$. Roughly speaking, we can use the implicit function theorem to find a local parameterization of $\mathcal{G}(X)$ in terms of these functions, and then verify the inner curvature condition by finding a lower bound on the smallest singular value of the Jacobian of this parameterization. The details of this strategy are carried out in Lemma \ref{ThmRelateGeometricGLMAssumptions}.

Informally, having inner curvature $c > 0$ means that we can travel the manifold $\mc M$ along the direction of tangent spaces in the ``obvious" way up to distances of about $c$ without encountering problems. We now formalize what we mean by the ``obvious" way. Fo $x \in \mc M$, define $\alpha(x) = \min \{ \beta \in \mathbb{N} \, : \, B_c(x) \cap \mc M \subset U_{\beta(x)} \}$, where we recall $B_{c}(x)$ is the ball of radius $c$ around $x$ and the minimum appears in this definition solely to provide some measurable function that chooses a single element of the cover. For $x \in \mc M$ and unit vector $v$, define the function $\psi_{x,v} \, : \, \mathbb{R} \mapsto U_{\alpha(x)}$ by
\be \label{eq:psi}
\phi_{x,v}(s) = \mathrm{argmin}_{x' \in U_{\alpha(x)}} \| x + sv - x' \|.
\ee
This function determines short paths along $\mc M$ that start by travelling in direction $v$ from point $x$. When $\mc M$ has inner curvature $n^{-c}$, it follows that for any $s < n^{-c}$ and $v \in \mc T_{\mc M}(x)$, the tangent space of $\mc M$ at $x$, we have
\be \label{eq:PhiSmall}
\phi_{x,v}(s) = x + s' v'
\ee
for some $s', v'$ satisfying $|s-s'|, \, \frac{\|v-v'\|}{s} = O(n^{c} s^{2})$. 

Having made an assumption that it is possible to travel along $\mathcal{G}$ in a fairly smooth way, we next make assumptions about how small changes along $\mathcal{G}$ result in changes to the posterior. For $m < i \leq n$, $1 \leq j \leq d$, $1 \leq \ell \leq k$ and $X = (x_{1},\ldots,x_{n}) \in \mathcal{X}^{n}$ with $x_{i} = (x_{i1},\ldots,x_{id})$, define
\be
\dot{T}_{i,j}^{(\ell)} = \dot{T}_{i,j}^{(\ell)}(X) = \frac{\partial}{\partial x_{ij}} T_{\ell}(X,Y).
\ee
For $\theta \in \Theta$, also define
\be
D_{i,j}(\theta,X) = \frac{\partial}{\partial x_{ij}} \log\{p(\theta \mid (X,Y))\} = \frac{1}{p(\theta | X,Y)} \frac{\partial}{\partial x_{ij}} p(\theta | (X,Y))
\ee
and
\be
D_{\max}(\theta,X)  = \max_{i,i',j,j'}  \frac{1}{p(\theta | (X,Y))} \left| \frac{\partial^{2}}{\partial x_{ij} \partial x_{i'j'}} p(\theta | (X,Y)) \right|.
\ee
With the same indices, define
\be
W_{\ell} = W_{\ell}(X) \equiv \bigg\{ &v = (0,0,\ldots,0,v_{m+1}, \ldots,v_{n}) \in \left( \mathbb{R}^{d} \right)^{n}\, : \\
\, &\sum_{i=m+1}^{n} \sum_{j=1}^{d} |v_{i}^{(j)}|^{2} = 1, \, \sum_{i=m+1}^{n} \sum_{j=1}^{d} v_{i}^{(j)} \dot{T}_{i,j}^{(\ell)} = 0 \bigg\}
\ee
to be the collection of (normalized) directions that don't influence $T_{\ell}$ and keep the first $m$ data points fixed, and let
\be
W = W(X) \equiv \bigcap_{\ell=1}^{k} W_{\ell}(X).
\ee
Note that, by dimension-counting, we must have $\mathrm{Dim}(\mathrm{span}(W)) \geq d(n-m-k)$. Thus, as long as $m \leq n-k-1$, $W$ is non-empty. We also note that, for $X' \in \mathcal{G}$, $W(X')$ generates the tangent space of $\mathcal{G}$ at $X'$.\footnote{This is a small imprecision, since in principal we allow $\mathcal{G}$ to have self-intersections. We deal with this rigorously by only ever using the tangent space associated with a particular element of the covering in Assumption \ref{DefnInnercurvature}. Since the target audience of this paper may not have a lot of familiarity with differential geometry, we will continue to slightly abuse notation by \textit{e.g.} talking about ``the tangent space of $\mathcal{G}$ at $y$" as a shorthand for ``the tangent space of the element $U_{\alpha(y)}$ of the cover of $\mathcal{G}$ around the point $y$" when this does not cause mathematical confusion. \label{footnote_self_intersection}  }

Finally, for any set $S \subset \mathbb{R}^{\ell}$ and any $A > 0$, define
\be
S_{A} = \{ x \in \mathbb{R}^{\ell} \, : \, \inf_{s \in S} \|s - x\| \leq A\}
\ee
to be the $A$-thickening of $S$. With this notation, we state conditions that are often easier to work with than Assumption \ref{ass:LargeFluctuations}:

\begin{assumption} [Good Points] \label{AssumptionsGoodPoints}
Fix a collection
\be
\mathfrak A = \{\aI, \aIV, \aII, \aIII,\AI,\AIV,\AII,\AIII, c\}
\ee
of finite, positive constants satisfying $A_1,A_2,A_3,A_4 > 1$, $a_1,a_2,a_3,a_4,c \in (0,\infty)$, $x \in \tilde{\mc Z}$, and $v \in W$. With notation as above, say that $(x,v)$ is $\mathfrak{A}$-\textit{good} if
\begin{enumerate}
\item \textbf{$\mathcal{G}$ is often locally well-approximated by $W$:} $\mathcal{G}(x)$ has inner curvature at least $n^{-c}$.

\item \textbf{Posterior Sensitivity:}

There exists $ \mathcal{B}_{1} \subset \Theta$ so that
\be \label{IneqAssumptionDMax}
\sup_{\theta \in \mathcal{B}_{1}} \sup_{X \, : \, |X-x| \leq \AI^{-1}n^{-\aI}} D_{\max}(\theta,X) &\leq \AI n^{\aI} \\
\inf_{X \, : \, |X-x| \leq \AI^{-1}n^{-\aI}} p(\mathcal{B}_{1} | X, Y) &\geq 1 - \AIV n^{-\aIV}.
\ee

For $X$ satisfying $|X-x| \leq \AI^{-1}n^{-\aI}$, define the set of points that are not highly sensitive to perturbations:
\be \label{EqSomeDir}
\mathcal{B}_{1}(v,X) = \left\{ \theta \in \mathcal{B}_{1} \, : \, \left| \sum_{i,j} v_{i,j} D_{i,j}(\theta,X) \right| \in (\AII^{-1} n^{-\aII}, \AII n^{\aII})  \right\}.
\ee

This set is large in the following sense:
\be \label{IneqAssumptionSomeGoodDirection}
\lambda(\mathcal{B}^c_{1}(v,X)) \leq \AIII n^{-\aIII}.
\ee
\item \textbf{Non-super-concentration of Posterior:} For any subset $S \subset \mathcal{B}_{1}$ with $\lambda(S) \leq  \AIII n^{-\aIII} + \AIV n^{-\aIV}$, we have
\be \label{AssumptionNoSuperConcentration}
\sup_{X \, : \, |X-x| \leq  \in \AI^{-1}n^{-\aI}} \int_{\theta \in S} p( d \theta | x) < \frac{1}{2}.
\ee

\end{enumerate}
\end{assumption}

We give quick explanations of the purpose behind these assumptions:

\begin{enumerate}
\item When $\mathcal{G}$ has positive inner curvature, it is well-approximated by its tangent space.
\item Inequality \eqref{IneqAssumptionDMax} guarantees that the derivative $D_{i,j}$ doesn't change too quickly near $x$. Inequality \eqref{EqSomeDir} guarantees that, by moving the data $X$ a small distance $\epsilon$ in some direction $v \in W$, we can change $p(\theta | X,Y)$ by at least $\epsilon \Delta(\theta)$ outside of a small ``bad'' set $\mathcal B_1^c(v,x)$. Furthermore, the constant $\Delta(\theta)$ is at least polynomially small and at most polynomially large in $n$.
\item This guarantees that the small ``bad'' set $\mathcal{B}_{1}^c(v,X)$ allowed by the previous assumption is not close to covering the entire posterior mass.
\end{enumerate}

More briefly, these tell us that $\mathcal{G}(X)$ typically looks ``locally linear'' and that $p(\theta |X,Y)$  typically responds in a ``locally linear'' way to small changes in $X$.  We now check that Assumption \ref{AssumptionsGoodPoints} implies a large fluctuation condition:

\begin{lemma} \label{LemmaSuffCondDerivatives}
Fix $\gamma$ satisfying the first part of Assumption \ref{GoodGamma} and $Y$ so that $\gamma_{n,Y}$ satisfies the second part of Assumption \ref{GoodGamma}, and also $\tilde{\mc Z}$, $m$ and a collection of constants $\mathfrak{A}$. Assume that a random pair $X \sim \gamma_{n,Y}$ and $v \sim \mathrm{Unif}(\{w \in W \, : \, \| w \| = 1\})$ is $\mathfrak{A}$-good in the sense of Assumption \ref{AssumptionsGoodPoints} w.e.p.  Then

\be 
\{\| p(\cdot | \phi_{X,V}(t)) - \pi \|_{\TV} \geq \frac{s(n)}{40 A_{3}} n^{-a_{3}-2}\}
\ee 

holds w.e.p., where $s(n) = \Omega(n^{-C})$ for some $C = C(\mathcal{A}) > 0$ is given in the proof.
\end{lemma}

\begin{proof}

Throughout this proof, we view $Y$ as fixed; all calculations are done conditional on $Y$ when relevant. Next, note that we need only pay attention to points $x_{1},\ldots,x_{n}$ in the support of $\gamma$.

Fix some pair $(x,v)$ that are good in the sense of Assumption \ref{AssumptionsGoodPoints}. We will use the fact that, by the first part of the assumption, the tangent plane $W(x)$ to $\mathcal{G}$ shifts very little over neighborhoods of size $o(n^{-c})$.\footnote{As in footnote \ref{footnote_self_intersection}, we are again actually looking at the tangent plane of a part of the open cover (rather than of the manifold as a whole) in order to avoid self-intersections. We again leave this as only a footnote, since all our calculations are local and there is no possibility of confusion if you simply assume that self-intersections don't occur. } Let $R = R(n) = \frac{1}{4} \min(n^{-2c}, \frac{A_{3}(A_{1}+A_{3})}{A_{1}} n^{-a_1-a_3})$.

Fix $v \in W(x)$ and $0 \leq s \leq R(n)$. Thus, for $\theta \in \mathcal{B}_{1} \cap \mathcal{B}_{1}(v,x)$ and all $0 \leq s < R(n)$, we have by Inequalities \eqref{EqSomeDir} and \eqref{IneqAssumptionDMax}
\be
\left(A_3^{-1} n^{-a_3} - s  A_1 n^{a_1} \right)p(\theta | x) \leq \left|\frac{d}{dt} p(\theta \mid x+t v)\right|(0) \leq \left(A_3 n^{a_3} + s A_1 n^{a_1} \right)p(\theta | x);
\ee
since $0 \leq s \leq R(n) \leq \frac{1}{4} (A_1+A_3)^{-1} n^{-a_1-a_3}$,
\be
\left(\frac{3}{4} A_3^{-1} n^{-a_3}  \right)p(\theta | x) \leq \left|\frac{d}{dt} p(\theta \mid x+t v)\right|(0) \leq \left( \frac{4}{3} A_3 n^{a_3} \right)p(\theta | x).
\ee

Applying part (1) of Assumption \ref{AssumptionsGoodPoints}, recalling $0 \leq s \leq R(n) \leq \frac{1}{4} n^{-2c}$, and using \eqref{eq:PhiSmall}, this yields
\be
\left(\frac{1}{2} A_3^{-1} n^{-a_3}  \right)p(\theta | x) &\leq \left| \frac{d}{dt} p(\theta \mid \phi_{x,v}(t)) \right|(0) \leq \left( 2 A_3 n^{a_3} \right)p(\theta | x),
\ee
for all sufficiently large $n \in \mathbb{N}$. Note we could have chosen any constant $<3/4$ for the lower bound and $>4/3$ for the upper bound, and just selected $1/2$ and $2$ for convenience. Reparameterizing the previous expression, 
\be
\left(\frac{1}{2} A_3^{-1} n^{-a_3}  \right)p(\theta | \phi_{x,v}(s)) &\leq \left| \frac{d}{dt} p(\theta \mid \phi_{x,v}(t)) \right|(s) \leq \left( 2 A_3 n^{a_3} \right)p(\theta | \phi_{x,v}(s)),
\ee

Thus, for all $0 \leq s \leq R(n)$, we have by Gronwall's inequality
\be \label{IneqNoClosePostPointwise}
\log \left( \frac{p(\theta | \phi_{x,v}(s))}{p(\theta | \phi_{x,v}(0))} \right) = \log \left( \frac{p(\theta | \phi_{x,v}(s))}{p(\theta | x)} \right) \in \Theta^{+}(s) \cup \Theta^{-}(s),
\ee
where
\be
\Theta^{+}(s) &= \left\{ \theta \in \mathcal{B}_{1} \cap \mathcal{B}_{1}(v,x) \, : \,  \frac{1}{2 A_{3}} n^{-a_{3}} \leq s^{-1} \, \log \left( \frac{p(\theta | \phi_{x,v}(s)) }{p(\theta | x)} \right) \leq 2 A_{3} n^{a_{3}} \right\} \\
\Theta^{-}(s) &= \left\{ \theta \in \mathcal{B}_{1} \cap \mathcal{B}_{1}(v,x) \, : \,  \frac{1}{2 A_{3}} n^{-a_{3}} \leq -s^{-1} \, \log \left( \frac{p(\theta | \phi_{x,v}(s)) }{p(\theta | x)} \right) \leq 2 A_{3} n^{a_{3}} \right\}.
\ee
Applying Inequality \eqref{IneqNoClosePostPointwise}, we have
\be
\| p &(\cdot | \phi_{x,v}(s)) - p(\cdot | x) \|_{\mathrm{TV}} \geq \int_{\theta \in \Theta^{+}} | p(\theta | x) - p(\theta | \phi_{x,v}(s)) | d \theta \\
& \qquad + \int_{\theta \in \Theta^{-}} | p(\theta | x) - p(\theta | \phi_{x,v}(s)) | d \theta \\
&\geq \left( e^{\frac{s}{2 A_{3}} n^{-a_3}} - 1 \right) \int_{\theta \in \Theta^{+}} p(\theta | x) d \theta +  \left( 1-e^{-\frac{s}{2 A_{3}} n^{-a_{3}}}  \right) \int_{\theta \in \Theta^{-}} p(\theta | \phi_{x,v}(s)) d \theta \\
&\geq e^{\frac{s}{2 A_{3}} n^{-a_{3}}} \left(1  - e^{-\frac{s}{2 A_{3}} n^{-a_{3}}} \right) \int_{\theta \in \Theta^{+}} p(\theta | x) d \theta +  \left( 1-e^{-\frac{s}{2 A_{3}} n^{-a_{3}}}  \right) e^{-2 A_{3} s n^{a_{3}}}\int_{\theta \in \Theta^{+}} p(\theta | x) d \theta \\
&\geq \left( e^{-2 A_{3} s n^{a_{3}}} + e^{\frac{s}{2 A_{3}} n^{-a_{3}}} \right) \left(1  - e^{-\frac{s}{2 A_{3}} n^{-a_{3}}} \right) \int_{\theta \in \Theta^{+} \cup \Theta^{-}} p(\theta | X) d \theta \\
&\geq \frac{s}{2 A_{3}} n^{- a_3} \int_{\theta \in \mathcal{B}_{1} \cap \mathcal{B}_{1}(v,x)} p(\theta | X) d \theta \\
&\stackrel{\text{Ineq. } \eqref{AssumptionNoSuperConcentration}}{\geq} \frac{s}{5 A_{3}} n^{-a_{3}}.  
\ee
Observe, however, that there is nothing special about the point $x = \phi_{x,v}(0)$ on the curve $\{ \phi_{x,v}(s)\}$, except that it is in the middle of the range of values of $s$ for which $\phi_{x,v}$ is guaranteed to exist. Thus, by the same calculation, we can conclude that for all $0 \leq |s_{1}|, |s_{2}| \leq \frac{1}{2} R(n)$,
\be \label{IneqMainCalculationLargePerturbations}
\| p &(\cdot | \phi_{x,v}(s_{1}) - p(\cdot |  \phi_{x,v}(s_{2}))) \|_{\mathrm{TV}} \geq \frac{|s_{1}-s_{2}|}{10 A_{3}} n^{-a_{3}}.
\ee
This tells us that for our choice of  $(x,v)$, small perturbations of the data $\phi_{x,v}(s) \in \mathcal{G}$ will have quite different posterior distributions for \textit{most} small values of $s$.

We now check that, for \textit{any} distribution $\pi$, we have
\be
\| p(\cdot | X) - \pi \|_{\mathrm{TV}} \geq \frac{S(n)}{160 A_{3}} n^{-a_{3}-2}
\ee
holds w.e.p.

To do this, we will define for $0 \leq s \leq \frac{1}{2} R(n)$ a Metropolis-Hastings transition kernel $Q_s$ with stationary measure equal to $\gamma_{n,Y}$ conditioned on being in the set $\mathcal{G}(x) \cap \tilde{\mc Z}$, which we denote $\tilde{\gamma}_{n,Y}$.\footnote{This Metropolis-Hastings kernel is purely a proof device - it need not be computationally tractable.}. To sample from $Q_{S}(x_{0},\cdot)$,
\begin{enumerate}
\item Sample $V \sim \mathrm{Unif}(W(x_{0}))$ and $S \sim \mathrm{Unif}([-s,s])$.
\item Propose the point $X' = \phi_{x_0,V}(S)$
\item Accept or reject $X'$ using the usual Metropolis-Hastings rejection step.
\end{enumerate}

By Assumption \ref{GoodGamma}, $\gamma_{n,Y}$ changes little over small intervals and so the probability of rejection can be made arbitrarily small by choosing a small value of $s$. More precisely, it implies that there exists some $1 < \tilde{a} < \infty$ so that for $X \sim \gamma_{n,Y}$, the random measure $Q_{n^{-\tilde a}}(X,\cdot)$ has the following property:  w.e.p., either the proposed point $X'$ is accepted \textit{or} the proposed point is not in $\tilde{\mc Z}$.

Set $s = s(n) = \min(\frac{1}{2} R(n), n^{-\tilde{a}})$. Sample $X \sim \tilde{\gamma}_{n,Y}$, and let $V,S$ be the additional random variables sampled to construct $X'$ in the above algorithm for sampling from $Q_{S}(X,\cdot)$. We now consider two cases:

\begin{enumerate}
\item $\inf_{-s(n) \leq t \leq s(n)} \| p(\cdot | \phi_{X,V}(t)) - \pi \|_{\mathrm{TV}} \geq \frac{s(n)}{40 A_{3}} n^{-a_{3}-2}$. In particular, choosing $t=0$ gives $\ \| p(\cdot | X) - \pi \|_{\mathrm{TV}} \geq \frac{s(n)}{40 A_{3}} n^{-a_{3}-2}$ as well.

\item $\inf_{-s(n) \leq t \leq s(n)} \| p(\cdot | \phi_{X,V}(t)) - \pi \|_{\TV} \leq \frac{s(n)}{40 A_{3}} n^{-a_{3}-2}$. Although the infimum may not be achieved for any value of $t \in [-s(n),s(n)]$, there \textit{does} always exist some $s' = s'(X,V) \in [-s(n), s(n)]$ satisfying
\be
\| p(\cdot | \phi_{X,V}(s')) - \pi \|_{\TV} \leq 2 \inf_{-s(n) \leq t \leq s(n)} \| p(\cdot | \phi_{X,V}(t)) - \pi \|_{\TV}.
\ee  When $(X,V)$ is $\mathfrak{A}$-good, we have by Inequality \eqref{IneqMainCalculationLargePerturbations} that for all $0 \leq t \leq s(n)$,
\be
\| p(\cdot | \phi_{X,V}(t)) - \pi \|_{\TV} &\geq  \| p(\cdot | \phi_{X,V}(t)) - p(\cdot |  \phi_{X,V}(s')) \|_{\TV} - \| p(\cdot |  \phi_{Z,V}(s')) - \pi \|_{\TV}   \\
&\geq  \frac{|t-s'|}{10 A_{3}} n^{-a_{3}} - \frac{s(n)}{20 A_{3}} n^{-a_{3}-2}.
\ee

Since $S$ was chosen uniformly in $[-s(n),s(n)]$, we have $\P[|S-s'| \geq n^{-2} s(n)] \geq 1 - n^{-2}$, and so this implies

\be
\P\left[ \left\{ \| p(\cdot | \phi_{X,V}(t)) - \pi \|_{\TV} \leq\frac{s(n)}{20 A_{3}} n^{-a_{3}-2} \right\} \cap \left\{ (X,V) \text{ is } \mathfrak{A}\text{-good} \right\} \right] \leq n^{-2}.
\ee
\end{enumerate}

Thus, in both cases, $\{\| p(\cdot | \phi_{X,V}(t)) - \pi \|_{\TV} \geq \frac{s(n)}{40 A_{3}} n^{-a_{3}-2}\} \cup \{(X,V) \text{ is not } \mathfrak{A}\text{-good}\}$ holds w.e.p. Since $ \{(X,V) \text{ is } \mathfrak{A}\text{-good}\}$ holds w.e.p. by assumption, we conclude that $\{\| p(\cdot | \phi_{X,V}(t)) - \pi \|_{\TV} \geq \frac{s(n)}{40 A_{3}} n^{-a_{3}-2}\}$ holds w.e.p.
\end{proof}

\subsubsection{Application to Exponential Family Regression Models} \label{SubsecExpRegApp}

In this section, we apply our abstract results to exponential family regression models of the form

\be \label{EqGenExpFam}
p( (x,y) | \beta,\sigma) = \prod_{i=1}^{n} a(y_i,\sigma) e^{ \frac{b(x_i\beta) y_i - c(x_i\beta)}{d(\sigma)}}
\ee
where $y_{i} \in \mathcal Y$ is a sample space, $X_{i} \in \mathcal X \equiv \mathbb{R}^{d}$ is a sample space for the covariates, $(\sigma, \beta) \in \Theta = [0,\infty) \times \mathbb{R}^{d}$ is a parameter space, and $a \, : \, \mathcal Y \mapsto [0,\infty)$, $b \, : \mathbb R  \mapsto \mathbb{R}$, $c \, : \, \mathbb R \mapsto \mathbb{R}$ and $d \, : \, \mathbb R_+ \mapsto \mathbb{R}_+$ are further functions. Recall that the logistic regression model \eqref{EqLogRefDef} is a special case.

We now fix a prior $p$ and data-generating distribution $\gamma$ for the covariates, and introduce some basic assumptions about models in this class that will be used (along with moderate assumptions about the test statistics being used) to check Assumption \ref{AssumptionsGoodPoints}. To do so, we need some notation related to the collection of points for which the MLE is fixed.

Throughout, we will assume $\sigma$ is known. Then, the MLE for a GLM is defined by the equations:
\be \label{eq:LikelihoodEquations}
0 = \frac{\partial}{\partial \beta_{j}} \log p((x,y) | \beta,\sigma) =  \sum_{i=1}^{n} \frac{x_{ij}}{d(\sigma)} \left(\frac{\partial}{\partial \beta_j} b(x_{i} \beta) y_{i} - \frac{\partial}{\partial \beta_j} c(x_i \beta) \right).
\ee
Denote by $\hat \beta = \hat{\beta}(z_{0})$ the point that solves \eqref{eq:LikelihoodEquations} for the observed data $z_0 = (y_0,x_0)$. Let $f: \bb R^{np + p} \to \bb R^{p}$ be the function of $x,\beta$ defined by the set of partial derivatives on the right side in \eqref{eq:LikelihoodEquations}, so that the set of points where the MLE is equal to some particular $\beta$ is contained in 
\be \label{EqGLMJacSetDef}
S(\beta) = \{x : f(x, \beta) = 0\},
\ee
with $f_j(x,\beta) = [f(x,\beta)]_j$. Let
\be \label{EqGLMJacDef}
J(x, \beta) = \left[ \left( \frac{\partial f_j}{\partial \beta_k} \right) (x,\beta) \right]
\ee
be the Jacobian of the map $f$, with elements
\be \label{EqGLMJacExpanded}
J_{jk}(x,\beta) = \sum_{i=1}^{n} \frac{x_{ij} x_{ik}}{d(\sigma)} \left(\frac{\partial^2}{\partial \beta_j \partial \beta_k} b(x_{i} \beta) y_{i} - \frac{\partial^2}{\partial \beta_j \partial \beta_k} c(x_i \beta) \right).
\ee

Define the variables
\be \label{eq:SimpleCond}
\zeta_i = x_{i}' \beta, \quad i=1,\ldots,n.
\ee

Viewing the $y_{i}$ as fixed, \eqref{EqGLMJacExpanded} is of the form
\be \label{eq:SimpleJac}
J_{jk} \equiv  \sum_{i=1}^n x_{ij} x_{ik} F(\zeta_i).
\ee

For $1 \leq i \leq n$, we then define the derivative at $x_{1},\ldots,x_{n}$:
\be
V_{jk}[i] &= \left( \frac{\partial}{\partial \delta} J_{jk}(x_{1},\ldots,x_{i-1},x_{i} + \delta(1,1,\ldots,1), x_{i+1},\ldots,x_{n}, \beta) \right)(0) \\
&= x_{ik} F(\zeta_{i}) + x_{ij} F(\zeta_{i}) + x_{ij}x_{ik} F'(\zeta_{i}) \alpha,
\ee
where $\alpha = \sum_{j=1}^{p} \beta_{j}$ and $(1,1,\ldots,1)$ is a vector of ones.

Let $\mathcal{M} = \{(i,j) \in \{1,2,\ldots,p\}^{2} \, : \, i \leq j\}$ and $m = |\mathcal{M}|$. Define $W$ to be the by $n \times m$ matrix whose $i$'th row is $\{V_{\ell}[i]\}_{\ell \in \mathcal{M}} \equiv (W_{1i},\ldots,W_{mi}) \equiv W^{(i)}$.  For $\epsilon > 0$, $i \in \{1,2,\ldots,n\}$, vectors $x, \xi \in \mathbb{R}^{p}$ and $v \in \mathbb{R}^{m}$ with $\| \xi \| = \|v \| = 1$, define
\be
R_{v}^{(i)}(\epsilon; \xi, x) = \sum_{(j,k) \in \mathcal{M}} &[ (x_{k} + \epsilon \xi_{k}) F(\zeta + \epsilon \xi' \beta) + (x_{j} + \epsilon \xi_{j}) F(\zeta + \epsilon \xi'\beta)  \\
&+ (x_{j} + \epsilon \xi_{j})(x_{k} + \epsilon \xi_{k}) F'(\zeta + \epsilon \xi'\beta) \alpha] v_{jk}
\ee
and
\be
\hat{R}_{v}^{(i)}(\xi,x) = \left( \frac{d}{d \epsilon} R_{v}^{(i)}(\epsilon; \xi, x)\right) (0),
\ee
the directional derivative of $<W^{(i)},v>$ with respect to $x_{i}$ in the direction $\xi$.

\begin{assumptions} \label{AssumptionsGLMVersion}

We make the following assumptions about the GLM, data-generating distribution $\gamma$, and true parameter $\beta_{0}$:

\begin{enumerate}
\item \textbf{Good $\gamma$:} Assumption \ref{GoodGamma} holds. Furthermore, the MLE $\hat{\beta}_{n}$ based on $n$ datapoints satisfies
\be \label{IneqFastConvMLE}
\P[\| \hat{\beta}_{n} - \beta_{0} \| > \omega_{n} \textnormal{ infinitely often}] = 0
\ee
for some sequence $\lim_{n \rightarrow \infty} \omega_{n} =0$, where $\beta_{0}$ is the true parameter value. \footnote{Notice that for this condition to hold, it is enough for $\hat \beta_n \to \beta_0$ a.s.}

\item \textbf{Moderate Growth and Nonsingularity of Functions:} The functions $a,b,c,d$ are smooth. Furthermore, for all $i \in \{1,2,\ldots,n\}$,
\be  \label{IneqModGrowth1}
\sup_{\|x\| = \| \xi \| = 1} \inf_{\|v\| = 1} | \hat{R}_{v}^{(i)}(\xi,x)| > 0.
\ee

\item \textbf{Moderate Growth of Data:} Sample $X_{1},\ldots,X_{n} \stackrel{i.i.d.}{\sim} \gamma$ and $Y_{i} | X_{i} \sim p(\cdot | X_{i}, \beta_{0})$. We have\footnote{Note that, in this expression, $\{ \frac{\max_{1 \leq i \leq n}(|\log(\| (X_{i},Y_{i}) \|)|)}{\log(n)}\}_{n \in \mathbb{N}}$ is a (random) sequence indexed by $n \in \mathbb{N}$; the event that this sequence is $O(1)$ is exactly the event that this sequence is bounded. }
\be \label{IneqModGrowth2}
\P\left[ \frac{\max_{1 \leq i \leq n} |\log(\| (X_{i},Y_{i}) \|)|}{\log(n)} = O(1) \right] = 1.
\ee

\item  \textbf{Anticoncentration of Data:} For fixed $\theta$, define the collection of independent random variables $V_{i,j} = D_{i,j}(X,\theta)$ and define $\rho_{\max}(\theta)$ to be the maximum of the densities of these random variables. We assume that there exists $\mathcal{B}_{2} \subset \Theta$ and $c,C > 0$ satisfying
\be \label{IneqAntiConc}
\sup_{ \theta \in \mathcal{B}_{2}} \rho_{\max}(\theta) &= n^{O(1)}  \\
\lambda(\mathcal{B}_{2}) &\leq C n^{-c} \\
\beta_{0} &\in \mathcal{B}_{2}.
\ee
\end{enumerate}
\end{assumptions}

\begin{remark} [The (Not-So-)Difficult Condition] \label{RemSimpleAlgBadCond}
As we will see, most of these conditions are straightforward to verify either by inspection of the model or by simple assumptions on the tails of various distributions. The one exception seems to be Inequality \eqref{IneqModGrowth1} - while it appears to hold for a wide variety of models (and its analogues appear to hold for a wide variety of control variates), we don't know a way to verify it in any great generality.

While this makes e.g. Theorem \ref{ThmActualLogisticRegression} less general than it might otherwise be, there is a straightforward way to use a computer to obtain a certificate that this condition holds for a \textit{particular} model of interest. To see this, note that $\hat{R}_{v}^{(i)}(\xi,x)$ is an inner product. Thus, for \textit{any} choice $\xi = \xi^{(1)}, x = x^{(1)}$, the set 
\be 
V^{(1)} = \mathrm{Span}(\{ v \, : \, |\hat{R}_{v}^{(i)}(\xi^{(1)},x^{(1)})| = 0\})
\ee 
is an easily-computed hyperplane; furthermore, it is of codimension 1 \textit{unless} the \textit{coefficients} for this inner product \textit{all} vanish. Continuing to choose points $(\xi^{(2)}, x^{(2)}), \ldots$ and computing $V^{(2)}, V^{(3)}, \ldots$ the same way, we can verify Inequality \eqref{IneqModGrowth1} by checking that $\cap_{j=1}^{\ell} V^{(j)} = \{0\}$ for some (finite) $\ell$. Note that this intersection is almost-surely equal to $\{0\}$ if e.g. $\ell = m + 1$ and $v^{(j)}$ have some density with respect to Lebesgue measure, so in practice we expect that simply sampling such pairs uniformly at random will work in most situations.

If we wish to consider a data-generating process for the covariates that does not have full support on \textit{e.g.} a unit ball, we can run this algorithm while only choosing pairs $x^{(k)}, \xi^{(k)}$ such that $\{ x^{(k)} + s \xi^{(k)}\}_{s=0}^{S}$ is in the support of the data-generating process for sufficiently small $S > 0$.

\end{remark}

We now check Assumption \ref{AssumptionsGLMVersion} for several families of GLMs. For any GLM in canonical form, we always have $b(x) = x$. The following are some popular families of GLMs in canonical form:

\begin{defn} \label{DefCommonGLM} [Some Common GLMs]
\begin{enumerate}
 \item logistic regression: $a= 1$, $c = \log(1+e^{x_i\beta})$, $d = 1$.
 \item binomial logistic regression: $a = \binom{n}{y_{i}}$, $c = n \log(1+e^{x_i\beta})$, $d = 1$
 \item Poisson regression: $a = (y_i!)^{-1}$, $c = e^{x_i \beta}$, $d = 1$
\end{enumerate}
\end{defn}

\begin{lemma} \label{LemmaGLMsAreOK}
Assume that the prior $p$ has bounded density function and subexponential tails. Then the functions in Definition \ref{DefCommonGLM} satisfy all four conditions of Assumption \ref{AssumptionsGLMVersion} as long as they satisfy the first condition. Similarly, the second part of Assumption \ref{GoodGamma} is satisfied as long as $\gamma$ satisfies the first part.
\end{lemma}

\begin{proof}
We start with some preliminary calculations:

\be
D_{i,j}(\theta,X) &= \frac{1}{p(\theta | X)} \frac{\partial}{\partial x_{ij}} p(\theta | X) \\
&=\frac{1}{p(\theta) \prod_{i=1}^{n} a(y_i,\sigma) e^{ \frac{b(x_i\beta) y_i - c(x_i \beta)}{d(\sigma)}} } \frac{\partial}{\partial x_{ij}} \left( p(\theta) \prod_{i=1}^{n} a(y_i,\sigma) e^{ \frac{b(x_i\beta) y_i - c(x_i\beta)}{d(\sigma)}}  \right)\\
&= \frac{1}{d(\sigma)} \left\{ \frac{\partial}{\partial x_{ij}} b(x_i\beta)y_i - \frac{\partial}{\partial x_{ij}} c(x_i \beta) \right\} \\
&= \frac{\beta_j}{d(\sigma)} \left\{ b'(x_i\beta)y_i - c'(x_i \beta) \right\} \\
\ee

and

\be \label{EqCalcDMax}
D_{\max}(\theta,X) & = \max_{i,i',j,j'}  \frac{1}{p(\theta | X)} \left| \frac{\partial^{2}}{\partial x_{ij} \partial x_{i'j'}} p(\theta | X) \right| \\
&= \max_{i,i',j,j'} \frac{1}{d(\sigma)^2} \Bigg( \left\{ \frac{\partial}{\partial x_{ij}} b(x_i\beta)y_i - \frac{\partial}{\partial x_{ij}} c(x_i \beta) \right\} \left\{ \frac{\partial}{\partial x_{i'j'}} b(x_{i'} \beta) y_{i'} - \frac{\partial}{\partial x_{i'j'}} c(x_{i'} \beta) \right\}  \\
&\quad + \textbf{1}_{\{i = i'\}} \left\{ \frac{\partial^2}{\partial x_{ij} \partial x_{ij'}} b(x_i \beta) y_i - \frac{\partial^2}{\partial x_{ij} \partial x_{ij'}} c(x_i\beta) \right\}  \Bigg) \\
&= \max_{i,i',j,j'} \frac{\beta_j \beta_{j'}}{d(\sigma)^2} \Bigg( \left\{ b'(x_i\beta)y_i - c'(x_i \beta) \right\} \left\{ b'(x_{i'} \beta) y_{i'} - c'(x_{i'} \beta) \right\}  \\
&\quad + \textbf{1}_{\{i = i'\}} \left\{ b''(x_i \beta) y_i - c''(x_i\beta) \right\}  \Bigg).
\ee

We now check the parts of Assumption \ref{AssumptionsGLMVersion} in order:

\begin{enumerate}
\item This is assumed.
\item The first part of the condition is clear by inspection.  To verify Inequality \eqref{IneqModGrowth1}, we start by restricting our attention to vectors $\xi$ that satisfy $\xi'\beta = 0$. In that case,

\be
\hat{R}_{v}^{(i)}(\xi,x) = \sum_{(j,k) \in \mathcal{M}} &[(\xi_{j} + \xi_{k})F(\zeta) + (\xi_{j} x_{k} + \xi_{k} x_{j})\alpha F'(\zeta)]v_{jk}.
\ee

We observe that, allowing the directions of the vectors $x, \xi$ to vary over their allow ranges, this expression can only be identically 0 if every row of $v$ is proportional to the vector $\beta$. In that case, considering the vector $\xi \propto \beta$ immediately allows you to see that the first derivative is not also identically 0 in this case.

\item Inequality \eqref{IneqModGrowth2} follows almost immediately from the fact that the prior $p$, the covariate-generating distribution $\gamma$, and the model likelihoods all have subexponential tails.

More carefully: since $\gamma$ and $p$ have subexponential tails, there exists some $0 < A < \infty$ such that, for $\theta \sim p$ and $X_{1},\ldots,X_{n} \sim \gamma$,
\be \label{IneqModGrowthAMx}
\P[\max(|\theta|,\max_{1 \leq i \leq n} \|X_{i} \|) > A \log(n)] \leq n^{-3}.
\ee
In the two logistic regression cases, $Y_{i} \leq n$ deterministically. In the remaining case, define the random quantity $A_{\max} = \max(|\theta|,\max_{1 \leq i \leq n} \|X_{i} \|)$; typical concentration inequalities for subexponential random variables tell us $\P[\max_{1 \leq i \leq n} |Y_{i}| > (A_{\max}^{2}+1) \log(n)^{2}] \leq n^{-3}$ for all $n$ sufficiently large.\footnote{This bound is extremely conservative.}  Combining these two bounds on $\max_{1 \leq i \leq n} |Y_{i}|$ with Inequality \eqref{IneqModGrowthAMx} completes the proof.

\item We have explicit formulas for $D_{ij}$ and its arguments; the claim can be verified by seeing that none of these formulas have super-polynomially large values outside of $O(n)$ singularities.
\end{enumerate}

Finally, we verify the second part of Assumption \ref{GoodGamma}. In any generalized linear model in canonical form, the ratio of densities
\be
\frac{g_{n,y}(X)}{g_{n,y}(\tilde X )}
\ee
can be written in the form
\be
\frac{g_{n,y}(X)}{g_{n,y}(\tilde X )} = H(X, \tilde{X}) \,e^{-(X-\tilde X)'y}
\ee
for some function $H$ that does not depend on $y$. Thus, as long as there exists a constant $c$ such that $\| Y \|<n^c$ w.e.p, then we can choose $d_1,d_2$ such that the second part of Assumption \ref{GoodGamma} holds whenever the first part holds.
\end{proof}

Finally, we check that Assumption \ref{AssumptionsGLMVersion} implies  Assumptions \ref{AssumptionsGoodPoints} for reasonable test statistics:

\begin{lemma} \label{ThmRelateGeometricGLMAssumptions}
Consider a GLM model that satisfies Assumption \ref{AssumptionsGLMVersion}, along with test statistics $T_{1},\ldots,T_{d}$  as in the statement of Theorem \ref{ThmActualLogisticRegression}.  Then Assumption \ref{AssumptionsGoodPoints} is also satisfied.

\end{lemma}

\begin{proof}

We consider the three parts of Assumption \ref{AssumptionsGoodPoints}  in order; the first is by far the longest.

\textbf{First Condition:}
Let $S$ be as in Equation \eqref{EqGLMJacSetDef}, and let  $x \in S(\hat{\beta})$ be a point at which the Jacobian $J$ defined in Equation \eqref{EqGLMJacDef} is invertible. By the implicit function theorem, there exists an open set $U = U(x) \subset \bb R^{np}$ containing $x$ and a continuously differentiable function $g: U \to \bb R^{p}$ such that 
\be
g(x') = \hat \beta
\ee
and
\be
f(x',g(x')) = 0
\ee
for all $x' \in U$, and furthermore
with
\be
\left(\frac{\partial g}{\partial x_{ij}}\right) (x') = -J^{-1}(x',\hat \beta)  \left(\frac{\partial f}{\partial x_{ij}}\right) (x',g(x')).
\ee

We will need to verify that $U$ contains a ball of radius $n^{-O(1)}$ and the smallest singular value $\sigma(J)$ of $J$ is $n^{-O(1)}$, at least at most such points $x$. To do this, we will need a new bound on the smallest singular value of a random matrix that is stated and proved in Appendix \ref{SecAppSingVal}. In particular, we will verify the assumptions of Lemma \ref{BasicRMT} to bound the smallest singular value of $J$, then invoke the quantitative version of the implicit function theorem found in \cite{phien2012some}  to bound the size of $U$. We now set up the notation required to verify  the assumptions of Lemma \ref{BasicRMT}, following the notation for that lemma closely.

Inequality \eqref{IneqModGrowth1} (and the Assumption \ref{GoodGamma} that $\gamma$ has nonzero density on the unit ball) immediately implies that conditions \eqref{IneqNotIdZero} and \eqref{GoodSetNotEmpty} of Lemma \ref{BasicRMT} are satisfied. Condition \eqref{BoringUpperBoundRMT} of  Lemma \ref{BasicRMT} follows immediately from the assumption that the functions $a,b,c,d$ in the definition of the GLM are smooth and the support of the covariates is compact. Applying Lemma \ref{BasicRMT}, we conclude that there exists $\eta > 0$ so that
\be
\{ \sigma(J(X,\beta)) \geq \eta \}
\ee
holds w.e.p. for any fixed $\beta$, including $\beta_{0}$. Since $\| \hat{\beta} - \beta_{0} \| = o(1)$  w.e.p. by Inequality \eqref{IneqFastConvMLE}, and $J(\cdot,\beta)$ is a continuous function of $\beta$, this implies that there exists a constant $\eta'$ such that
\be
\{ \sigma(J(X,\hat{\beta})) \geq \eta' \}
\ee
holds w.e.p. Applying the main result of \cite{phien2012some} with the same bounds on the derivatives of $J$  gives the desired result.

\textbf{Second Condition:} We see that Inequality \eqref{IneqAssumptionDMax} holds w.e.p. whenever  $a,b,c,d$ are smooth, $\gamma$ has bounded support and Inequality  \eqref{IneqModGrowth2} holds. Inequalities \eqref{EqSomeDir} and \eqref{IneqAssumptionSomeGoodDirection} follow from an application of Lemma \ref{LemmaAntiConc}, with bound on $\rho_{\max}$ coming from Inequality \eqref{IneqAntiConc}.

\textbf{Third Condition:} We can see by inspection that the derivatives of $p$ are polynomially bounded in all arguments; the arguments themselves are polynomially bounded w.e.p. by the fact that $a,b,c,d$ are smooth, the fact that $\gamma$ has compacy support, and Inequality \eqref{IneqModGrowth2}.

\end{proof}

\subsection{Application to Logistic Regression} \label{SecProofMainThm}

\begin{theorem} \label{ThmActualLogisticRegression}

Under the conditions of Theorem \ref{ThmExampleTorpid}, the logistic regression model satisfies:
\begin{enumerate}
    \item Parts (a) and (b) of Assumption \ref{ass:LargeFluctuations} 
    \item Assumption \ref{ass:GoodSetHP},
    \item and Inequality \eqref{IneqWarmStartFake},
\end{enumerate}
all with constants $c_{i}, C_{i}$ that do not depend on $n$.
\end{theorem}

\begin{proof} 

By our assumption on $\gamma$ and the usual strong consistency theorem for the MLE \footnote{See \textit{e.g.} Theorem 2 of \cite{fahrmeir1985consistency}. Most of the conditions can be verified immediately, so we give short details only for one of them: Condition $(S_{\delta})$ clearly holds when all variables in the Fisher information matrix are replaced by their expectation; application of Hoeffding's inequality for i.i.d. bounded random variables then implies that $(S_{\delta})$ itself holds. }, the first part of Assumption \ref{AssumptionsGLMVersion} is satisfied. By Lemma \ref{LemmaGLMsAreOK}, the remaining parts of  Assumption \ref{AssumptionsGLMVersion}, all parts of Assumption \ref{GoodGamma}, and thus Assumption \ref{ass:GoodSetHP} are also satisfied. This means that, by Lemma \ref{ThmRelateGeometricGLMAssumptions}, Assumptions \ref{AssumptionsGoodPoints} are satisfied. 

The result then follows almost immediately: Inequality Inequality \eqref{IneqWarmStartFake} follows from Lemma \ref{LemmaWarmStartSuffCond} (as well as the usual Bernstein-von Mises theorem showing convergence of the rescaled posterior), and Parts (a) and (b) of Assumption \ref{ass:LargeFluctuations} follow from applying Lemma \ref{LemmaSuffCondDerivatives} (whose conditions are satisfied because Assumptions \ref{AssumptionsGoodPoints} are satisfied).

\end{proof}

\subsection{Verifying Other Control Variates}\label{RemVerifContVar}
We give a quick guide to verifying the conclusion of Theorem \ref{ThmActualLogisticRegression} for other control variates. Following our proof of that theorem, the same conclusion holds with $T_{1},\ldots,T_{d}$ augmented by any collection of control variates $T_{d+1},\ldots, T_{k}$, as long as the first condition in Assumption \ref{AssumptionsGoodPoints} is satisfied. Unfortunately, verifying this condition is the most difficult part of the proof of Theorem \ref{ThmActualLogisticRegression}, as it requires us to analyze the singular values of a random matrix with highly dependent entries. We mention two very different approaches to checking this result for new control variates:

\begin{enumerate}
\item In Remark \ref{RemSimpleAlgBadCond}, we isolate a single difficult-to-check condition that the proof hinges on, and give a simple algorithm that can be used to verify the condition for specific algorithms.

\item In Appendix \ref{AppSubsecGrid}, we define a special collection of control variates (see Equation \eqref{EqSimpleGridCont}) and show that they satisfy the first condition in Assumption \ref{AssumptionsGoodPoints}.

\end{enumerate}

We recall that adding additional control variates makes it harder to verify Assumption \ref{ass:LargeFluctuations}(a), but easier to verify the remaining assumptions of Theorem \ref{thm:SpecGapBound} with good constants. For this reason, having a collection of control variates for which Assumption \ref{ass:LargeFluctuations}(a) is easily satisfied can lead to substantially sharper bounds for specific algorithms.

\section{Subsampling Constructions} \label{SecAltCons}

In this section, fix notation as in Algorithm \ref{GenAusterity} and Proposition \ref{PropIneqChoiceOfSAust}.

We now give an alternative to Algorithm \ref{GenAusterity} which generates the same Markov chain but accesses datapoints in order. We first need a short technical lemma. 

Fix a subset $S = \{1,2,\ldots,n\}$ and integer $k \in \mathbb{N}$, and let $\ell = \lfloor \frac{n}{k} \rfloor$. Let $Y[1], Y[2],\ldots,Y[\ell]$ be defined recursively by sampling $Y[1] \sim C_{k}(S)$ and then according to the following rule:
\be
Y[i] \sim C_{k}(S\backslash \cup_{j=1}^{i-1} Y[j]).
\ee
Let $\tau$ be any stopping time for this sequence, let $\mu_{\tau}$ be the distribution of the sequence $Y \equiv (Y[1],Y[2],\ldots,Y[\tau])$, and for any $t$ let let $M(Y,t) = \cup_{i=1}^{t} \{Y[i]\}$. We have:

\begin{lemma} \label{LemmaSillyRepPerm}
Let $(Y_{1},\tau_{1}),(Y_{2},\tau_{2}),\ldots$ be an adapted sequence with $Y_{i} \sim \mu_{\tau_{j}}$. It is possible to couple this sequence to a permutation $\sigma \sim \mathrm{Unif}[S_{n}]$ such that
\be
\cup_{i=1}^{I} M(Y_{i},\tau_{i}) = \{1,2,\ldots, \max(\sigma(\cup_{i=1}^{I} M(Y_{i},\tau_{i})))\}
\ee
for all $I \in \mathbb{N}$.
\end{lemma}

\begin{proof}
The permutation $\sigma$ is obtained by just ordering the elements of the subsample as they arrive.

More formally, put each subsample $Y_{i}$ in an arbitrary order $Y_{i} = (Y_{i}^{(1)},\ldots,Y_{i}^{(\tau_{i})})$. Concatenate these to the sequence $(Z_{1},Z_{2},\ldots) = (Y_{1}^{(1)},\ldots,Y_{1}^{(\tau_{1})},Y_{2}^{(1)},\ldots)$.  We then define $\sigma$ by initially setting
\[
\sigma^{-1}(1) = Z_{1}
\]
and then inductively
\[
\sigma^{-1}(j+1) = Z_{\min \{ i \, : \,  Z_{i} \notin \{ \sigma^{-1}(1),\ldots,\sigma^{-1}(j) \}}.
\]
\end{proof}

We are now ready to prove Proposition \ref{PropIneqChoiceOfSAust}:

\begin{proof} [Proof of Proposition \ref{PropIneqChoiceOfSAust}]

We note that the observed minibatches in Algorithm \ref{GenAusterity} are of exactly the form studied in Lemma \ref{LemmaSillyRepPerm}.

Thus, by applying the random permutation $\sigma$ guaranteed to exist by Lemma \ref{LemmaSillyRepPerm} before running the algorithm, we can obtain a representation for which the datapoints in the interval $(n-m,n]$ are \textit{guaranteed} to be used after \textit{all} of the datapoints in $[1,n-m]$.

Thus, it remains only to bound the number of steps required to use $(n-m)$ out of $n$ datapoints. By standard concentration bounds, the average number of draws per step in any time interval of fixed length $T = \Theta(n\log(n))$ is $O\left(\exp(\lambda) \right)$, with high probability. Thus, it is sufficient to check that $\Omega(n\log(n))$ \textit{draws} are required with high probability. But this is exactly what is bounded by the usual coupon-collector bound on the time it takes to collect $(n-m)$ coupons when collecting $k$ at a time\footnote{This follows immediately from the classical coupon-collector bound on the time to collect $(n-m)$ coupons when collecting 1 at a time, as in \cite{ErRo61Coupon}.}.
\end{proof}

\section{Bounds on Singular Values of Structured Random Matrices} \label{SecAppSingVal}

In this section, we give a simple bound on the smallest singular value of a very special class of random matrices with highly dependent entries. These can be applied to show that the manifolds $\mathcal{G}(x)$ defined in \eqref{ManifoldDefinition} will typically have inner curvature that is not too close to 0. The variables defined in this section use notation that is independent of the rest of the paper, though we have tried to make it look similar.

For a square matrix $A$, we denote by $\sigma(A)$ its smallest singular value. For a $p$ by $q$ matrix $A$ with $p < q$ and $I \subset \{1,2,\ldots,q\}$, we denote by $A_{I}$ the submatrix of $A$ that keeps only the rows in the set $I$, and denote by $\sigma_{\mathrm{maximin}}(A) = \max_{I \subset \{1,2,\ldots,q\} \, : \, |I|=p} \sigma(A_{I})$ the \textit{largest} value of the \textit{smallest} singular values of all square $p \times p$ submatrices.

We denote by $M$ a $p$ by $p$ matrix, and we write each entry $M_{ij} = M_{ij}(x_{1},\ldots,x_{n}) \equiv \sum_{k=1}^{n} m_{ijk}(x_{k})$ as an additive function of points $x_{1},\ldots,x_{n} \in \mathbb{R}^{d}$. Next, define the set $\mathcal{M}$ to be either $\{1,2,\ldots,p\}^{2}$ when $M$ is not symmetric or $\{ (i,j) \in \{1,2,\ldots,p\}^{2} \, : \, i \leq j \}$ when $M$ is symmetric. In either case we let $m = | \mathcal{M}|$.

We then define derivatives: let
\be
W_{ij}^{(k)}(x_{1},\ldots,x_{n}) &= \left(\frac{d}{d \delta} M_{ij}(x_{1},\ldots,x_{k-1},x_{k} + \delta (1,1,\ldots,1), x_{k+1},\ldots, x_{n})\right)(0) \\
&=  \left(\frac{d}{d \delta} m_{ijk}(x_{k} + \delta (1,1,\ldots,1)) \right)(0);
\ee
note that, as $x_{k}$ is a vector, this last expression is generally \textit{not} equal to $m_{ijk}'(x_{k})$.

Finally,  for $s \in \mathbb{N}$, $\xi \in \mathbb{R}^{m}$ and $v \in \mathbb{R}^{m}$, define

\be
R_{v}^{(k)}(\xi,s; x) = \left( \frac{d^{s}}{d \delta^{s}} \sum_{(i,j) \in \mathcal{M}} W_{ij}^{(k)}(x_{1},\ldots,x_{k-1}, x_{k} + \delta \xi, x_{k+1},\ldots,x_{n}) v_{ij} \right)(0).
\ee

\begin{lemma} \label{BasicRMT}
Fix an interval $I \subset \mathbb{R}$. Assume that there exists some $\mathcal{H} \subset I^{p}$, an integer $1 \le S < \infty$ and $a > 0$ so that, for all $v \in \mathbb{R}^{m}$ and all $k \in \{1,2,\ldots,n\}$,
\be \label{IneqNotIdZero}
\sup_{\xi \in \mathbb{R}^{p}} \sup_{s \in \{1,\ldots,S\}} |R_{v}^{(k)}(\xi,s; x)| > a
\ee
at all points $x$ with $x_{k} \in \mathcal{H}$. Let $X_{1},\ldots,X_{n}$ be independent (but not necessarily identically distributed) random variables drawn from $I^{n}$, with densities bounded by $A^{-1}$ and $A$, and assume that
\be  \label{GoodSetNotEmpty}
\min_{i} \P[X_{i} \in \mathcal{H}] > 0.
\ee

Furthermore, assume 
\be \label{BoringUpperBoundRMT}
\sup_{x_{1},\ldots,x_{n} \in I} \max_{k,\ell,k',\ell',i,j} \left| \frac{\partial^{2}}{\partial x_{k \ell} \partial x_{k' \ell'}} M_{ij}(x_{1},\ldots,x_n) \right| < \infty.
\ee

Then there exists $c> 0$ depending only on $A,p,a,S,I,\mathcal{H}$ and the values of the expressions \eqref{GoodSetNotEmpty},\eqref{BoringUpperBoundRMT} so that the random matrix $M = M(X_{1},\ldots,X_{n})$ satisfies
\be
\{ \sigma(M) > c \}
\ee
w.e.p.

\end{lemma}

\begin{proof}

By Inequalities \eqref{IneqNotIdZero} and \eqref{GoodSetNotEmpty}, there exist constants $\eta, \epsilon > 0$ so that for all $v \in \mathbb{R}^{m}$ with $\|v \| = 1$ and all $k \in \{1,2,\ldots,n\}$,
\be
\P[| \sum_{\ell \in \mathcal{M}} W_{\ell}^{(k)} v_{\ell}| < \eta] < 1 - \epsilon.
\ee

Observing that the rows $W^{(1)},\ldots,W^{(n)}$ of $W$ are independent, this implies that for any $Q \subset \{1,2,\ldots,n\}$,
\be
\sup_{\| v \| = 1} \P\left[ \max_{i \in Q} | \langle W^{(i)},v \rangle| < \eta \right] \leq (1-\epsilon)^{|Q|}.
\ee
We can thus choose a sequence $i_{1} = 1, i_{2} \in \{\frac{n}{m}+1,\ldots,\frac{2n}{m}\}, \ldots,i_{m} \in \{\frac{n(m-1)}{m},\ldots,n\}$ so that, applying the previous inequality and taking a union bound,
\be
\P\left[ \min_{1 \leq q \leq p} d(W^{(i_{q})}, \Span(\{ W^{(i_{q'})} \}_{q' \neq q} )) < \eta \right] \leq m (1-\epsilon)^{\frac{n}{m}-1},
\ee
where $d(w,H)$ denotes the usual distance between a vector $w$ and hyperplane $H$.

Applying Lemma 1.11 of \cite{chafai2009singular}, this means
\be \label{IneqAuxBadSigma}
\P[\sigma_{\mathrm{maximin}}(W) < \eta] \leq m (1- \epsilon)^{\frac{n}{m}}.
\ee
On the event $\{\sigma_{\mathrm{maximin}}(W) < \eta\}$, for every vector $w$ of norm $\|w \| = 1$ we can find a solution to
\be
W \Delta = w
\ee
with norm
\be \label{IneqSmallNorm}
\| \Delta \| = O(m \eta^{-1}).
\ee

For $i \in \mathcal{M}$, denote by $w^{(i)}$ the $\{0,1\}$-valued vector with a single 1 at index $i$. For $\alpha > 0$, denote by $\mathcal{SC}(\alpha)$ the collection of probability measures $\pi$ on $\mathbb{R}$ with the property that $\pi(L) > \alpha$ for some interval $L$ of length $|L| < \alpha$.  

We will now apply Lemma \ref{LemmaSouravAntiCons}, stated in the section. We verify the assumptions in that lemma. The bound on $A$ is given by assumption, the  bound on $B$ is given by Inequality \eqref{IneqSmallNorm} applied to the vector $w^{(i)}$, the bound on $C$ given by Inequality \eqref{BoringUpperBoundRMT}, and the bound on the probability of $\mathcal{J}_{1}$ (and hence $\mathcal{J}_{2} = \mathcal{J}_{1}$) given by  Inequality \eqref{IneqAuxBadSigma}. By Lemma \ref{LemmaSouravAntiCons}, there exists some $\eta' > 0$ not depending on $n$ so that for all $1 \leq i \leq p$, the conditional distribution $ \mathcal{L}(M_{ii} | \{ M_{jk}\}_{(j,k) \neq (i,i)})$ satisfies event
\be \label{IneqNonConcCondJac}
\P[ \mathcal{L}(M_{ii} | \{ M_{jk}\}_{(j,k) \neq (i,i)}) \in \mathcal{SC}(\eta')] = e^{-\Omega(\frac{n}{m})},
\ee
where in this probability statement we view the randomness as coming from the collection $\{ M_{jk}\}_{(j,k) \neq (i,i)}$, which are functions of the random variables $X_1,\ldots,X_n$. \footnote{Since we are conditioning on $\{ M_{jk}\}_{(j,k) \neq (i,i)}$ here, this distribution is not directly parameterized by $x_{1},\ldots,x_{n}$ as in the statement of Lemma \ref{LemmaSouravAntiCons}. This can be fixed by locally reparameterizing $M_{j,k}$ to express it as a linear function on some small domain. As discussed in Remark \ref{RemImpFuncQuant} immediately following  Lemma \ref{LemmaSouravAntiCons}, by the quantitative version of the implicit function theorem given in \cite{phien2012some},  there exists such a reparameterization of the form required by Lemma \ref{LemmaSouravAntiCons} on a sufficiently large region of $I^{n}$; the above bounds on derivatives of $M_{jk} = M_{jk}(x_{1},\ldots,x_{n})$ are exactly those required by \cite{phien2012some}. }  Applying the main result of \cite{friedland2013simple}, this implies that there exists some $\eta'' > 0$ also not depending on $n$ so that

\be \label{IneqNonSingJac}
\P[\sigma(M) \geq \eta'' ] \geq 1 - e^{-\Omega(\frac{n}{m})},
\ee
completing the proof.\footnote{As stated, the main result of \cite{friedland2013simple} requires that the entry $M_{ii} \notin \mathcal{SC}(\eta')$ \textit{and} independent of the random variables $ \{ M_{jk}\}_{(j,k) \neq (i,i)}$. However, inspecting the proof, we see that it only uses the fact that the conditional distribution $\mathcal{L}(M_{ii} | \{ M_{jk}\}_{(j,k) \neq (i,i)}) \notin \mathcal{SC}(\eta')$. }

\end{proof}

\section{Anti-Concentration Bounds} \label{SecAppAntiConc}

We give several simple anti-concentration bounds.

\subsection{Anticoncentration for Sums}

We introduce one piece of notation: for $\epsilon > 0$, define $\mathcal{I}_{\epsilon}$ to be the collection of intervals of length less than $\epsilon$. We have:

\begin{lemma} \label{LemmaAntiConc}
Let $X_{1},\ldots,X_{m}$ be a sequence of independent (but not necessarily i.i.d.) random variables with densities uniformly bounded by $0 < \rho_{\max} < \infty$, and let $v \in \mathbb{R}^{d}$ a vector with norm $\| v \| = 1$. Then for any $\epsilon > 0$,
\be
\sup_{I \in \mathcal{I}_{\epsilon}}\P \left[\sum_{i=1}^{m} v_{i} X_{i} \in I\right] \leq \rho_{\max} \epsilon \sqrt{m}.
\ee
\end{lemma}

\begin{proof}
Since $\| v \| = 1$, there exists some $i \in \{1,2,\ldots,m\}$ satisfying $|v_{i}| \geq \frac{1}{\sqrt{m}}$. Assume WLOG that this $i=1$. Denote by $\mathcal{F}$ the $\sigma$-algebra generated by $X_{2},\ldots,X_{m}$. We then have
\be
\sup_{I \in \mathcal{I}_{\epsilon}} \P[\sum_{i=1}^{m} v_{i} X_{i} \in I] &=\sup_{I \in \mathcal{I}_{\epsilon}} \E[\P[\sum_{i=1}^{m} v_{i} X_{i} \in I | \mathcal{F}]] \\
&\leq \E[ \sup_{I \in \mathcal{I}_{\epsilon}} \P[\sum_{i=1}^{m} v_{i} X_{i} \in I | \mathcal{F}]] \\
&= \sup_{I \in \mathcal{I}_{\epsilon}} \P[v_{1} X_{1} \in I ] \\
&\leq \sup_{I \in \mathcal{I}_{\sqrt{m}\epsilon}} \P[X_{1} \in I ] \\
&\leq \rho_{\max} \epsilon \sqrt{m},
\ee
where the last line we use the fact that $X_{1}$ is a random variable whose PDF has density bounded by $\rho_{\max}$. This completes the proof.
\end{proof}

\subsection{Anticoncentration for Functions} \label{SecAppAntiFunc}

We give an immediate generalization of Lemma 1.2 \cite{chatterjee2019general}, then apply it to obtain a simple anticoncentration result for functions of many random variables. As Chaterjee writes of Lemma 1.2 of \cite{chatterjee2019general}, we suspect that both of these are fundamentally well-known results.

\begin{lemma} \label{IneqBasicallySourav}
Let $Y_{1},\ldots,Y_{k}$ be random variables with the same distribution $\mu$ on $\mathbb{R}$. Assume that
\be \label{IneqRateCloseAbst}
\P[\min_{1 \leq i < j \leq k} |Y_{i} - Y_{j}| \geq \epsilon] \geq 1- \delta
\ee
for some $\epsilon, \delta > 0$. Then for all $a\in \mathbb{R}$,
\be
\P[Y_{1} \in (a,a+\epsilon)] \leq \delta + \frac{1}{k}.
\ee
\end{lemma}

\begin{remark}
To use this theorem efficiently, the random variables $Y_{1},\ldots,Y_{k}$ should usually not be independent. We note that this result can be generalized to variables with slightly different distributions, as in \cite{chatterjee2019general}, but this is beyond the scope of the current article.
\end{remark}

\begin{proof}
The argument is essentially identical to Lemma 1.2 \cite{chatterjee2019general}: we simply note that \textit{either} the rare event whose probability is bounded in Inequality \eqref{IneqRateCloseAbst} occurs, \textit{or} at most one of $Y_{1},\ldots,Y_{k}$ is in the interval $(a,a+\epsilon)$.
\end{proof}

We next give a quick consequence that says, roughly: if $Y = F(X_{1},\ldots,X_{n})$ is a function of many independent random variables, \textit{and} the first derivatives of $F$ is not always very small, \textit{and} the second derivatives of $F$ are never very large, \textit{then} $Y$ cannot be concentrated on any small interval. The statement of the lemma is more complicated than this primarily because we wish to allow the bounds on the derivatives to fail on certain exceptional sets. 

\begin{lemma} \label{LemmaSouravAntiCons}
Fix constants $A,B,C,D > 0$ and $m,k \in \mathbb{N}$ and a $C^{2}$ function $F \, : \, [0,1]^{n} \mapsto \mathbb{R}$. Denote by $\mathcal{J}_{1} = \{ x \in [0,1]^{n} \, : \, \sup_{v \in \mathbb{R}^{n} \, : \, \| v \| = 1} |(\nabla_{v} F)(x) |  \geq B \}$ the set of points for which the derivative is large in some direction; for $x \in \mathcal{J}_{1}$, let $v = v(x)$ be the vector achieving the supremum in the condition. Next, let $\mathcal{J}_{2} = \{x \in \mathcal{J}_{1} \, : \, \sup_{y \, : \, \| x - y \| \leq D} \| H(F)(y)\|_{\infty} \leq C\}$, where $H(F)$ denotes the usual Hessian of $F$, the set of points for which the Hessian is not too large. Finally, let $X_{1},\ldots,X_{n}$ be a sequence of independent random variables with densities that are all bounded above by $A$. Then for all $\eta > 0$ and integers $k,m$ satisfying $2 km^{-1} \leq \min(\frac{B}{2AC}, \frac{1}{\eta},D)$,
\be \label{IneqPertBound}
\sup_{s \in \mathbb{R}} \P[|F(X_{1},\ldots,X_{n}) - s| <  \frac{B}{2Am}] &\leq \frac{1}{k+1} + \prod_{i=1}^{n}(1 - 2 \frac{k}{m}|v_{i}(X_{1},\ldots,X_{n})|) \\
& \quad  + \P[\max_{i} v_{i}(X_{1},\ldots,X_{n}) > \eta] + \P[(X_{1},\ldots,X_{n}) \notin \mathcal{J}_{2}].
\ee

\end{lemma}

\begin{remark}
We note that, if $A,B,C,D$ are all $n^{-O(1)}$, then we can choose $m,k$ so that the bound $\frac{B}{2Am}$ and the first three terms on the right-hand side of \eqref{IneqPertBound} are all $n^{-\Omega(1)}$. We will always use the theorem in this qualitative way, but keep the above quantitative version in case it is of independent interest.
\end{remark}

\begin{proof}
Write $X_{i} = f_{i}(U_{i})$, where $U_{1},\ldots,U_{n}$ are i.i.d. $\mathrm{Unif}[0,1]$ random variables and $f_{i}$ is the inverse of the CDF of $X_{i}$. In the following, we consider only the case that $(X_{1},\ldots,X_{n}) \in \mathcal{J}_{2}$ and $\max_{i} v_{i}(X_{1},\ldots,X_{n}) \leq \eta$; we will then try to obtain a bound of the form given by the first two terms in Inequality \eqref{IneqPertBound}.

For $a \in [-1,1]$ define $\tilde{U}_{i}[a] = U_{i} + a v_{i}$ modulo 1. Let $Y = F(X_{1},\ldots,X_{n})$ and $\tilde{Y}[a] = F(f_{1}(\tilde{U}_{1}[a]),\ldots,f_{n}(\tilde{U}_{n}[a]))$. Note that $\tilde{U}_{i}$ are still i.i.d. $\mathrm{Unif}[0,1]$ random variables, and so $\tilde{Y}[a]$ has the same distribution as $Y$ for every fixed $a$. Furthermore, $\tilde{Y}[0] = Y$.

Next, consider $0 \leq a < b$. As long as $0 \leq U_{i} + a v_{i} \leq 1$ for all $i \in \{1,2,\ldots,n\}$,
\be
| \tilde{Y}[a]  - \tilde{Y}[b] - (a-b) \nabla_{v} F(f_{1}(\tilde{U}_{1}[a]),\ldots,f_{n}(\tilde{U}_{n}[a])) | \leq C (a-b)^{2}.
\ee
When this occurs,
\be
| \tilde{Y}[a]  - \tilde{Y}[b] | &\geq (a-b) \frac{B}{A} - C (a-b)^{2}  \\
&\geq (\frac{B}{A} - C(a-b))(a-b),
\ee
and so for $(a-b) \leq \frac{B}{2AC}$,
\be
| \tilde{Y}[a]  - \tilde{Y}[b] | \geq \frac{B}{2A} (a-b).
\ee
Thus, for any integers $k,m$ satisfying $2 km^{-1} \leq \frac{B}{2AC}$, we can couple the $(k+1)$ random variables $\tilde{Y}[0], \tilde{Y}[(-k+1)m^{-1}],\ldots, \tilde{Y}[k m^{-1}]$ so that
\be
\P[\min_{|p|, |q| \leq k} | \tilde{Y}[p] - \tilde{Y}[q] | \leq \frac{B}{2Am}] \leq 1-\P[0 \leq U_{i} \pm \frac{k}{m} v_{i} \leq 1] \leq 1 - \prod_{i=1}^{n}(1 - 2 \frac{k}{m}|v_{i}|)
\ee
The result then follows from an application of Lemma \ref{IneqBasicallySourav}.

\end{proof}

\begin{remark} \label{RemImpFuncQuant}
We  combine this lemma with a quantitative versions of the implicit function theorem (\textit{e.g.} \cite{phien2012some}) to apply this bound for \textit{conditioned} random variables. More precisely, we consider the case where $Y = F(X_{1},\ldots,X_{n})$, and we wish to condition on the values of several other random variables $H_{i}(X_{1},\ldots,X_{n}) = c_{i}$ of the same form, for $i \in \{1,2,\ldots,\ell\}$. The implicit function theorem tells us that, under certain conditions, we can find a \textit{local} parameterization of the space $\{(x_{1},\ldots,x_{n}) \, : \,  H_{i}(x_{1},\ldots,x_{n}) = c_{i} \, \forall \, 1 \leq i \leq \ell \}$ of the form $\hat{F} \, : \, \mathbb{R}^{n-\ell} \mapsto \mathbb{R}$. When this is possible, we can express the conditioned random variable $\hat{Y}$ as a function  $\hat{F}(\hat{X}_{1},\ldots,\hat{X}_{n-\ell})$ for which Lemma \ref{LemmaSouravAntiCons} applies. Furthermore, the same bounds on derivatives that appear in the statement of the lemma are sufficient to apply \cite{phien2012some}.
\end{remark}

\section{Grid-Based Control Variates} \label{AppSubsecGrid}

Informally, we expect that adding ``extra" control variates will typically make the constants $c_{1},C_{1}$ appearing in Part (a) of Assumption \ref{ass:LargeFluctuations} worse, but will typically make the other constants appearing in our assumptions  better. In this section we introduce a collection of control variates that never substantially worsen the bounds on the constants appearing in Part (a) of Assumption  \ref{ass:LargeFluctuations} that were derived in Section \ref{SecGenSuffCondLarge}; we have found them to be extremely useful in improving the constants appearing in the remaining assumptions. They are especially important for verifying the ``warm start" condition given in Inequality \eqref{IneqCondWarmStart}, which generally does \textit{not} hold for geometrically ergodic chains and generic initial distributions. 

The construction is simple. Fix a constant $a  \in \mathbb{R} $, let $\Lambda_{n,a} = \{ (j_{1} n^{-a}, \ldots, j_{d} n^{-a})\}_{j_{1},\ldots,j_{d} \in \mathbb{Z}}$, and define the $i$'th control variate
\be \label{EqSimpleGridCont}
T_{i}(x) = \mathrm{argmin}_{y \in \Lambda_{n,a}} \| x - y \|,
\ee breaking ties arbitrarily. In other words, $T_{i}$ breaks up the state space of the data into polynomially-small boxes, and identifies which box a datapoint is in.

We observe that adding $T_{1},\ldots,T_{n}$ to a collection of control variates that satisfy Assumption \ref{DefnInnercurvature} with constant $c > 0$ will result in a new collection of control variates that satisfy Assumption \ref{DefnInnercurvature} with some constant that is strictly greater than 0, w.e.p. Since Assumption \ref{DefnInnercurvature} is the only place that control variates are used in Section \ref{SecGenSuffCondLarge}, the remaining arguments go through without modification. 

Having shown that these new control variates can easily be incorporated without \textit{hurting} the bounds in our earlier, we now give an informal sketch of a situation in which they might actually \textit{help}. 

Note that any $C^{1}$ control variate $T(x)$ can be estimated by an approximate version $\hat{T}$ according to the formula
\be
\hat{T}(x) = T(T_{1}(x_{1}), \ldots, T_{n}(x_{n})),
\ee
with error $\sup_{x} \| \hat{T}(x) - T(x) \| = O(n^{-a+1})$ when $T$ has uniformly bounded derivatives. This natural leads to an approximate subsampling kernel $\hat{K}$, obtained by replacing $T$ with $\hat{T}$ wherever it appears in the definition of $K$.\footnote{This ``definition" of $\hat{K}$ of course depends on the particular presentation of a subsampling algorithm. Naively applying this to \textit{all} precomputed quantities appearing in the published presentation seems to work well in many examples, including  Algorithm 7 of  \cite{bardenet2017markov} and the main algorithm of \cite{maclaurin2014firefly}. } When $K$ and $\hat{K}$ are sufficiently close, satisfying an inequality of the form
\be \label{IneqPertCheapFunc}
d(K,\hat{K}) \ll \min(f(K), f(\hat{K}))
\ee
for some notion of distance $d$ and mixing rate $f$, standard perturbation theory for Markov chains (see \textit{e.g.} \cite{alquier2016noisy}) implies that the two transition kernels must have spectral gaps $\lambda(K)$, $\lambda(\hat{K})$ satisfying
\be \label{IneqSmallPertSameGap}
\frac{\lambda(K)}{\lambda(\hat{K})} \in [0.5,2].
\ee

\textit{To conclude:} it is straightforward to verify that $T_{1},\ldots,T_{n}$ satisfy Assumption \ref{DefnInnercurvature} when using the arguments in Section  \ref{SecGenSuffCondLarge}. In some situations, it is straightforward to verify \eqref{IneqPertCheapFunc}. When \textit{both of these} hold, we may apply our main result, Theorem \ref{thm:SpecGapBound}, to bound the spectral gap of $\hat{K}$, and then Inequality \eqref{IneqSmallPertSameGap} to conclude that a very similar bound holds for $K$. The main advantage of this approach is that Assumption \ref{ass:LargeFluctuations} may hold for $\hat{K}$ with very large values of $s^{(n)}$, even when it holds for $K$ only with very small values of $s^{(n)}$.

This approach might feel like cheating - if $a$ is large, we are essentially using the original data! Indeed, using a large value of $a$ does cause problems - but these problems appear in the \textit{interpretation} of the final result. In applying Theorem \ref{thm:SpecGapBound} for exact chains (such as \cite{maclaurin2014firefly}), choosing even very large values of $a$ does not change the interpretation - these new control variates are purely a technical aid. In applying the result for approximate chains (such as many in \cite{bardenet2017markov}), large values of $a$ may change the constants $c_{1},C_{2}$ for which Part (a) of Assumption \ref{ass:LargeFluctuations} holds. In particular, if $a$ is too large, our main result just gives the (obvious) conclusion that the complete set of control variates essentially determines the posterior.

Thus, using this grid approximation trick requires us to choose a value of $a$ that is large enough to be useful, but small enough to not determine the posterior distribution. To give a very rough illustration: for approximately normal posterior distributions and algorithms that already use the posterior MLE and the Hessian at the MLE as control variates, we suspect that any choice of $a \leq 1$ will have a small impact on Part (a) of Assumption \ref{ass:LargeFluctuations} and any choice of $a > 1$ will have a substantial impact.

\end{appendix}

\end{document}